\documentclass[11pt]{article}

\usepackage[margin=1.02in]{geometry}
\usepackage{amsmath,amssymb,amsthm,mathtools,bbm}
\usepackage{microtype}
\usepackage{enumitem}
\usepackage{booktabs}
\usepackage{algorithm,algpseudocode}
\usepackage{xcolor}

\newcommand{\bs}{\sigma}
\newcommand{\A}{\mathcal{A}}
\usepackage{hyperref}
\usepackage{bookmark}
\usepackage[nameinlink,capitalize,noabbrev]{cleveref}
\usepackage{amsmath}

\DeclareMathOperator*{\argmin}{arg\,min}
\newcommand{\ip}[2]{\langle #1,#2\rangle}
\newcommand{\R}{\mathbb{R}}
\newcommand{\GOE}{\operatorname{GOE}}
\newcommand{\cN}{\mathcal{N}}
\newcommand{\F}{\mathrm{F}}
\hypersetup{
  colorlinks=true,
  linkcolor=black,
  citecolor=black,
  urlcolor=black,
  pdftitle={Stable and Online Algorithms for Random Matrix Discrepancy},
  pdfauthor={Eren C. Kizildag and Shuangping Li},
  pdfsubject={Gaussian matrix discrepancy}
}

\newtheorem{theorem}{Theorem}[section]
\newtheorem{proposition}[theorem]{Proposition}
\newtheorem{lemma}[theorem]{Lemma}
\newtheorem{corollary}[theorem]{Corollary}
\newtheorem{definition}[theorem]{Definition}
\newtheorem{remark}[theorem]{Remark}

\crefname{theorem}{theorem}{theorems}
\Crefname{theorem}{Theorem}{Theorems}
\crefname{proposition}{proposition}{propositions}
\Crefname{proposition}{Proposition}{Propositions}
\crefname{lemma}{lemma}{lemmas}
\Crefname{lemma}{Lemma}{Lemmas}
\crefname{corollary}{corollary}{corollaries}
\Crefname{corollary}{Corollary}{Corollaries}
\crefname{remark}{remark}{remarks}
\Crefname{remark}{Remark}{Remarks}
\newcommand{\distr}{\stackrel{d}{=}}

\newcommand{\cF}{\mathcal F}
\newcommand{\cS}{\mathcal S}

\newcommand{\E}{\mathbb E}
\newcommand{\Pp}{\mathbb P}

\newcommand{\dist}{d_{\mathrm H}}

\newcommand{\eps}{\varepsilon}

\newcommand{\svec}{\operatorname{svec}}
\newcommand{\mat}{\operatorname{mat}}
\newcommand{\Sym}{\operatorname{Sym}}
\newcommand{\bY}{\overline{Y}}
\newcommand{\ind}{\mathbbm{1}}

\definecolor{addedblue}{RGB}{0,0,0}
\newcommand{\newedit}[1]{{\color{addedblue}#1}}
\usepackage[T1]{fontenc}
\usepackage{tikz}
\usepackage{graphicx}
\usetikzlibrary{arrows.meta}
\makeatletter
\providecommand{\theHALG@line}{}
\renewcommand{\theHALG@line}{\thealgorithm.\arabic{ALG@line}}
\makeatother
\newcommand{\norm}[1]{\left\|#1\right\|}
\newcommand{\op}{\mathrm{op}}
\newcommand{\Tr}{\operatorname{Tr}}
\newcommand{\Var}{\operatorname{Var}}

\newcommand{\Round}{\operatorname{Round}}
\newcommand{\convp}{\xrightarrow{\Pp}}
\newcommand{\ellzero}{\ell_0}
\hypersetup{pdftitle={Stable and Online Algorithms for Random Matrix Discrepancy},
 pdfauthor={Eren C. Kizildag and Shuangping Li},
 pdfsubject={Algorithms and structural barriers for Gaussian matrix discrepancy}}
\title{\newedit{Stable and Online Algorithms for Random Matrix Discrepancy}}
\author{{\sf Eren C. K{\i}z{\i}lda\u{g}}\thanks{Department of Statistics, University of Illinois Urbana-Champaign; e-mail: {\tt kizildag@illinois.edu}.} 
\and
{\sf Shuangping Li}\thanks{Department of Statistics and Data Science, Yale University; e-mail: {\tt shuangping.li@yale.edu}.}}
\date{}
\begin{document}
\pagenumbering{gobble}
\maketitle
\begin{abstract}
\color{addedblue}
We study the average-case matrix discrepancy problem: given independent normalized $d\times d$ Gaussian orthogonal ensemble matrices $A_1,\dots,A_N$ and a fixed margin $\kappa>0$, find signs $\sigma_1,\dots,\sigma_N\in\{-1,1\}$ such that the operator norm of $\sum_{i=1}^N \sigma_i A_i$ is at most $\kappa\sqrt{N}$. Focusing on the proportional regime $N/d^2\to \tau\in(0,\infty)$ as $d\to\infty$ followed by the small-margin limit $\kappa\downarrow 0$, we characterize the density required by stable offline algorithms and by online algorithms.

In the offline setting, we construct a polynomial-time \emph{recenter-and-round} algorithm that is noise-stable and succeeds whenever $\tau=\Omega(\frac{1}{\kappa^2\log(1/\kappa)})$, along with a matching lower bound for all stable algorithms. In the online setting where each sign must be chosen irrevocably upon observing the corresponding matrix, we determine the exact limiting performance of the \emph{Frobenius-greedy} algorithm, establishing that it succeeds when $\tau>\tau_{\rm FG}(\kappa)\sim \frac{\pi}{4\kappa^2}$, as well as a matching lower bound for all online algorithms by conditioning on a revealed prefix. At the core of our algorithms lies rotational symmetry, which enables us to transfer Frobenius norm control into operator norm guarantees.

Together, our results identify the algorithmic phase transition points for random matrix discrepancy: $\Theta(\frac{1}{\kappa^2\log(1/\kappa)})$ for stable offline algorithms and $\Theta(\frac{1}{\kappa^2})$ for online algorithms. Both thresholds lie far above the satisfiability scale $\Theta(\log(1/\kappa))$, as shown by Maillard~\cite{maillard2025}.
\end{abstract}
\newpage
\tableofcontents
{\color{addedblue}
\newpage
\pagenumbering{arabic}
\section{Introduction}\label{sec:intro}
Discrepancy minimization is a central problem in combinatorics and theoretical computer science, with applications in convex geometry, numerical integration, design of randomized experiments, and rounding in optimization~\cite{matousek1999geometric,chazelle2000discrepancy,harshaw2024balancing}. It asks how to partition a collection
of objects into two groups that are as balanced as possible. In its classical vector formulation, given $v_1,\dots,v_N\in\R^m$
with $\|v_i\|_\infty\le 1$, one seeks signs
$\sigma_i\in\{-1,1\}$ attaining minimal discrepancy, 
\[
 \operatorname{disc}(v_1,\ldots,v_N)
 :=\min_{\sigma\in\{-1,1\}^N}
 \left\|\sum_{i=1}^N\sigma_i v_i\right\|_\infty.
\]
\paragraph{Prior work} Choosing the signs independently and uniformly at random yields an $O(\sqrt{N\log m})$ discrepancy, which becomes $O(\sqrt{N\log N})$ for $m=\Theta(N)$. Removing the $\sqrt{\log N}$ factor turned out to be a major challenge. Building on the partial coloring technique of Beck~\cite{beck1981roth}, Spencer’s celebrated \emph{`six standard deviations suffice’} theorem~\cite{spencer1985} improved this guarantee to $O(\sqrt{N\log(2m/N)})$ ($m\ge N$). For the square case, this is optimal up to constants as there are $N$ vectors in $[-1,1]^N$ with discrepancy $\Omega(\sqrt{N})$.\footnote{See, e.g.,~\cite[Section 1.5]{chazelle2000discrepancy} for a construction based on Hadamard matrices.}

Spencer’s original proof was non-constructive. Obtaining efficient (i.e., polynomial-time) algorithms for finding signings achieving small discrepancy required decades of further progress. In a 2010 breakthrough, Bansal~\cite{bansal2010constructive} gave the first polynomial-time construction attaining Spencer's bound up to constants for $m=\Theta(N)$, while Lovett and Meka~\cite{lovettMeka2015} subsequently developed a constructive partial coloring framework attaining the full Spencer bound, and Rothvo{\ss}~\cite{rothvoss2017convex}
extended this framework to general symmetric convex sets. The connection to convex geometry is also central to Banaszczyk's vector-balancing theorem~\cite{banaszczyk1998} and its algorithmic realization through the Gram--Schmidt walk~\cite{bansal2019gramschmidt}.

Another longstanding problem in discrepancy theory was the Koml\'os conjecture (recorded by Spencer~\cite{Spencer1987TenLectures}), which asserts that there exists a universal $C>0$
such that every collection $v_1,\dots,v_N\in\R^m$ with $\|v_i\|_2\le 1$ admits signs $\sigma_i\in\{-1,1\}$ satisfying $\|\sum_{i=1}^N \sigma_i v_i\|_\infty\le C$. Following recent progress by Bansal and
Jiang~\cite{bansal2025decoupling,bansal2026exposition}, Guo, Fang, and Lu~\cite{guo2026directional} resolved the conjecture; see also a simplified proof by Karingula and Lovett~\cite{karingula2026elementary}  and an extension to complex-valued signs ($\sigma_i\in\mathbb{C},|\sigma_i|=1$) by Guillen and Kobzar~\cite{guillenKobzar2026complex}. 

\paragraph{Matrix discrepancy} Matrix discrepancy is a natural generalization of Spencer’s vector balancing setup to matrices: given symmetric matrices $A_1,\dots,A_N\in\R^{d\times d}$ with $\|A_i\|_{\rm op}\le 1$, one seeks signs $\sigma_i\in\{-1,1\}$ attaining minimal discrepancy, defined by the operator norm of the signed sum:
\begin{equation}\label{eq:MatrixDiscrepancy}
     \operatorname{disc}_{\rm op}(A_1,\ldots,A_N)
 =\min_{\sigma\in\{-1,+1\}^N}
 \left\|\sum_{i=1}^N\sigma_i A_i\right\|_{\rm op}.
 \end{equation}
Taking $A_i=\operatorname{diag}(v_i)$ recovers Spencer's original setup. As in the vector discrepancy, choosing $\sigma_i$ uniformly at random yields an $O(\sqrt{N\log d})$ discrepancy by matrix concentration inequalities~\cite{ahlswede2002strong,tropp2012userfriendly}. The \emph{Matrix Spencer conjecture}, formulated in works of
Zouzias~\cite{zouzias2012matrix} and
Meka~\cite{meka2014discrepancy}, asks guarantees analogous to vector case—i.e., whether  one can improve upon the random colorings and subsequently obtain a $O(\sqrt{N\max\{1,\log(d/N)\}}$ bound. Namely, the Matrix Spencer conjecture is a generalization of Spencer’s theorem  from diagonal to arbitrary symmetric matrices. 

This extension introduces a substantial geometric obstacle. For a symmetric matrix $M$, $\|M\|_{\rm op} = \sup_{\|u\|_2 =1}|u^T Mu|$, 
so one must simultaneously control  an infinite family of quadratic forms in every direction.  Unlike diagonal matrices, general symmetric input matrices need not share an eigenbasis. Thus, the operator-norm constraint cannot
be reduced to balancing $d$ fixed coordinate directions in general. Further, even for a single Gaussian signing, the probability that all eigenvalues lie in a prescribed interval is governed by a spectral large-deviation problem—contrasting the Gaussian vector discrepancy with independent coordinates. Extending available methods (e.g., partial coloring) to this setting has therefore
motivated new geometric estimates for the associated discrepancy bodies; see~\cite{reisRothvoss2020,dadushJiangReis2022}. 

Despite substantial progress for structured inputs~\cite{levyRamadasRothvoss2017,bandeiraKuniskyMixonZeng2024,bandeiraBolcskei2026,akbasSra2026algebraic}, the general Matrix Spencer conjecture remained a major open problem for over a decade. In the square case ($d=N$), Hopkins, Raghavendra and Shetty~\cite{hopkinsRaghavendraShetty2022} resolved the conjecture under an additional assumption $\|A_i\|_F\le N^{1/4}$ using ideas from quantum communication complexity. Dadush, Jiang and Reis~\cite{dadushJiangReis2022} developed a mirror descent framework to obtain partial colorings, resolving the conjecture for block-diagonal matrices with block side $O(N/d)$, while Bansal, Jiang, and Meka~\cite{bansalJiangMeka2022} resolved the conjecture when each matrix has rank at most $O(N/\log^3 N)$ using refined matrix concentration inequalities~\cite{bandeiraBoedihardjoVanHandel2023}.

A series of breakthroughs beginning in August 2026 completely resolved the Matrix Spencer conjecture and brought substantial algorithmic advances. Akbas and Sra~\cite{akbasSra2026proof} gave a constructive proof in the square case; their exposition was subsequently simplified by Bansal and Kook~\cite{bansalKook2026simpler}. Soon after, Song and Zhang~\cite{songZhang2026matrix} obtained a randomized algorithm with almost-linear runtime in the dense input size measured in the real-arithmetic model. Kathuria~\cite{kathuria2026matrix} established the conjectured rectangular dependence with a randomized polynomial-time algorithm and, in a companion work~\cite{kathuria2026kadison}, gave a deterministic polynomial-time algorithm for Weaver's discrepancy theorem, both in the real-arithmetic model; see also the deterministic algorithm of Ezeunala and Jiang~\cite{ezeunalaJiang2026}. Originally introduced by Weaver~\cite{weaver2004}, the latter model was instrumental in the resolution of Kadison--Singer problem from functional analysis by
Marcus, Spielman, and Srivastava~\cite{marcusSpielmanSrivastava2015}; see also rank-one discrepancy bounds~\cite{kyngLuhSong2020} and worst-case hardness results~\cite{spielmanZhang2022weaver}. Even more recently, Akbas and Sra~\cite{akbasSra2026boolean}
developed small-ball inequalities for Boolean matrix series, recovering various high-profile problems including Kadison--Singer, Matrix Spencer, and Koml\'os conjectures as corollaries.

\paragraph{Average-case setting}
Another line of research studies  discrepancy minimization for random inputs rather than worst-case instances. Specifying an input distribution allows one to investigate typical optimum, phase transitions, and the limits of efficient algorithms in finer detail.

For independent standard normal vectors $v_1,\dots,v_N\sim\mathcal{N}(0,I_d)$, nonconstructive arguments yield that $\operatorname{disc}(v_1,\dots,v_N) =\Theta(\sqrt{N}\,2^{-N/d})$ with high probability (whp) whenever $d=o(N)$~\cite{turner2020balancing,costello2009balancing,karmarkar1986probabilistic}. By contrast, the best known polynomial-time algorithm achieves only $2^{-\Omega(\log^2 N/d)}$, provided $d$ grows sufficiently slow~\cite{turner2020balancing,yakir1996differencing,karmarkar1982differencing}. Substantial separation between these existential and algorithmic guarantees illustrates a conjectured \emph{statistical--computational gap}: finding near-optimal signings efficiently may be hard. For rigorous evidences supporting this conjecture, see~\cite{gamarnik2023algorithmic,vafa2025symmetric}. Discrete input distributions have also received considerable attention; e.g.,~\cite{altschuler2022discrepancy} obtained constant-discrepancy guarantees for Bernoulli/Poisson random vectors.

A closely related model is the \emph{symmetric binary perceptron (SBP)}~\cite{aubinPerkinsZdeborova2019} which asks for a signing with discrepancy at most $\kappa\sqrt{N}$ (for fixed $\kappa>0$) in the proportional regime $d=\Theta(N)$. For this model, prior work characterized the satisfiability threshold precisely and yielded sharp algorithmic lower bounds for restricted classes of algorithms. We discuss SBP and these guarantees below.

We pursue this program for random matrix discrepancy. In this setting, the objective is to minimize the discrepancy $\operatorname{disc}_{\rm op}(A_1,\ldots,A_N)$~\eqref{eq:MatrixDiscrepancy} of a collection of independent random matrices~\cite{kuniskyZhang2023}. We adopt the satisfiability version due to Maillard~\cite{maillard2025}: given independent, normalized Gaussian Orthogonal Ensemble (GOE) matrices in the proportional regime $N/d^2 \to \tau\in(0,\infty)$ as $d\to\infty$ and a fixed margin $\kappa>0$, find signs such that $\|\sum_{i=1}^N \sigma_iA_i\|_{\rm op}\le \kappa\sqrt{N}$. For any signing independent of the input, $N^{-1/2}\sum_{i=1}^N \sigma_i A_i$ is itself a normalized GOE matrix. Consequently, $N^{-1/2}\|\sum_{i=1}^N \sigma_i A_i\|_{\rm op}\to 2$ in probability as $d\to\infty$~\cite{anderson2010introduction}. We therefore focus on margins $\kappa<2$, which requires an improvement over random independent signings.

Motivated by a line of work on binary perceptron~\cite{gamarnikKizildagPerkinsXu2022,gamarnikKizildagPerkinsXu2023,liSchrammZhou2024}, we focus on the small margin limit, $\kappa\downarrow 0$ following $d\to\infty$. For all sufficiently small fixed $\kappa>0$, Maillard's results~\cite{maillard2025} imply that signings $\sigma\in\{-1,1\}^N$ with $\|\sum_i \sigma_i A_i\|_{\op}\le \kappa\sqrt{N}$ exist provided $\tau = \Omega_\kappa(\log\frac1\kappa)$. This leads to a central algorithmic question: 
\begin{center}
    \emph{For a prescribed margin $\kappa$, at what densities $\tau$ do polynomial-time algorithms succeed whp?}
    \end{center}
In the small-margin limit, we completely resolve this question for stable offline algorithms and online algorithms, by devising polynomial-time algorithms within these classes and proving lower bounds of matching order. Our results identify the location of algorithmic phase transition at $\Theta((\kappa^2\log(1/\kappa))^{-1})$ for stable offline algorithms and $\Theta(\kappa^{-2})$ for online algorithms, both of which lie far above the logarithmic satisfiability scale $\Theta(\log\frac1\kappa)$.
\paragraph{Model and satisfiability} We now discuss the model in detail. Let $A_1,\dots,A_N\in\R^{d\times d}$ be iid matrices drawn from the normalized Gaussian Orthogonal Ensemble, $A_i\sim \GOE(d)$. That is, for each $i$, $A_i$ has independent normal upper triangular entries: (a) $(A_i)_{k\ell}=(A_i)_{\ell k}\sim\cN(0,1/d)$ for $1\le k<\ell \le d$, and (b) $(A_i)_{kk}\sim\cN(0,2/d)$ for $1\le k\le d$. For a fixed margin $\kappa>0$, we seek 
signs $\sigma\in\{-1,1\}^N$ satisfying
\begin{equation}\label{eq:target} \left\|\sum_{i=1}^N\sigma_i A_i\right\|_{\rm op} \le \kappa\sqrt{N}. \end{equation}
For any $\sigma\in\{-1,1\}^N$ and $\kappa<2$, the event~\eqref{eq:target} has probability $\exp(-\Theta(d^2))$~\cite{maillard2025}, motivating the \emph{proportional regime} $N/d^2\to\tau\in(0,\infty)$ as $d\to\infty$, as the regime of interest. Below,  we refer $\tau$ as the \emph{density}. In what follows, we first fix $\kappa$ and $\tau$ and let $d\to\infty$; we then study the resulting density bounds as $\kappa\downarrow 0$. 

Maillard's satisfiability bounds~\cite{maillard2025} imply that, for every $\epsilon\in(0,1)$ and sufficiently small fixed $\kappa>0$,  there are no signings satisfying~\eqref{eq:target} when $\tau<(1-\epsilon)\tau_{\rm SAT}(\kappa)$, while they exist when $\tau>(1+\epsilon)\tau_{\rm SAT}(\kappa)$, both whp as $d\to\infty$, where $\tau_{\rm SAT}(\kappa)\sim \log(1/\kappa)/(2\log 2) = \Theta_\kappa(\log(1/\kappa))$.\footnote{This guarantee is asymptotic as $\kappa\downarrow 0$; for a fixed $\kappa>0$, exact location of the satisfiability threshold remains open~\cite{maillard2025}.} Thus, signings with discrepancy at most $\kappa\sqrt{N}$ exist as soon as $\tau$ grows only logarithmically in $1/\kappa$. This argument is nonconstructive; the main focus of our paper is on how much additional density is required to identify such signs in polynomial time.

\subsection{Our results}\label{sec:informal}
We study algorithmic phase transitions for random matrix discrepancy in both \emph{offline} and \emph{online} settings. In the offline setting, the algorithm may inspect the entire input before choosing the signs. Our particular focus is on \emph{stable} offline algorithms, for which tiny input perturbations lead to only small changes in the output (Definition~\ref{def:strong}). This class encompasses a broad range of algorithms for random computational problems; see below. In the online setting, which has been extensively studied in discrepancy theory, the input matrices arrive sequentially, and the algorithm must choose each sign irrevocably upon observing the corresponding matrix.

We now informally state our main contributions. We say that an algorithm \emph{succeeds} if it produces signs satisfying~\eqref{eq:target} whp over its input and internal randomness as $d\to\infty$. 
\begin{samepage}
\begin{theorem}[Offline stable algorithms, informal]\label{thm:informal-stable}
There exist universal constants $0<c<C<\infty$ such that, for every sufficiently small fixed $\kappa>0$:
\begin{enumerate}[label=(\roman*),leftmargin=*]
\item If $\tau\ge C/[\kappa^2\log(1/\kappa)]$, the polynomial-time recenter-and-round algorithm (Algorithm~\ref{alg:recenter}) succeeds and is stable in the sense of Definition~\ref{def:strong}.
\item If $\tau\le c/[\kappa^2\log(1/\kappa)]$, no stable algorithm succeeds, regardless of its running time.
\end{enumerate}
\end{theorem}
\end{samepage}
Starting from independent random signs, Algorithm~\ref{alg:recenter} iteratively solves a strongly convex quadratic program over $[-1,1]^N$ that balances the Euclidean residual against the expected number of sign flips, then applies independent, mean-preserving rounding to obtain binary signs. The quadratic regularization term makes the objective function strongly convex, ensuring a unique minimizer and enabling the noise-stability analysis. Our analysis exploits a key geometric property of independent Gaussian columns: with high probability, every direction admits sufficiently many columns with large positive projections. Proved using learning-theoretic tools (VC inequality and Sauer’s lemma), this property enables fractional cancellations with a small number of flips, and obtaining a contraction estimate for the residual.
See Section~\ref{sec:offline-overview} for details.
\begin{figure}[H]
\centering
\begingroup
\color{black}
\definecolor{phaseHard}{RGB}{219,105,75}
\definecolor{phaseSuccess}{RGB}{33, 158, 33}
\definecolor{phaseGap}{RGB}{226,229,232}
\begin{tikzpicture}[x=1.06cm,y=0.82cm,font=\small,text=black,
 guide/.style={densely dashed,draw=black!60,line width=0.5pt},
 rowname/.style={anchor=east,font=\small\bfseries},
 bartext/.style={align=center,inner sep=1pt},
 scaletext/.style={anchor=north,align=center,inner sep=1pt}]
\draw[guide] (1.6,4.50) -- (1.6,0.14);
\draw[guide] (6.8,3.30) -- (6.8,0.14);
\draw[guide] (9,2.10) -- (9,0.14);
\node[rowname] at (-0.22,4) {Satisfiability};
\path[fill=phaseHard!33] (0,3.50) rectangle (1.48,4.50);
\path[fill=phaseSuccess!30] (1.72,3.50) rectangle (12,4.50);
\draw[black!35] (0,3.50) rectangle (12,4.50);
\draw[guide] (1.6,3.50) -- (1.6,4.50);
\node[bartext] at (0.74,4) {UNSAT};
\node[bartext] at (6.86,4) {SAT};
\node[rowname] at (-0.22,2.8) {Stable};
\path[fill=phaseHard!33] (0,2.30) rectangle (6.68,3.30);
\path[fill=phaseSuccess!30] (6.92,2.30) rectangle (12,3.30);
\draw[black!35] (0,2.30) rectangle (12,3.30);
\draw[guide] (6.8,2.30) -- (6.8,3.30);
\node[bartext] at (3.34,2.8) {Hard for stable algorithms};
\node[bartext] at (9.46,2.8) {Stable algorithms succeed};
\node[rowname] at (-0.22,1.6) {Online};
\path[fill=phaseHard!33] (0,1.10) rectangle (8.88,2.10);
\path[fill=phaseSuccess!30] (9.12,1.10) rectangle (12,2.10);
\draw[black!35] (0,1.10) rectangle (12,2.10);
\draw[guide] (9,1.10) -- (9,2.10);
\node[bartext] at (4.44,1.6) {Hard for all online algorithms};
\node[bartext] at (10.56,1.6) {Greedy succeeds};
\draw[->,line width=0.65pt] (0,0.45) -- (12.3,0.45);
\node[anchor=west,inner sep=1pt] at (12.34,0.45) {$\tau$};
\node[scaletext] at (1.6,0.05) {$\Theta\!\left(\log\frac1\kappa\right)$};
\node[scaletext] at (6.8,0.05) {$\Theta\!\left(\frac1{\kappa^2\log(1/\kappa)}\right)$};
\node[scaletext] at (9,0.05) {$\Theta(\kappa^{-2})$};
\end{tikzpicture}
\endgroup
\caption{Algorithmic phase transitions for stable offline and online algorithms.}\label{fig:phases}
\end{figure}
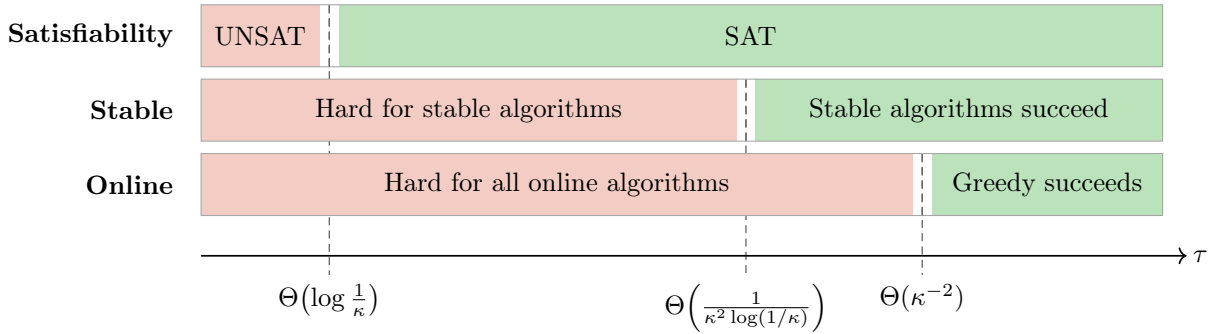

\begin{samepage}
\begin{theorem}[Online algorithms, informal]\label{thm:informal-online}
There exists a universal constant $c>0$ such that,
for every sufficiently small fixed $\kappa>0$:
\begin{enumerate}[label=(\roman*),leftmargin=*]
\item The Frobenius-greedy algorithm
(Algorithm~\ref{alg:Frobenius-greedy})
runs in $O(Nd^2)$ time, linear in the size of its input $(A_1,\dots,A_N)$, and succeeds whenever
$\tau>\tau_{\rm FG}(\kappa)\sim \pi/(4\kappa^2)$
as $\kappa\downarrow0$.
\item If $\tau\le c/\kappa^2$, every randomized online algorithm has exponentially small success probability.
\end{enumerate}
In fact, we determine the exact limiting performance of Frobenius-greedy for every fixed $\tau>0$.
\end{theorem}
\end{samepage}
Our online algorithm is a greedy procedure, using the Euclidean sign rule also studied in~\cite{czervinskiEtAl2026greedy}. Suppose that the first $t-1$ signs yield the partial sum $S_{t-1}:=\sum_{j=1}^{t-1}\sigma_j A_j$. Upon observing $A_t$, the algorithm chooses $\sigma_t\in\{-1,1\}$ to minimize the Frobenius norm $\|S_{t-1}+\sigma_tA_t\|_F$. Using fluid-limit (i.e., differential equations method) and martingale arguments, we explicitly characterize the evolution of this matrix. We then exploit rotational symmetry to convert the Frobenius-norm bounds into operator norm guarantees. See Section~\ref{sec:online-overview} for details. 

Consequently, in the density window $1/(\kappa^2\log\frac1\kappa)\ll \tau \ll 1/\kappa^2$, our stable algorithm succeeds, while any online algorithm has exponentially small success probability. This behavior mirrors symmetric binary perceptron, see below. 
\paragraph{Comparison with symmetric binary perceptron (SBP)} Our results reveal a close correspondence between the small-margin algorithmic landscape of random matrix discrepancy and that of the SBP. Given independent $v_1,\dots,v_N\sim\cN(0,I_d)$ with $d/N\to\alpha\in(0,\infty)$, the SBP seeks a signing $\sigma\in\{-1,1\}^N$ satisfying $\|\sum_{i=1}^N\sigma_i v_i\|_\infty\le\kappa\sqrt{N}$. Its satisfiability threshold is known exactly for every fixed $\kappa>0$~\cite{perkinsXu2024frozen,abbeLiSly2021contiguity}: $\alpha_{\rm SAT}(\kappa) =-\log 2/\log\mathbb{P}(|\cN(0,1)|\le\kappa) \sim\log2/\log(1/\kappa)$ as $\kappa\downarrow 0$. In the small margin regime, overlap-gap arguments rule out stable algorithms at densities $\alpha>C\kappa^2\log\frac1\kappa$~\cite{gamarnikKizildagPerkinsXu2022} and online algorithms at $\alpha>C'\kappa^2$~\cite{gamarnikKizildagPerkinsXu2023} for some absolute $C,C'>0$. Up to constants, both of these scales are sharp, with matching stable and online algorithms~\cite{huangSellkeSun2026projection,wengiel2026online,bansalSpencer2020}. The margin limit was also important in pinning down the algorithmic landscape for asymmetric binary perceptron~\cite{liSchrammZhou2024} (see also related works~\cite{joLee2026online,montanariZhongZhou2024tractability}). The structure of the solution space in binary perceptron models and its relationship to algorithmic tractability have also been extensively studied~\cite{baldassi2015subdominant,perkinsXu2024frozen,abbeLiSly2021contiguity,abbeLiSly2022binary,barbierElAlaouiKrzakalaZdeborova2024,gongHuangLiSellke2026stable}.

To compare these results with ours, recall that the space of $d\times d$ symmetric matrices has dimension $D=d(d+1)/2$. The corresponding dimension-to-sign ratio is therefore $D/N\to(2\tau)^{-1}$. Under this convention, our stable offline and online transition scales (Theorems~\ref{thm:informal-stable} and~\ref{thm:informal-online}) respectively become $\Theta(\kappa^2\log(1/\kappa))$ and $\Theta(\kappa^2)$, thereby matching the corresponding algorithmic thresholds of the SBP. Furthermore, Maillard's satisfiability scale~\cite{maillard2025} becomes $\log 2/\log(1/\kappa)$, agreeing with the SBP even at the level of the leading constant. Despite replacing coordinate-wise balance by a noncommutative spectral constraint, random matrix discrepancy therefore exhibits the same order and hierarchy of the small margin algorithmic scales as the SBP.

More recently, Huang, Sellke and Sun~\cite{huangSellkeSun2026projection} studied the empirical projection distributions attainable by Lipschitz algorithms through stochastic control, and obtained exact variational characterizations of algorithmic thresholds for various perceptron models with bounded continuous activations; see also the earlier stochastic-control characterization by Montanari and Zhou~\cite{montanariZhou2024projections} and Sun’s ICM survey~\cite{sun2026randomized}. For the SBP, their arguments settle the threshold for stable Lipschitz algorithms at every fixed $\kappa>0$, suggesting a natural direction for matrix discrepancy: characterizing its algorithmic thresholds at fixed positive $\kappa>0$ beyond the small margin asymptotics. 
 
\paragraph{Related work}
The study of random matrix discrepancy was initiated by Kunisky and Zhang~\cite{kuniskyZhang2023}, who obtained upper and lower bounds on the discrepancy~\eqref{eq:MatrixDiscrepancy}, and analyzed the matrix hyperbolic cosine (MHC) algorithm of Zouzias~\cite{zouzias2012matrix}. Their bounds identify the order of the typical discrepancy for normalized GOE inputs in the proportional regime $N=\Theta(d^2)$. Kunisky, Oertel, Wengiel, and Zhang~\cite{kuniskyOertelWengielZhang2025} subsequently established sharp discrepancy bounds when $d^2=o(N/\log N)$ and improved the analysis of the MHC algorithm. Maillard~\cite{maillard2025} studied the satisfiability variant analogous to the SBP, obtaining lower and upper bounds on the satisfiability threshold in the proportional regime that agree to leading order as $\kappa\downarrow 0$. Random rank-one inputs were studied in~\cite{caiChenZhang2025weaver}.

On the algorithmic side, the improved MHC analysis yields a terminal discrepancy bound of order $O(d\log d)$ for normalized GOE inputs, which is $O(\sqrt{N}\log N)$ when $N=\Theta(d^2)$. This estimate alone does not certify a fixed-margin guarantee $\kappa\sqrt{N}$ at any fixed proportional density. Very recently, Chen and El Alaoui~\cite{chenElAlaoui2026} developed an approximate message passing (AMP) algorithm for random matrix discrepancy, together with a novel free-probabilistic state evolution describing its limiting iterates through a correlated Gaussian semicircular process. For $N/d^2\to \tau$, their algorithm achieves a discrepancy of $g_{\rm AMP}(\tau)\sqrt{N}$ whp for an explicit $g_{\rm AMP}(\tau)\in(0,2)$ and small $\tau$, $0<\tau<\tau^*\approx 2.87$. In particular, their algorithm yields signs satisfying~\eqref{eq:target} for margins $\kappa\in(g_{\rm AMP}(\tau),2)$. Our work complements these results by addressing the small-margin limit, where $\kappa\downarrow0$ and the density $\tau$ grows: we pinpoint the density required
by stable offline and online algorithms up to constants, through polynomial-time algorithms and matching lower bounds.

\paragraph{Online discrepancy} Online discrepancy, where the arriving vectors must be assigned signs irrevocably, has a long history dating back to Spencer's balancing games~\cite{spencer1977balancing}. For vectors with Euclidean norm at most one, an adaptive adversary can force $\Omega(\sqrt{N})$ discrepancy even in dimension two. Against an oblivious adversary, which fixes the input sequence in advance, substantially stronger guarantees are possible. Alweiss, Liu, and Sawhney~\cite{alweiss2021discrepancy} developed a self-balancing walk that maintains $O(\log(dN))$ discrepancy for all signed prefix sums whp, provided $\|v_i\|_2\le 1$. Kulkarni, Reis, and Rothvo{\ss}~\cite{kulkarni2024optimal} established the existence of an online algorithm attaining
the optimal $O(\sqrt{\log N})$ guarantee; see also an improved running time guarantee by Aden-Ali~\cite{aden2026optimal}. 

For the random arrivals,  Bansal and Spencer~\cite{bansalSpencer2020} gave an online algorithm achieving $O(\sqrt{d})$ discrepancy whp for independent uniform vectors in $\{-1,1\}^d$, with a bound independent of the time horizon;
see~\cite{wengiel2026online} for an extension to sub-Gaussian inputs. For Gaussian arrivals with $N \asymp d$, \cite{fiedler2026meanfield} characterized the
optimal asymptotic constant through a mean-field stochastic control problem. More general stochastic arrivals were studied in~\cite{bansalJiangSinglaSinha2020,bansal2021online}; the latter obtained $O(\log^4(dN))$ discrepancy for independent samples from an arbitrary distribution supported on the Euclidean unit ball. For a more recent work on online vector balancing with iid binary sparse vectors, see~\cite{altschuler2026threshold}. 

Contrasting online vector balancing, the history of online matrix discrepancy is quite recent. Zouzias~\cite{zouzias2012matrix} introduced the matrix hyperbolic cosine algorithm, whose average-case performance was subsequently analyzed in~\cite{kuniskyZhang2023,kuniskyOertelWengielZhang2025}. To the best of our knowledge, our results yield the first sharp characterization of an online algorithmic thresholds for the matrix discrepancy problem under random arrivals.
\paragraph{Overlap geometry and algorithmic hardness} Our algorithmic lower bounds are based on the  overlap gap property (OGP) framework, pioneered by Gamarnik and Sudan~\cite{gamarnik2014limits,gamarnikSudan2017limits}; see~\cite{gamarnik2021overlap,gamarnik2022disordered,gamarnik2025turing} for surveys. Pairwise OGP forbids intermediate overlaps between pairs of near-optimal solutions, while ensemble or coupled versions impose overlap restrictions on solutions to correlated
instances~\cite{chenGamarnikPanchenkoRahman2019suboptimality,gamarnikJagannathWein2024hardness}. Multi-OGP extends this viewpoint to forbidden overlap
patterns among larger collections of
solutions~\cite{rahmanVirag2017half,gamarnikSudan2017sequential,wein2021optimal,breslerHuang2021phase,gamarnik2023algorithmic}, and branching OGP excludes hierarchical overlap patterns~\cite{huangSellke2025lipschitz,huang2023algorithmic}. Coupling and interpolation arguments convert these geometric obstructions into lower bounds for local, low-degree, AMP, and Lipschitz algorithms by exploiting
suitable notions of input
stability~\cite{gamarnikSudan2017limits,gamarnikJagannathWein2024hardness,gamarnikJagannath2021amp,huangSellke2025lipschitz}. More recently, related geometric obstructions have also been developed
for online algorithms~\cite{gamarnik2025optimal,dhawan2025sharp,dhawan2026algorithmic,du2025algorithmic,bhamidi2026finding,gamarnik2026optimal}. Our lower bound uses an ensemble multi-OGP for feasible signings, following the
perceptron arguments~\cite{gamarnikKizildagPerkinsXu2022,gamarnikKizildagPerkinsXu2023,liSchrammZhou2024}.

\section{Main Results and Technical Overview}\label{sec:overview}
We next formally define the offline and online algorithms, along with rigorous performance guarantees, lower bounds, and an overview of our technical arguments. See Sections~\ref{sec:offline-proofs} and~\ref{sec:online-proofs} for complete proofs.
\subsection{Offline algorithm: construction and stability}\label{sec:offline-overview}
Dubbed \emph{recenter-and-round}, our offline algorithm (Algorithm~\ref{alg:recenter}) starts with independent random signs and repeatedly improves the current signing by solving a strongly convex quadratic program and rounding the output.
\paragraph{Gaussian vectorization}
Let $D:=d(d+1)/2$ and $\Sym_d$ be the set of all $d\times d$ real symmetric 
matrices. For any $M\in \Sym_d$, define the operator $\svec:\Sym_d\to \R^D$ by 
\begin{equation}\label{eq:svec}
\svec(M) = \bigl((M_{aa})_{a=1}^d,(\sqrt{2}M_{ab})_{1\le a<b\le d}\bigr)\in\R^D.
\end{equation}
Given iid $A_1,\dots,A_N\sim \GOE(d)$, define
\begin{equation}\label{eq:vectorization}
 g_i=\sqrt{d/2}\,\svec(A_i),\qquad G=[g_1\ \cdots\ g_N].
\end{equation}
Then
\begin{equation*}
    \sum_i\sigma_i A_i=\sqrt{2/d}\,\mat(G\sigma),
\end{equation*}
where $\mat=\svec^{-1}$ (i.e., $\mat(\svec(M))=M$). Set $\alpha_d:=N/D$, so that $\alpha_d\to \alpha=2\tau \in(0,\infty)$.
Figure~\ref{fig:gaussian-vectorization} illustrates this transformation.
\begin{figure}[!htbp]
\centering
\begingroup
\resizebox{\linewidth}{!}{%
\definecolor{gvLavender}{HTML}{C8C4F0}
\definecolor{gvCoral}{HTML}{DF9692}
\definecolor{gvMuted}{HTML}{EFEDEF}
\begin{tikzpicture}[x=1cm,y=1cm,>=Stealth,
  every node/.style={font=\normalsize},
  arrow/.style={->,line width=.85pt},
  bracket/.style={line width=.8pt,line cap=round}]
\path[use as bounding box] (-.12,.20) rectangle (15.75,5.60);

\node at (1.66,5.42) {$A_i$};
\foreach \r in {1,...,4} {
  \foreach \c in {1,...,4} {
    \pgfmathsetmacro{\xx}{(\c-1)*.83}
    \pgfmathsetmacro{\yy}{4.34-(\r-1)*.83}
    \ifnum\r=\c
      \fill[gvLavender] (\xx,\yy) rectangle ++(.78,-.78);
      \node at (\xx+.39,\yy-.39) {$a_{\r\c}$};
    \else
      \ifnum\r<\c
        \fill[gvCoral] (\xx,\yy) rectangle ++(.78,-.78);
        \node at (\xx+.39,\yy-.39) {$a_{\r\c}$};
      \else
        \fill[gvMuted] (\xx,\yy) rectangle ++(.78,-.78);
        \node[text=black!35] at (\xx+.39,\yy-.39) {$a_{\c\r}$};
      \fi
    \fi
  }
}
\draw[bracket] (-.06,4.42)--(-.20,4.42)--(-.20,.99)--(-.06,.99);
\draw[bracket] (3.33,4.42)--(3.47,4.42)--(3.47,.99)--(3.33,.99);

\draw[arrow] (3.80,2.705)--(4.99,2.705);
\node at (6.64,5.42) {$g_i$};
\node at (5.47,2.705) {$\sqrt{\frac d2}$};
\foreach \k/\t in {
  0/a_{11},1/a_{22},2/a_{33},3/a_{44},
  4/\sqrt2\,a_{12},5/\sqrt2\,a_{13},6/\sqrt2\,a_{14},
  7/\sqrt2\,a_{23},8/\sqrt2\,a_{24},9/\sqrt2\,a_{34}} {
  \ifnum\k<4
    \fill[gvLavender] (6.07,4.99-\k*.47) rectangle ++(1.14,-.44);
  \else
    \fill[gvCoral] (6.07,4.99-\k*.47) rectangle ++(1.14,-.44);
  \fi
  \node[font=\small] at (6.64,4.77-\k*.47) {$\t$};
}
\draw[bracket] (5.99,5.06)--(5.85,5.06)--(5.85,.25)--(5.99,.25);
\draw[bracket] (7.29,5.06)--(7.43,5.06)--(7.43,.25)--(7.29,.25);

\draw[arrow] (7.78,2.705)--(9.17,2.705);
\node at (12.62,5.42) {$G$};
\foreach \x in {9.80,10.81,14.73} {
  \foreach \k in {0,...,9} {
    \ifnum\k<4
      \fill[gvLavender,opacity=.30] (\x,4.99-\k*.47) rectangle ++(.62,-.44);
    \else
      \fill[gvCoral,opacity=.30] (\x,4.99-\k*.47) rectangle ++(.62,-.44);
    \fi
    \fill[black!35] (\x+.31,4.77-\k*.47) circle (.020);
  }
}
\foreach \k in {0,...,9} {
  \ifnum\k<4
    \fill[gvLavender] (12.73,4.99-\k*.47) rectangle ++(.62,-.44);
  \else
    \fill[gvCoral] (12.73,4.99-\k*.47) rectangle ++(.62,-.44);
  \fi
  \fill (13.04,4.77-\k*.47) circle (.024);
}
\draw[line width=.7pt] (12.69,5.03) rectangle (13.39,.28);
\node at (12.04,2.705) {$\cdots$};
\node at (14.02,2.705) {$\cdots$};
\draw[bracket] (9.64,5.06)--(9.50,5.06)--(9.50,.25)--(9.64,.25);
\draw[bracket] (15.51,5.06)--(15.65,5.06)--(15.65,.25)--(15.51,.25);
\end{tikzpicture}%
}
\endgroup
\caption{Gaussian vectorization for $d=4$. Read the diagonal entries and one copy of each off-diagonal entry into a vector. Each matrix $A_i$ becomes one column $g_i$ of $G$.}
\label{fig:gaussian-vectorization}
\end{figure}
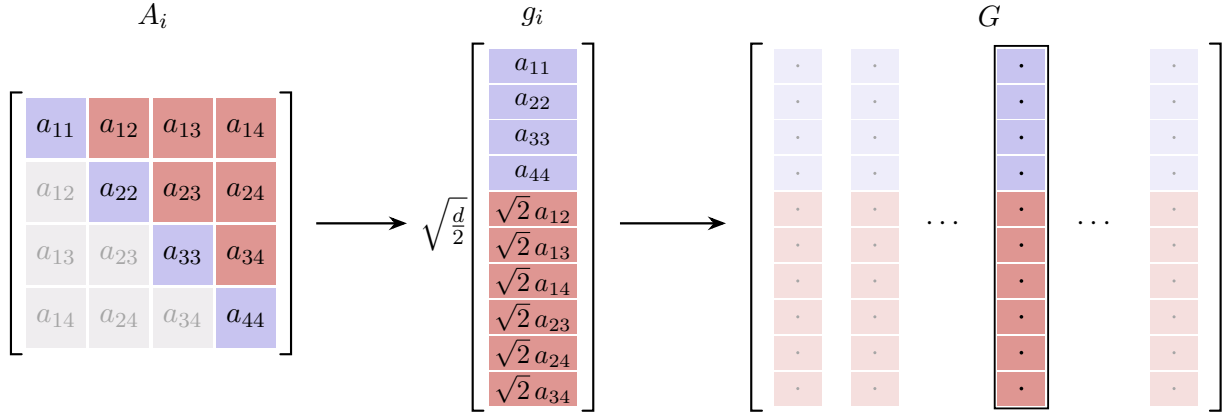

\paragraph{Recenter-and-round algorithm}
Set $K=K_d:=\lceil \log_2 \alpha_d\rceil +1$. Given a current signing $\sigma\in\{-1,1\}^N$, consider the strongly convex program
\begin{equation}\label{eq:QP-original}
  z(G,\sigma)
  :=\arg\min_{z\in[-1,1]^N}\Psi_{G,\sigma}(z),
\end{equation}
where
\begin{equation}\label{eq:Psi}
  \Psi_{G,\sigma}(z)
  :=\frac1{2D}\norm{Gz}_2^2
  +\sum_{i=1}^N\left[p_i(z;\sigma)+\frac12p_i(z;\sigma)^2\right],
  \qquad
  p_i(z;\sigma):=\frac{1-\sigma_i z_i}{2}.
\end{equation}
Note that
\begin{equation}\label{eq:objective}
      \Psi_{G,\sigma}(z)
=\frac1{2D}\norm{Gz}_2^2+\frac18\norm z_2^2-\frac34\ip{\sigma}{z}+\frac{5N}{8},
\end{equation}
with Hessian $\nabla_z^2\Psi_{G,\sigma}(z)=G^\top G/D+I_N/4 \succeq I_N/4$. Thus, $\Psi_{G,\sigma}$ is strongly convex and the minimizer in~\eqref{eq:QP-original} is unique. For $z\in[-1,1]^N$, let $\operatorname{Round}(z;U)\in\{\pm1\}^N$ be the coordinate-wise mean-preserving rounding:
\begin{equation}\label{eq:rounding}
  \operatorname{Round}(z;U)_i
  =\begin{cases}
    +1,&U_i\le(1+z_i)/2,\\
    -1,&U_i>(1+z_i)/2,
  \end{cases}
\end{equation}
where $U_1,\ldots,U_N$ are independent uniform random variables on $[0,1]$. In particular, $\sigma_i$ is flipped with probability $p_i(z;\sigma)$. 

Figure~\ref{fig:recenter-round} illustrates the recentering and rounding steps.
\begin{figure}[!htbp]
\centering
\begingroup
\resizebox{\linewidth}{!}{%
\begin{tikzpicture}[x=1cm,y=1cm,>=Stealth,
  every node/.style={font=\small},
  dot/.style={circle,fill=black,inner sep=1.9pt},
  fractional/.style={circle,fill=white,draw=black,line width=.85pt,inner sep=2pt},
  move/.style={->,black,line width=1.15pt},
  level/.style={draw=rrcoral!85!black,line width=.55pt}]
\definecolor{rrlavender}{RGB}{177,173,233}
\definecolor{rrcoral}{RGB}{231,107,105}
\def\s{1.65}
\def\cx{.38077}
\def\cy{.38077}

\begin{scope}[shift={(1.8,2.55)}]
  \node[font=\small\bfseries] at (0,2.35) {Recenter};
  \fill[rrlavender!48] (-\s,-\s) rectangle (\s,\s);
  \begin{scope}
    \clip (-\s,-\s) rectangle (\s,\s);
    \begin{scope}[shift={(\cx,\cy)},rotate=-45]
      \fill[rrcoral,opacity=.19] (0,0) ellipse (1.90 and 1.178);
      \fill[rrcoral,opacity=.20] (0,0) ellipse (1.28 and .794);
      \fill[rrcoral,opacity=.24] (0,0) ellipse (.66 and .409);
      \foreach \a/\b in {.66/.409,1.28/.794,1.90/1.178,2.52/1.563}
        \draw[level] (0,0) ellipse [x radius=\a cm,y radius=\b cm];
    \end{scope}
  \end{scope}
  \draw[black!65,line width=.7pt] (-\s,-\s) rectangle (\s,\s);
  \foreach \x/\y in {-1.65/-1.65,-1.65/1.65,1.65/-1.65}
    \fill[black!45] (\x,\y) circle (1.3pt);
  \draw[move,shorten <=4pt,shorten >=5pt] (\s,\s) -- (\cx,\cy);
  \node[dot] at (\s,\s) {};
  \node[anchor=south east,inner sep=2pt] at (1.59,1.73) {$\sigma^{(k)}$};
  \node[fractional] at (\cx,\cy) {};
  \node[anchor=north east,inner sep=3pt] at (\cx,\cy-.02) {$z^{(k)}$};
\end{scope}

\begin{scope}[shift={(7.0,2.55)}]
  \node[font=\small\bfseries] at (0,2.35) {Round};
  \fill[rrlavender!48] (-\s,-\s) rectangle (\s,\s);
  \draw[black!65,line width=.7pt] (-\s,-\s) rectangle (\s,\s);
  \foreach \x/\y in {-1.65/-1.65,1.65/1.65,1.65/-1.65}{
    \draw[->,black!35,dashed,line width=.65pt,shorten <=5pt,shorten >=3pt]
      (\cx,\cy) -- (\x,\y);
    \fill[black!45] (\x,\y) circle (1.5pt);
  }
  \draw[move,shorten <=5pt,shorten >=4pt] (\cx,\cy) -- (-\s,\s);
  \node[fractional] at (\cx,\cy) {};
  \node[anchor=west,inner sep=4pt] at (\cx+.03,\cy-.02) {$z^{(k)}$};
  \node[dot] at (-\s,\s) {};
  \node[anchor=south west,inner sep=2pt] at (-1.59,1.73) {$\sigma^{(k+1)}$};
\end{scope}

\begin{scope}[shift={(12.2,2.55)}]
  \node[font=\small\bfseries] at (0,2.35) {Repeat};
  \fill[rrlavender!48] (-\s,-\s) rectangle (\s,\s);
  \begin{scope}
    \clip (-\s,-\s) rectangle (\s,\s);
    \begin{scope}[shift={(-.99,.99)},rotate=-45]
      \fill[rrcoral,opacity=.19] (0,0) ellipse (1.90 and 1.178);
      \fill[rrcoral,opacity=.20] (0,0) ellipse (1.28 and .794);
      \fill[rrcoral,opacity=.24] (0,0) ellipse (.66 and .409);
      \foreach \a/\b in {.66/.409,1.28/.794,1.90/1.178,2.52/1.563,3.14/1.947}
        \draw[level] (0,0) ellipse [x radius=\a cm,y radius=\b cm];
    \end{scope}
  \end{scope}
  \draw[black!65,line width=.7pt] (-\s,-\s) rectangle (\s,\s);
  \foreach \x/\y in {-1.65/-1.65,1.65/1.65,1.65/-1.65}
    \fill[black!45] (\x,\y) circle (1.3pt);
  \draw[move,shorten <=4pt,shorten >=5pt] (-\s,\s) -- (-.99,.99);
  \node[dot] at (-\s,\s) {};
  \node[anchor=south west,inner sep=2pt] at (-1.59,1.73) {$\sigma^{(k+1)}$};
  \node[fractional] at (-.99,.99) {};
  \node[anchor=north west,inner sep=3pt] at (-.92,.98) {$z^{(k+1)}$};
\end{scope}

\draw[->,line width=.8pt,black!65] (3.75,2.55) -- (5.00,2.55);
\draw[->,line width=.8pt,black!65] (8.95,2.55) -- (10.20,2.55);
\end{tikzpicture}%
}
\endgroup
\caption{A simple illustration of recentering and rounding}
\label{fig:recenter-round}
\end{figure}
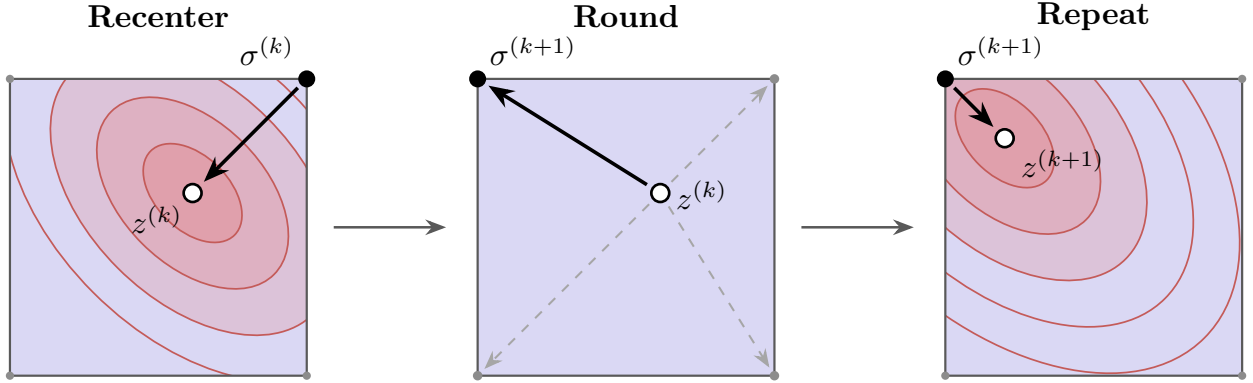

We now describe the algorithm in more detail. Starting with a random sign vector $\xi$, Algorithm~\ref{alg:recenter} iteratively updates the current signing $\sigma^{(k)}\in\{-1,1\}^N$ to $\sigma^{(k+1)}$ by first solving the quadratic program~\eqref{eq:QP-original}. The objective function is carefully designed: $\|Gz\|_2^2/(2D)$ favors a small residual, $\sum_i p_i(z;\sigma)$ controls the (expected) number of bit flips and subsequently the rounding noise, and $\sum_i p_i(z;\sigma)^2/2$ ensures the optimizer of~\eqref{eq:QP-original} is unique.
The optimization step is followed by a mean-preserving rounding~\eqref{eq:rounding}, yielding $\sigma^{(k+1)}\in\{-1,1\}^N$. Note that the algorithm is indeed \emph{offline}: at any given step, it has access to all columns of $G$ and may revise every previously chosen sign.

We analyze~\Cref{alg:recenter} with exact minimizers. For a polynomial-time implementation in real arithmetic, each minimizer may instead be computed to an $\ell_2$ error of $D^{-10}$ by standard convex programming methods.\footnote{Couple the exact and approximate versions using the same initial signing and rounding uniforms. Before their first disagreement, the probability that some rounded coordinate differs in a round is at most $\frac12\|\widetilde z-z\|_1\le\frac12\sqrt N D^{-10}$. A union bound over $K_d=O_\alpha(1)$ rounds gives output agreement with probability $1-o(1)$. Applying the same bound to both Gaussian marginals transfers stability uniformly in the correlation.} All discrepancy and stability guarantees below also hold for this implementation.

\begin{remark}[Approximate minimizers]
We analyze~\Cref{alg:recenter} with exact minimizers, but feasible $\ell_2$-approximations of accuracy $D^{-10}$ suffice for a polynomial-time implementation in real arithmetic. Using the same initial signing and rounding uniforms, the conditional probability of disagreement in each round, given agreement so far, is at most $\sqrt N D^{-10}/2$. A union bound over $K_d=O_\alpha(1)$ rounds therefore transfers the discrepancy guarantee. Applying the coupling to both Gaussian inputs also transfers stability, uniformly in their correlation.
\end{remark}

Algorithm~\ref{alg:recenter} is defined for $\alpha_d\ge 2$ and the rigorous guarantees hold for sufficiently large $\alpha_d$, see below. Outputting a random independent $\sigma\in\{-1,1\}^N$ for $\alpha_d<2$ specifies it on all inputs. 
\begin{algorithm}[H]
\color{addedblue}
\caption{Recenter-and-round}\label{alg:recenter}
\algrenewcommand{\algorithmicrequire}{\textbf{Input:}}
\begin{algorithmic}[1]
\Require $G\in\R^{D\times N}$ with $\alpha_d = N/D\ge 2$.
\State Draw an independent uniform vector $\xi\in\{-1,+1\}^N$. 
\State Set $\sigma^{(0)}\gets\xi$  and $K_d\gets\lceil\log_2 \alpha_d\rceil+1$.
\For{$k=0,\ldots,K_d-1$}
\State Compute the unique minimizer of $z^{(k)}=\arg\min_{z\in[-1,1]^N}\Psi_{G,\sigma^{(k)}}(z)$; see~\eqref{eq:QP-original}.
\State Draw independent $U^{(k)}=(U_i^{(k)}:i\in[N])\sim \operatorname{Unif}([0,1]^N)$, set $\sigma^{(k+1)}\gets\operatorname{Round}(z^{(k)};U^{(k)})$.
\EndFor
\State \Return $\widehat\sigma=\sigma^{(K_d)}$.
\end{algorithmic}
\end{algorithm}

\paragraph{Discrepancy guarantee}
Our first main result establishes that for small enough $\kappa>0$, Algorithm~\ref{alg:recenter} finds signs satisfying~\eqref{eq:target} whp provided $\tau \ge C_+/(\kappa^2\log\frac1\kappa)$, for an absolute $C_+>0$.
\begin{theorem}[Recenter-and-round algorithm: discrepancy upper bound]\label{thm:stable-upper}
There exists universal constants $\tau_0,C,C_+>0$ such that the following holds. Let $A_1,\dots,A_N\sim \GOE(d)$ be independent random matrices with $N/d^2\to \tau\ge \tau_0$, $G\in\R^{D\times N}$ be the matrix with columns $g_i$ where $g_i = \sqrt{d/2}\svec(A_i)$ per~\eqref{eq:vectorization}, and $\widehat{\sigma}$ be the output of~\Cref{alg:recenter} with input $G$. Then, as $d\to\infty$,
\begin{equation}\label{eq:spectral-upper}
 \Pp\!\left[\frac1{\sqrt N}\left\|\sum_{i=1}^N\widehat\sigma_iA_i\right\|_{\op}
 \le\frac C{\sqrt{\tau\log(2+\tau)}}\right]\longrightarrow1.
\end{equation}
Consequently, for every sufficiently small $\kappa>0$, Algorithm~\ref{alg:recenter} succeeds as long as $\tau \ge C_+/(\kappa^2\log\frac1\kappa)$.
\end{theorem}
See below for a proof overview, and Section~\ref{sec:offline-proofs} for the complete proof.
\paragraph{Noise stability}
We now formally define the stability condition. Given iid $A_1,\dots,A_N\sim \GOE(d)$, the algorithm takes $G\in\R^{D\times N}$ with columns $g_i =\sqrt{d/2}\svec(A_i)\sim\cN(0,I_D)$ as its input. Let $\omega\in \Omega$ denote the internal randomness of the algorithm, sampled independently of the input. Couple two independent Gaussian inputs $G,H\in\R^{D\times N}$ by 
$G^\rho:=\rho G +\sqrt{1-\rho^2}H$.
\begin{definition}[Noise stability]\label{def:strong}
Given $N/D\to \alpha\in(0,\infty)$, a family $\A_{D,N}:\R^{D\times N}\times \Omega\to \{-1,1\}^N$ of randomized algorithms is stable if, for any $\eps>0$, there exists a $\rho_0=\rho_0(\eps,\alpha)<1$ such that
\[
 \lim_{d\to\infty}\ \sup_{\rho\in[\rho_0,1)}
 \Pp\!\left(\dist\big(\A_{D,N}(G,\omega),\A_{D,N}(G^\rho,\omega)\big)>\eps N\right)=0,
\]
where $\dist$ denotes Hamming distance. Here, $\A_{D,N}$ is permitted to depend on $\kappa>0$.
\end{definition}
Informally, $\A_{D,N}$ is stable if, uniformly for all $\rho\ge \rho_0$, the outputs on correlated copies $G$ and $G^\rho$ are close in Hamming distance. That is, the normalized Hamming distance between two outputs is at most $\varepsilon$ whp. 

\begin{theorem}[Stability of recenter-and-round]\label{thm:algorithm-stability}
For every fixed $\alpha\ge \alpha_0$, \Cref{alg:recenter} is noise stable in the sense of \Cref{def:strong}.
\end{theorem}
Several remarks are in order. The stability condition in~\Cref{def:strong} implies finite-parameter variants of stability considered in works~\cite{gamarnikKizildagPerkinsXu2022,gamarnikKizildagPerkinsXu2023,liSchrammZhou2024}, adapted to matrix discrepancy in~\Cref{def:alt-stability}. We adopt~\Cref{def:strong} both to record the stronger
stability property of our algorithm and to provide cleaner lower bound arguments.

Next, the order of quantifiers in~\cref{def:strong} is crucial: for each prescribed Hamming tolerance $\eps>0$, a single correlation
value $\rho_0<1$ must make the probability vanish as
$d\to\infty$, uniformly over $\rho\in[\rho_0,1)$. This is strictly stronger than requiring, for every $\eps>0$, that $\lim_{\rho\uparrow1}\limsup_{d\to\infty}\Pp[\dist(\A(G,\omega),\A(G^\rho,\omega))>\eps N]=0$. For example, consider the deterministic algorithm
$\widehat{\A}(G)=\operatorname{sign}(G_{11})\boldsymbol{1}_N$,
where $\boldsymbol{1}_N$ is the all-ones vector. Its outputs on $G$ and $G^\rho$ disagree either in every coordinate
or in none, according to whether $G_{11}$ and $(G^\rho)_{11}$
have different signs. Consequently, for every $0<\eps<1$ and
$\rho\in[0,1)$, $\Pp[\dist(\widehat{\A}(G),\widehat{\A}(G^\rho,\omega))>\eps N] = \arccos(\rho)/\pi$. This probability tends to zero as $\rho\uparrow1$, but is independent
of the dimension and strictly positive at every fixed $\rho<1$.
Thus $\widehat{\A}$ satisfies the weaker condition but is not
stable in the sense of \cref{def:strong}. 

\paragraph{Proof techniques: fractional correction and repeated rounding}
Our algorithm alternates fractional correction with randomized rounding
to reduce the Euclidean residual $\|G\sigma\|_2$ while controlling the number of sign flips. Starting from an independent random sign vector $\xi$, set
$t=\frac12\sqrt{\log\alpha_d}$ and
$p:=\Pp[Z\ge t]$ for $Z\sim\cN(0,1)$, where
$\alpha_d=N/D\to\alpha$. For sufficiently large fixed $\alpha$, a uniform Gaussian-cap estimate gives, whp, $\inf_{u\in\mathbb{S}^{D-1}}|i\in[N]:\ip{u}{\xi_i g_i}\ge t| \ge pN/2$, see~\Cref{lem:caps}. In other words, simultaneously for every unit direction $u$, at least $pN/2$ of the signed columns $\xi_i g_i$ have projection at least $t$ onto $u$. We call the corresponding indices \emph{useful} in direction $u$. If a signing $\sigma$ differs from $\xi$ in at most $pN/8$ coordinates, then in any given direction at most $pN/8$ indices can cease to be useful. Thus, every direction still has at least $pN/2-pN/8=3pN/8$ useful indices, so this geometric property persists throughout the Hamming ball of radius $pN/8$ around $\xi$.

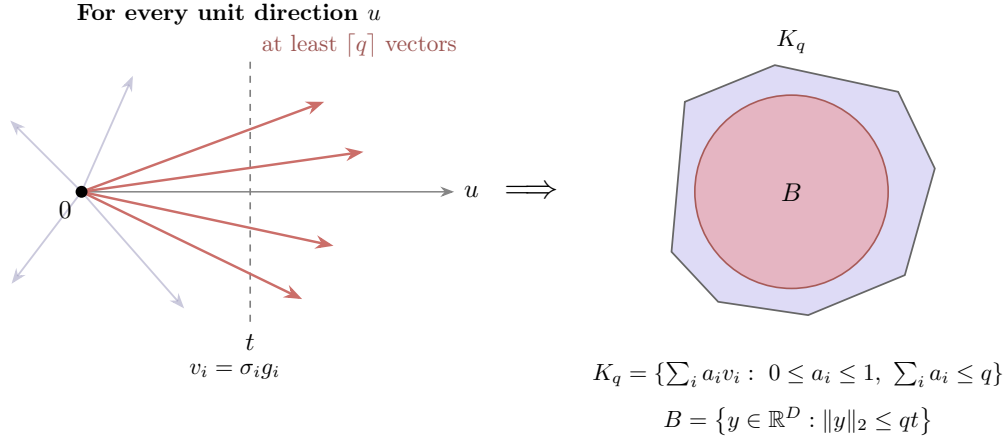
\begin{figure}[H]
\centering
\begingroup
\resizebox{.98\linewidth}{!}{%
\definecolor{glLavender}{HTML}{C8C4F0}
\definecolor{glCoral}{HTML}{EF817C}
\begin{tikzpicture}[x=1cm,y=1cm,>=Stealth,
 every node/.style={font=\normalsize},
 useful/.style={->,draw=glCoral!85!black,line width=1.05pt},
 other/.style={->,draw=glLavender!80!black,opacity=.55,line width=.8pt},
 dot/.style={circle,fill=black,inner sep=1.8pt}]
\path[use as bounding box] (-.08,-.05) rectangle (17.90,6.70);
\node[font=\small\bfseries] at (3.56,6.45) {For every unit direction $u$};

\coordinate (L0) at (1.35,3.79);
\draw[->,black!48,line width=.65pt] (L0)--(6.94,3.79);
\node[anchor=west,inner sep=4pt] at (6.94,3.79) {$u$};
\draw[black!50,dashed,line width=.7pt] (3.88,1.84)--(3.88,5.76);
\node[anchor=north,inner sep=4pt] at (3.88,1.80) {$t$};
\draw[other] (L0)--(.28,4.86);
\draw[other] (L0)--(2.12,5.53);
\draw[other] (L0)--(2.89,2.04);
\draw[other] (L0)--(.30,2.41);
\draw[useful] (L0)--(4.99,5.14);
\draw[useful] (L0)--(5.58,4.39);
\draw[useful] (L0)--(5.13,2.98);
\draw[useful] (L0)--(4.66,2.18);
\node[dot] at (L0) {};
\node[anchor=north east,inner sep=4pt] at (L0) {$0$};
\node[align=center,font=\small,text=glCoral!65!black] at (5.52,5.99) {at least $\lceil q\rceil$ vectors};
\node[font=\small] at (3.65,1.10) {$v_i=\sigma_i g_i$};

\node[font=\Large] at (8.10,3.79) {$\Longrightarrow$};

\begin{scope}[shift={(12.00,3.79)}]
  \filldraw[fill=glLavender,fill opacity=.60,draw=black!60,line width=.75pt]
    (-1.80,-.90)--(-1.60,1.35)--(-.25,1.90)--(1.60,1.50)--(2.15,.35)
    --(1.70,-1.25)--(.25,-1.85)--(-1.10,-1.65)--cycle;
  \filldraw[fill=glCoral,fill opacity=.42,draw=glCoral!70!black,line width=.7pt]
    (0,0) circle (1.45);
  \node[font=\small] at (0,2.25) {$K_q$};
  \node at (0,0) {$B$};
\end{scope}
\node[font=\small,align=center] at (12.10,1.05)
  {$K_q=\left\{\sum_i a_i v_i:\ 0\le a_i\le1,\ \sum_i a_i\le q\right\}$};
\node[font=\small,align=center] at (12.10,.35)
  {$B=\left\{y\in\mathbb{R}^D: \|y\|_2\le qt\right\}$};
\end{tikzpicture}%
}
\endgroup
\caption{Ball containment. If every unit direction has at least $\lceil q\rceil$ vectors with projection at least $t>0$, then $K_q$ contains the ball $B$ of radius $qt$ centered at zero. Thus every $y$ with $\|y\|_2\le qt$ can be written as $y=\sum_i a_i v_i$, with $0\le a_i\le1$ and $\sum_i a_i\le q$.}
\label{fig:geometric-lemma-proof}
\end{figure}

We then use a geometric lemma to turn this directional property into fractional adjustments of the signs that cancel the residual. Specifically, if every direction has at least $\lceil q\rceil$ useful vectors $v_i$, then
every $\tilde v\in\R^D$ with $\|\tilde v\|_2\le qt$ admits a
representation $\tilde{v} = \sum_{i\le N}a_i v_i$, where $a_i\in[0,1]$ and $\sum_{i\le N}a_i\le q$; see \Cref{lem:zonotope}. For a current signing $\sigma$ in a neighborhood of the sign vector $\xi$, let $v_i=\sigma_i g_i$, $s=G\sigma$,
$R=\|s\|_2$, and $q=R/(2t)$. Provided
$\lceil q\rceil\le3pN/8$, applying~\Cref{lem:zonotope} to $\tilde{v} = s/2$ produces a comparison point $z_i^\circ = \sigma_i(1-2a_i)$ which satisfies
$Gz^\circ=0$ and $\sum_i p_i(z^\circ;\sigma)\le q$.
Figure~\ref{fig:geometric-lemma-proof} illustrates this ball containment.

Consequently, the unique optimizer $z$ of~\eqref{eq:QP-original}
satisfies
$\Psi_{G,\sigma}(z)\le\Psi_{G,\sigma}(z^\circ)\le3q/2$, and hence
\[
    \|Gz\|_2^2\le3Dq,
    \qquad
    \sum_i p_i(z;\sigma)\le\frac32q.
\]
Thus, although the optimizer need not cancel the residual exactly,
both its fractional residual and its expected rounding cost are small.

We next control the noise introduced by rounding, conditional on the
input and all previous rounds. Writing
$\sigma'=\operatorname{Round}(z;U)$ and $e_i=\sigma_i'-z_i$,
the errors are conditionally independent and centered, with variances
$\nu_i=1-z_i^2\in[0,1]$. Their covariance satisfies
$C_z=G\operatorname{diag}(\nu_i)G^\top\preceq GG^\top$.
Thus, bounds on $\|G\|_{\rm op}$ and the column norms
(\Cref{lem:covariance}), together with a fourth-moment calculation,
control $\|\sum_i e_i g_i\|_2$ even when $z$ depends on the entire
past; see \Cref{lem:rounding}. The leading constants in this bound
are universal; dependence on the fixed aspect ratio enters only
through the failure probability. Combining it with the
fractional estimates gives,
\[
    R_{k+1}\le C\sqrt{\frac{DR_k}{t}}+\frac{CD}{t},
    \qquad
    d_H(\sigma^{(k+1)},\sigma^{(k)})
    \le\frac{3R_k}{2t}+\frac{D}{t^2},
\]
for a universal $C>0$  with high conditional probability, where $R_k=\|G\sigma^{(k)}\|_2$. For  sufficiently large fixed $\alpha$, crude bounds on the residual and the cumulative number of flips keep the trajectory within the prescribed Hamming neighborhood of the sign vector $\xi$ whp, ensuring that sufficiently many useful columns remain
available at every round. Lastly, with $R_0=O(D\sqrt{\alpha})$, the recurrence yields after $R_K = O(D/t) = O(D/\sqrt{\log \alpha})$ whp for $K=O(\log\alpha)$, proving the Euclidean estimate~\eqref{eq:euclidean-upper}.

\paragraph{Rotational symmetry and the spectral guarantee}
The preceding argument controls the Euclidean residual $\|G\widehat{\sigma}\|_2$ of the output $\widehat{\sigma}$ of~\Cref{alg:recenter}. The key insight is that the algorithm preserves rotational symmetry, allowing us to upgrade this Euclidean estimate to the claimed operator-norm guarantee.

For every orthogonal $Q\in\R^{D\times D}$, the identity $\Psi_{QG,\sigma}(z)=\Psi_{G,\sigma}(z)$ and uniqueness of the minimizer imply that, with the same random seed $\omega$, $\A(QG,\omega)=\A(G,\omega)$. Combined with $QG\distr G$, this
shows that $X=G\widehat{\sigma}$ has a rotationally invariant law;
see~\Cref{lem:uniform-direction}. Since $X\ne0$ almost surely, we may write $X=RU$, where $R=\|X\|_2$ and $U$ is uniform on $\mathbb{S}^{D-1}$ and independent of $R$. Moreover,
$\mat(U)\distr W/\|W\|_F$ for $W\sim\GOE(d)$, so $\sqrt{d}\left\|\mat(U)\right\|_{\op}\convp 2$. Together with $\sum_i\widehat{\sigma}_iA_i=\sqrt{2/d}\,R\mat(U)$, this converts the radial estimate $R=O(D/\sqrt{\log\alpha})$ into the claimed operator norm bound, even though the residual $X$ itself need not be Gaussian.

\begin{remark}[Comparison with vector discrepancy] This contrasts sharply with the classical $\ell_\infty$ vector discrepancy. For a uniform direction $U$, note that $\sqrt{D/(2\log D})\|U\|_\infty \convp 1$, consequently, the preceding radial estimate yields only $\|G\widehat{\sigma}\|_\infty/\sqrt{N} = O(
        \sqrt{\log D/(\alpha\log\alpha)})$.  At fixed aspect ratio $N/D\to\alpha$, this gives an $O_\alpha(\sqrt{N\log N})$ bound, falling short of establishing $O_\alpha(\sqrt N)$. In particular, the
Euclidean-to-$\ell_\infty$ conversion incurs a
$\sqrt{\log D}$ loss that has no analogue in the matrix case.
\end{remark}
\paragraph{Stability proof}
We prove stability by comparing the fractional minimizers and coupling their rounding. Strong convexity and the variational inequalities for
the box-constrained minimizers give
\[
 \frac{\|z(G,\sigma)-z(\widetilde G,\widetilde\sigma)\|_2}{\sqrt N}
 \le 6\sqrt{\frac{d_H(\sigma,\widetilde\sigma)}N}
 +\frac4D(\|G\|_{\rm op}+\|\widetilde G\|_{\rm op})
 \|G-\widetilde G\|_{\rm op},
\]
for arbitrary inputs $G,\widetilde G$ and signings
$\sigma,\widetilde\sigma$. For an entrywise $\rho$-correlated pair $\widetilde{G}=G^\rho$, Gaussian operator norm bounds imply that the
second term is $O_\alpha(\sqrt{1-\rho})$ with high probability,
uniformly over $\rho\in[0,1)$~\cite{vershynin2018}. Coupling the two runs through the same initial signing and rounding uniforms,
and applying Hoeffding's inequality, yields, for each fixed $\gamma>0$,
\[
 x_{k+1}\le3\sqrt{x_k}+C_\alpha\sqrt{1-\rho}+\gamma,
 \qquad x_0=0,\qquad
 x_k:=\frac1N d_H(\sigma^{(k)},\widetilde\sigma^{(k)}),
\]
whp throughout the run. The number of rounds
is bounded for fixed $\alpha$; we can choose $\gamma>0$ sufficiently small and $\rho_0<1$ sufficiently close to one so that $x_k\le \epsilon$ for all $k\le K$ whp, uniformly over $\rho\in[\rho_0,1)$. This establishes noise stability.
\subsection{Online Frobenius-greedy algorithm}\label{sec:online-overview}
In the online setting, inputs arrive sequentially and the corresponding signs are chosen irrevocably. Under this model, we consider a greedy procedure. Given the partial sum $S_{i-1}:=\sum_{j\le i-1}\sigma_j A_j$, it chooses $\sigma_i$ minimizing the Frobenius norm, $\sigma_i =\argmin_{\sigma\in\{-1,1\}}\|S_{i-1}+\sigma A_i\|_F$. We call this algorithm as \emph{Frobenius-greedy}  (\Cref{alg:Frobenius-greedy}) and exactly characterize its limiting performance below.
\paragraph{Online Frobenius-greedy algorithm} Initializing with $\sigma_1=1$ and $S_1=A_1$, the greedy procedure sets
\[
\sigma_k = \argmin_{\sigma\in\{-1,1\}} \bigl\|S_{k-1}+\sigma A_k\bigr\|_{\rm F},\quad\text{where}\quad S_{k-1} = \textstyle\sum_{j=1}^{k-1}\sigma_j A_j
\]
for $k=2,\dots,N$. In particular, $\sigma_k = -{\rm sign}\ip{S_{k-1}}{A_k}_{\rm F}$ using the identity 
\[
\|S_{k-1}+\sigma A_k\|_F^2
 =\|S_{k-1}\|_F^2+\|A_k\|_F^2
  +2\sigma\langle S_{k-1},A_k\rangle_F.
  \]
  Ties are broken arbitrarily, i.e., declare $\sigma_k=1$ if $\ip{S_{k-1}}{A_k}_{\rm F}=0$ which occurs with probability zero for $A_1,\dots,A_N\sim \GOE(d)$. Note that since each step takes $O(d^2)$ time and the algorithm runs for $N$ rounds, the total running time is $O(Nd^2)$, which is linear in the size of the input $(A_1,\dots,A_N)$. In what follows, we characterize the radial evolution of this algorithm and apply rotational symmetry to obtain operator-norm guarantees.
\begin{algorithm}[H]
\caption{Frobenius-greedy algorithm}\label{alg:Frobenius-greedy}
\begin{algorithmic}[1]
\Require $A_1,\dots,A_N\in\R^{d\times d}$ 
\State Set $S_0\gets 0$, $\sigma_1 = 1$, and $S_1 = \sigma_1 A_1$.
\For{$k=2,3,\ldots,N$}
  \State Set $\sigma_k = -{\rm sign}\ip{S_{k-1}}{A_k}_{\F}$. If $\ip{S_{k-1}}{A_k}_{\F} = 0$, declare $\sigma_k=1$.
  \State Update $S_k\gets S_{k-1} + \sigma_k A_k$
\EndFor
\State \Return $\sigma_1,\dots,\sigma_N$.
\end{algorithmic}
\end{algorithm}
\paragraph{Exact performance curve}
Define $r:[0,\infty)\to [0,\sqrt{\pi/8})$ as the unique solution to the differential equation 
\begin{equation}\label{eq:diff-eq}
    \frac{d}{dt}r(t)^2 = 1-\sqrt{\frac{8}{\pi}}\,r(t),\quad r(0)=0.
\end{equation}
Equivalently,
\begin{equation}\label{eq:r-implicit}
    t = -\sqrt{\frac{\pi}{2}}\,r(t) - \frac{\pi}{4}\log\left(1-\sqrt{\frac{8}{\pi}}\,r(t)\right).
\end{equation}
This solution also has the explicit representation
\[
    r(t)=\sqrt{\frac{\pi}{8}}\left[1+W_0\!\left(-e^{-1-4t/\pi}\right)\right],\qquad t\ge0,
\]
where $W_0$ is the principal real branch of the Lambert $W$ function~\cite{corless1996lambert}, characterized by $W_0(z)e^{W_0(z)}=z$ and $W_0(z)\ge-1$ for $z\in[-e^{-1},0)$. This branch gives $0\le r(t)<\sqrt{\pi/8}$.

The following result determines the exact limiting discrepancy of Frobenius-greedy.
\begin{theorem}[Exact Frobenius-greedy curve]\label{thm:greedy-performance}
    Suppose that $A_1,\dots,A_N\sim \GOE(d)$ where $N/d^2\to \tau \in(0,\infty)$. Then, the signs $\sigma_1,\dots,\sigma_N$ generated by Algorithm~\ref{alg:Frobenius-greedy} satisfy
    \[
    \frac{1}{\sqrt{N}}\left\|\sum_{i=1}^N \sigma_i A_i\right\|_{\rm op}\xrightarrow{\Pp} \Phi(\tau)
    \]
    where $\Phi(\tau) = \sqrt{\frac{2}{\tau}}r(2\tau)$ for $r$ defined through~\eqref{eq:diff-eq}-\eqref{eq:r-implicit}. Consequently, the algorithm succeeds whp if $\kappa>\Phi(\tau)$ and fails whp if $\kappa<\Phi(\tau)$.
\end{theorem}
We omit the edge case $\kappa=\Phi(\tau)$. Inverting the performance curve identifies the density threshold for Frobenius-greedy.
\begin{theorem}[Threshold for $\tau$]\label{thm:curve-invert}
The function $r(t)$ implicitly defined by~\eqref{eq:r-implicit} is the unique solution to the differential equation~\eqref{eq:diff-eq}. Moreover, for $0<\kappa<2$, let $a_\kappa\in(0,1)$ be the unique solution to 
    \[
    \kappa^2 = \frac{2a_\kappa^2}{-a_\kappa-\log(1-a_\kappa)}.
    \]
    Then, the algorithm works whp provided 
    \[
    \tau>\tau_{\rm FG}(\kappa) = \frac{\pi}{8}\bigl[-a_\kappa - \log(1-a_\kappa)\bigr] = \frac{\pi a_\kappa^2}{4\kappa^2} = \frac{\pi}{4\kappa^2}\bigl(1+o_\kappa(1)\bigr),
    \]
    where the last equality regards the limit $\kappa\downarrow 0$.
\end{theorem}
Figure~\ref{fig:frobenius-greedy-fluid} illustrates the limiting performance curve and radial dynamics.
\begin{figure}[H]
  \centering
  \includegraphics[width=\linewidth]{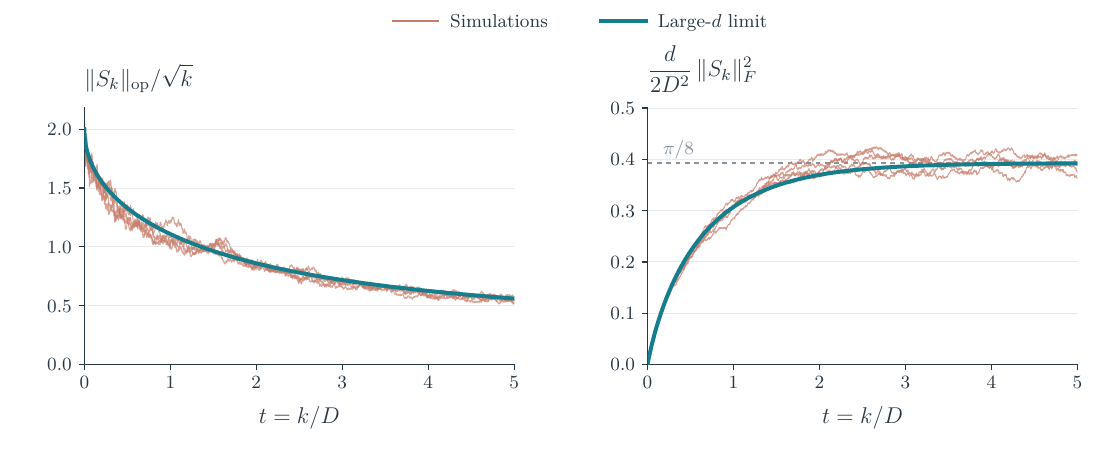}
  \caption{Frobenius-greedy dynamics at $d=45$: normalized operator-norm
  discrepancy (left) and normalized squared Frobenius norm (right).
  Thin curves show five independent runs; thick curves show the large-$d$ limits.
  The dashed line marks the equilibrium $\pi/8$.}
  \label{fig:frobenius-greedy-fluid}
\end{figure}

\paragraph{Proof ideas} Using the vectorization~\eqref{eq:vectorization}, write
$x_k=\sum_{i\le k}\sigma_i g_i$ and $R_k=\|x_k\|_2$, with $x_0=0$.
The greedy rule gives the exact radial recursion
\[
 R_k^2=R_{k-1}^2+\|g_k\|_2^2-2R_{k-1}|Z_k|,
\]
where $Z_k=\ip{u_{k-1}}{g_k}$ and
$u_{k-1}=x_{k-1}/R_{k-1}$ when $R_{k-1}>0$, with any fixed unit vector used otherwise. Conditional on $g_1,\dots,g_{k-1}$, the term
$Z_k$ is standard normal. Thus the incoming squared norm has
conditional mean $D$, while the greedy correction has conditional
mean $\sqrt{8/\pi} R_{k-1}$. For the normalized quantity $Y_k = R_k^2/D^2$, this gives the drift--martingale decomposition:
\[
Y_k = \frac1D \sum_{i=1}^k b(Y_{i-1}) + \frac{1}{D^2}\sum_{i=1}^k \eta_i,
\]
where $b:[0,\infty)\to \R$ is defined by $b(x) = 1-\sqrt{8/\pi}\sqrt{x}$ and $\eta_i =(\|g_i\|_2^2 - D) - 2R_{i-1}(|Z_i|-\sqrt{2/\pi})$.  On the rescaled time axis $t=k/D$, Doob's maximal inequality shows that the accumulated martingale
error vanishes uniformly on bounded time intervals. As for the drift term, while $b$
is not locally Lipschitz at zero, it is nonincreasing, which allows
a direct comparison of the corresponding integral equations. This establishes a \emph{fluid limit}: the interpolated trajectory
$\overline Y_D(t)=Y_{\lfloor Dt\rfloor}$ converges uniformly in
probability on bounded intervals to the unique nonnegative solution of
\[
    y'(t)=1-\sqrt{8/\pi}\sqrt{y(t)},\qquad y(0)=0.
\]
Solving this ODE gives the limiting radius, and rotational invariance
converts the radial limit into an operator-norm limit, as in the
offline analysis. In particular, we determine the entire limiting
performance curve of~\Cref{alg:Frobenius-greedy}, together with
its small-margin threshold asymptotics. The same scalar ODE governs the second moment of the optimal process in the $L^2$ mean-field problem studied by Fiedler, Jackson, Lacker, and Niles-Weed~\cite[Section~7]{fiedler2026meanfield}. Here, we establish this limit directly for Frobenius-greedy and use rotational invariance to obtain operator-norm guarantees.

\subsection{Algorithmic lower bounds}\label{sec:lower-overview}
In this section, we establish algorithmic lower bounds for stable offline and online algorithms that are sharp up to constants in the regime $\kappa\downarrow 0$. Given iid $A_1,\dots,A_N\sim \GOE(d)$ and $\kappa>0$, let $\mathcal{S}_{d,N}(\boldsymbol{A};\kappa)$ denote the set satisfying~\eqref{eq:target}:
\begin{equation*}
 \mathcal S_{d,N}(\boldsymbol{A};\kappa)
 =\left\{\sigma\in\{-1,+1\}^N:
\left\|\sum_{i=1}^N \sigma_i A_i\right\|_{\op}\le \kappa\sqrt{N}\right\}.
\end{equation*}
We now turn to the obstruction for stable algorithms.
\begin{theorem}[Lower bound for stable algorithms]\label{thm:stable-lower}
Set $c_{\rm st}=\log(4/e)/72 \approx 0.0053$.
For every sufficiently small fixed $\kappa>0$, fixed $0<\tau\le c_{\mathrm{st}}/(\kappa^2\log(1/\kappa))$ and large $d$ the following holds. No  stable algorithm (in the sense of~\Cref{def:strong}) whose input are iid $A_1,\dots,A_N\sim \GOE(d)$ with $N/d^2\to \tau$ can produce an output $\sigma\in \cS_{d,N}(\boldsymbol{A};\kappa)$ with probability tending to one. 
\end{theorem}
\Cref{thm:stable-lower} asserts that stable algorithm fail below the threshold $c_{\rm st}/(\kappa^2\log(1/\kappa))$. Namely, in the regime $\kappa\downarrow 0$,~\Cref{alg:recenter} is optimal up to constants within the class of stable algorithms. The stability notion we adopt (\Cref{def:strong}) leads to a cleaner lower bound argument; we show in~\Cref{sec:quantitative-stability} that our algorithmic lower bound still holds for a weaker finite-parameter notion of stability (\Cref{def:alt-stability}) similar to~\cite{gamarnikKizildagPerkinsXu2022,gamarnikKizildagPerkinsXu2023,liSchrammZhou2024}.

We next formalize the class of online algorithms and state our result.
\begin{definition}\label{def:On-Line}
A randomized algorithm $\A:\R^{d\times d}\cdots \times \R^{d\times d}\times \Omega\to \{-1,1\}^N$ with input $A_1,\dots,A_N$ and random seed $\omega\in\Omega$ independent of the input is \emph{online} if for every $1\le j\le N$, $\sigma_j:=\A(A_1,\dots,A_N,\omega)_j\in\{-1,1\}$ is measurable with respect to the $\sigma$-field generated by $(A_1,\dots,A_j,\omega)$. That is, for fixed $\omega\in\Omega$, $\sigma_j$ is a function of $A_1,\dots,A_j$ only.
\end{definition}
Note that~\Cref{alg:Frobenius-greedy} is online in the sense of~\Cref{def:On-Line} as $\sigma_j$ is a function of $A_1,\dots,A_j$ only. We next establish a sharp algorithmic lower bound within the class of online algorithms.
\begin{theorem}[Lower bound for online algorithms]\label{thm:Online-LowerBd}
There exists a universal constant $C>0$ such that, for every sufficiently small fixed $\kappa>0$ and fixed $0<\tau\le C/\kappa^2$, there is a constant $c=c(\kappa,\tau)>0$ such that the following holds for all large $d$. No online algorithm (in the sense of~\Cref{def:On-Line}) whose input are iid $A_1,\dots,A_N\sim\GOE(d)$ with $N/d^2\to\tau$ can produce an output $\sigma\in\cS_{d,N}(\boldsymbol{A},\kappa)$ with probability at least $e^{-cN}$.
\end{theorem}
Theorem~\ref{thm:Online-LowerBd} asserts that online algorithms provably fail if $\tau\le C/\kappa^2$ for a sufficiently small universal $C>0$. Namely, the guarantee of Algorithm~\ref{alg:Frobenius-greedy} is essentially sharp up to constants in the double limit $d\to\infty$ followed by $\kappa\downarrow 0$. We note that for any fixed signing, or a random signing independent of the input, in the regime $N/d^2\to\tau$ the probability that $\sum_{i\le N}\sigma_i A_i$ has operator norm at most $\kappa\sqrt{N}$ is $\exp(-\Theta(N))$~\cite{maillard2025}. As such, the probability guarantee ruled out by~\Cref{thm:Online-LowerBd} is essentially sharp (up to constants in the exponent). 

\paragraph{Proof ideas} Our proof for stable algorithms uses the overlap gap property, previously employed in the context of binary perceptron~\cite{gamarnikKizildagPerkinsXu2022,gamarnikKizildagPerkinsXu2023,liSchrammZhou2024}. For stable algorithms, we rule out $m$ feasible signings on correlated instances with pairwise overlaps in a narrow interval for an appropriate $m\in\mathbb{N}$—dubbed \emph{forbidden structure}; we control small ball probabilities of form $\|A_i\|_{\rm op}\le \kappa\sqrt{N},\forall i$ via passing to the Frobenius norm and Gaussian density bounds. We then interpolate from a common Gaussian instance to many independent instances, use stability of the algorithm and apply a multi-color Ramsey-theoretic argument to extract an $m$-tuple contradicting the \emph{forbidden structure}. For online algorithms, we condition on the random seed and the first $N-q$ input matrices, so that the signed prefix is fixed. Each fixed signing of the remaining $q$ matrices gives a Gaussian vector independent of the conditioned data. A uniform shifted small-ball estimate and a union bound over the $2^q$ suffix signings yield an exponentially small success probability. Taking $q$ of order $\kappa^2N$ proves the density lower bound of order $\kappa^{-2}$.

{\color{addedblue}
\section{The Offline Algorithm: Guarantee and Stability}\label{sec:offline-proofs}
In this section, we prove~\Cref{thm:stable-upper}. To that end, we record the following result.
\begin{proposition}\label{prop:Euclidean}
There exists universal constants $\alpha_0,C>0$ such that the following holds. Let $G\in\R^{D\times N}$ has independent standard normal entries, where $N,D\to\infty$ while $N/D\to \alpha\ge \alpha_0$. Denoting by $\widehat{\sigma}$ the output of~\Cref{alg:recenter} with input $G$, we have as $d\to\infty$
\begin{equation}\label{eq:euclidean-upper}
\Pp\!\left[\norm{G\widehat\sigma}_2\le\frac{CD}{\sqrt{\log \alpha}}\right]\longrightarrow1.
\end{equation}
\end{proposition}
Here, $C>0$ is a constant independent of $\alpha$. In what follows, we establish~\Cref{prop:Euclidean} and combine it with rotational symmetry to prove~\Cref{thm:stable-upper}.
\subsection{Proof of Proposition~\ref{prop:Euclidean}}
Throughout, $G=[g_1\ \cdots\ g_N]$ is a $D\times N$ matrix whose columns are iid random vectors $g_i\sim \cN(0,I_D)$. Set $\alpha_d=N/D$ define
\begin{equation}\label{eq:cap-parameters}
 t=\frac12\sqrt{\log \alpha_d},\qquad p=\Pp(Z\ge t),\qquad b=\frac p8,
 \qquad Z\sim\cN(0,1),
\end{equation}
where we suppress the dependence of $t,p$ and $b$ on the dimension for clarity. We work along a sequence where $\alpha_d = N/D$ converges to a fixed limit $2\tau$, larger than a universal constant.

We study uniform Gaussian geometry around the initial signing $\xi\in\{-1,1\}^N$.
\paragraph{Useful columns along every direction}
\begin{lemma}[Uniform Gaussian caps]\label{lem:caps}
There exists a universal $\alpha_0$ such that for all fixed $\alpha\ge \alpha_0$, if $w_1,\dots,w_N\sim \cN(0,I_D)$ are iid random vectors with $N/D\to \alpha$, then
\begin{equation}\label{eq:cap-event}
 \inf_{u\in S^{D-1}}\#\{i:\ip u{w_i}\ge t\}\ge\frac{pN}{2}
\end{equation}
with probability $1-\exp(-c_\alpha D)$ for all sufficiently large $D$.
\end{lemma}
\begin{proof}
For $t \ge 1$, Mills' inequality for Gaussian tail (see, e.g.,~\cite{vershynin2018}) yields for a universal $c>0$:
\begin{equation}\label{eq:mills}
 p\ge\frac c{1+t}\exp(-t^2/2)
 \ge\frac{c\,\alpha_d^{-1/8}}{\sqrt{\log \alpha_d}}.
\end{equation}
Next, for $u\in \mathbb{S}^{D-1}$, define the halfspace $H_u:=\{x\in\R^D:\ip{u}{x}\ge t\}$ and consider the class $\mathcal{H}_t:=\{H_u:u\in\mathbb{S}^{D-1}\}$ of all such halfspaces. It is known that the Vapnik-Chervonenkis (VC) dimension of $\mathcal{H}_t$ is at most $D+1$~\cite[Theorem~9.3]{shalev2014understanding}. For any $H\in\mathcal{H}_t$, define the empirical measure 
\[
P_N(H):=\frac1N \sum_{i=1}^N \ind\{w_i\in H\} = \frac1N \sum_{i=1}^N \ind\{\ip{u}{w_i}\ge t\}
\]
and let $P(H):=\mathbb{P}[\cN(0,1)\ge t]=p$. Then, VC inequality (see, e.g.,~\cite[Theorem~2.41]{wasserman2006all} as well as the original paper by Vapnik-Chervonenkis~\cite{VapnikChervonenkis1971}) along with Sauer's lemma~\cite[Lemma~6.10]{shalev2014understanding} (see also the original papers by Sauer~\cite{sauer1972density} and Shelah~\cite{shelah1972combinatorial}) implies that for every $a>0$,
\begin{equation*}
\mathbb{P}\left(\sup_{H\in\mathcal{H}_t}\bigl|P_N(H)-P(H)\bigr|>a\right)\le 8\left(\frac{2eN}{D+1}\right)^{D+1}\exp\left(-\frac{Na^2}{32}\right).
\end{equation*}
Taking $a=p/2$, we have that
\begin{align}
    \frac1D\log\left(8\left(\frac{2eN}{D+1}\right)^{D+1}\exp\left(-\frac{Na^2}{32}\right)\right)&= \frac{\log 8}{D} + \frac{D+1}{D}\log\frac{2eN}{D+1} - \frac{Na^2}{32D} \nonumber\\
    &= \log(2e\alpha_d)-\frac{p^2}{128}\alpha_d + o_D(1)\label{eq:log-of-exponent}.
\end{align}
By \eqref{eq:mills}, $\alpha_dp^2\ge c \alpha_d^{3/4}/\log \alpha_d$. Consequently, for sufficiently large $\alpha_d$, it holds that $\log(2e\alpha_d)  - \frac{p^2}{128}\alpha_d<0$. Thus, as $D\to\infty$, the exponent in~\eqref{eq:log-of-exponent} is bounded above by a $-c_\alpha<0$. Consequently, with probability at least $1-\exp(-c_\alpha D)$, we have that $P_N(H)\ge P(H)-\frac{p}{2} =\frac{p}{2}$ for every $H\in\mathcal{H}_t$, implying that
\[
\inf_{u\in\mathbb{S}^{D-1}}|i\in[N]:\ip{u}{w_i}\ge t|\ge \frac{pN}{2},
\]
thereby establishing~\eqref{eq:cap-event}.
\end{proof}
We highlight that the constant $c_\alpha$ arising in the exponent above depends on the aspect ratio $\alpha_d$ tending to a fixed constant.
\begin{lemma}[Robustness under a Hamming perturbation]\label{lem:robust-caps}
Let $\xi$ be the independent initial signing and $w_i=\xi_i g_i$. On the event \eqref{eq:cap-event}, every signing $\sigma$ with $\dist(\sigma,\xi)\le bN$ ($b=p/8$) satisfies
\begin{equation}\label{eq:robust-cap-event}
 \inf_{u\in S^{D-1}}\#\{i:\ip u{\sigma_i g_i}\ge t\}\ge\frac{3pN}{8}.
\end{equation}
\end{lemma}
\begin{proof}
Fix $u\in\mathbb{S}^{D-1}$. Changing one sign can remove at most one index from the set of columns with projection at least $t$. Thus, starting with $pN/2$ such columns and changing at most $pN/8$ leaves at least $3pN/8$ indices at the end. Since the event~\eqref{eq:cap-event} holds for every such $u$, the conclusion is simultaneous in $u$ and $\sigma$.
\end{proof}
Note that conditional on any realization of the initial signing $\xi\sim {\rm Unif}\{\pm 1\}^N$, the columns $\xi_i g_i$ are still independent standard Gaussian vectors. On the `cap event'~\eqref{eq:cap-event}, Lemma~\ref{lem:robust-caps} is deterministic—without assuming that the signs $\sigma$ are independent of $G$.
\paragraph{Canceling the residual with a 
small fractional budget}
For $v_1,\ldots,v_N\in\R^D$ and $q\ge0$, let
\begin{equation}\label{eq:zonotope}
 K_q(v)=\left\{\sum_i a_i v_i:0\le a_i\le1,\ \sum_i a_i\le q\right\}.
\end{equation}
In words, $K_q(v)$ is the set of all vectors one can `produce' by using $q$ units of fractional coefficients. 
\begin{lemma}[Ball containment]\label{lem:zonotope}
Suppose that for every $u\in S^{D-1}$, there exists at least $\lceil q\rceil$ indices $i\in[N]$ with $\ip{u}{v_i}\ge t$. Then
\begin{equation}\label{eq:zonotope-inclusion}
 qt B_2^D\subseteq K_q(v).
\end{equation}
\end{lemma}
In what follows, we employ Lemma~\ref{lem:zonotope} to construct a `fractional' combination of $v_i$ to obtain a small residual. To see this, for a current signing $\sigma\in\{\pm 1\}^N$ with $s=\sum_i \sigma_i g_i = G\sigma$ and $R=\|s\|_2$, we have by~\eqref{eq:Psi} that $z_i=\sigma_i(1-2a_i)$ with $Gz = s-2\sum_{i\le N}a_iv_i$. Assuming Lemma~\ref{lem:zonotope} is valid, one can therefore find $a_i\in[0,1]$ with $\sum_{i\le N}a_i \le q$ for which $\sum_{i\le N}a_iv_i = \frac{s}{2}$, yielding $Gz=0$.
\begin{proof}
For $q=0$, both the claimed point and the available coefficient vector are zero, hence we are done. Assume $q>0$. Note that the set $C_q:=\{a\in\R^N:0\le a_i\le 1,\sum_i a_i\le q\}$ is closed and bounded, thus compact. Further, it is also convex. Since $K_q(v)$ is the image of $C_q$ under the linear map $(a_1,\dots,a_N)\mapsto \sum_{i\le N}a_i v_i$, it is also compact and convex.

Next, fix a $u\in\mathbb{S}^{D-1}$ and let $q=\ell+\theta$, where $\ell=\lfloor q\rfloor$ and $0\le\theta<1$. Among $\lceil q\rceil$ indices $i$ with $\ip{u}{v_i}\ge t$, choose $\ell$ of them and assign $a_i=1$. If $\theta>0$, choose one more such index and set $a_i:=\theta$. In particular, $\sum_{i\le N}a_i = \ell+\theta=q$ and $a_i\in[0,1]$. As such, $x_u:=\sum_i a_iv_i \in K_q$, so $\ip{u}{x_u} = \sum_i a_i\ip{u}{v_i}\ge t\sum_i a_i = qt$. So, $h_{K_q}(u):=\sup_{x\in K_q}\ip{u}{x}\ge qt$ for every $u\in\mathbb{S}^{D-1}$.

We now prove $qtB_2^D\subseteq K_q$. Suppose this is not the case and take a $y\in qtB_2^D$ ($\|y\|_2\le qt$) such that $y\notin K_q$. Since $K_q$ is compact, $\exists x^*\in K_q$ such that
\[
\|y-x^*\|_2 = \min_{x\in K_q}\|y-x\|_2>0.
\]
Let $u^*:=\frac{y-x^*}{\|y-x^*\|_2}$. For any $x\in K_q$, convexity implies $x^*+\lambda(x-x^*)\in K_q$ for all $\lambda\in[0,1]$, thus
\[
\|y-(x^*+\lambda(x-x^*)\|_2\ge \|y-x^*\|_2 \Rightarrow -2\lambda \ip{y-x^*}{x-x^*}+\lambda^2\|x-x^*\|_2^2\ge 0.
\]
Dividing the last display by $\lambda$ and sending $\lambda\to 0$ yields $\ip{u^*}{x}\le \ip{u^*}{x^*}$. On the one hand, we have $h_{K_q}(u^*) = \sup_{x\in K_q}\ip{u^*}{x} = \ip{u^*}{x^*}\ge qt$ from above. So, $\ip{u^*}{y} = \ip{u^*}{x^*} + \ip{u^*}{y-x^*} = h_{K_q}(u^*)+\|y-x^*\|_2>h_{K_q}(u^*)\ge qt$. At the same time, $\ip{u^*}{y}\le \|u^*\|\cdot \|y\|\le qt$ by Cauchy-Schwarz, yielding a contradiction. Thus, $qtB_2^D\subseteq K_q(v)$.
\end{proof}
It is important that the cap condition holds in \emph{all} directions. A large projection in only the current residual direction would not be enough to conclude exact vector cancellation.

\paragraph{Gaussian norm bounds}
\begin{lemma}[Gaussian matrix and column norms]\label{lem:covariance}
Let $G=[g_1\cdots g_N]\in\R^{D\times N}$ have iid $\cN(0,1)$ entries.
Fix $\alpha\ge2$ and suppose $N/D\to\alpha$. There are universal constants
$C_0,c>0$ and a constant $C_\alpha>0$ such that, for all sufficiently
large $D$, the event
\begin{equation}\label{eq:gaussian-norm-event}
 \mathcal E_{\rm norm}:=
 \left\{\|G\|_{\rm op}^2\le C_\alpha D,\quad
        \max_{i\le N}\|g_i\|_2^2\le C_0D\right\}
\end{equation}
has probability at least $1-e^{-cD}$.
\end{lemma}
\begin{proof}
The Gaussian operator-norm bound~\cite[Corollary~5.35]{vershynin-nonasymp}
gives, for every $u>0$,
\[
 \Pp\!\left[\|G\|_{\rm op}>\sqrt D+\sqrt N+u\right]
 \le e^{-u^2/2}.
\]
Also, Gaussian concentration and $\E\|g_i\|_2\le\sqrt D$ give
$\Pp[\|g_i\|_2>\sqrt D+u]\le e^{-u^2/2}$.
Taking $u=\sqrt D$ in both bounds and a union bound over the $N$ columns,
we obtain, with probability at least $1-(N+1)e^{-D/2}$,
\[
 \|G\|_{\rm op}^2\le(2\sqrt D+\sqrt N)^2,
 \qquad \max_{i\le N}\|g_i\|_2^2\le4D.
\]
For sufficiently large $D$, $N\le2\alpha D$, so we may take
$C_0=4$ and $C_\alpha=(2+\sqrt{2\alpha})^2$; also
$(N+1)e^{-D/2}\le e^{-D/4}$.
\end{proof}
\subsection*{Proof of the recentering guarantee}
\paragraph{The fractional comparison point}
Fix a current signing $\sigma$ and write $v_i=\sigma_i g_i$, $s=G\sigma = \sum_i v_i$, $R=\|s\|_2$ and $q=\frac{R}{2t}$.
\begin{lemma}[Zero-residual comparison and objective control]\label{lem:comparison}
Suppose that for every $u\in\mathbb{S}^{D-1}$, at least $\lceil q\rceil$ of indices $i\in[N]$ satisfies $\ip{u}{v_i}\ge t$. Then, there exists a $z^\circ\in[-1,1]^N$ such that  $Gz^\circ =0$ and $\Psi_{G,\sigma}(z^\circ) \le \frac32 q$. Consequently, for $z\in \argmin_{z'\in [-1,1]^N}\Psi_{G,\sigma}(z')$, it follows that
\begin{equation}\label{eq:optimizer-bounds}
 \norm{Gz}_2^2\le3Dq,\qquad \sum_i p_i(z;\sigma)\le\frac32q.
\end{equation}
\end{lemma}
\begin{proof}
Observe that $\|s/2\|_2 = \frac{R}{2} = qt$, thus $\frac{s}{2}\in qtB_2^D$. Thus, by Lemma~\ref{lem:zonotope}, there exists $a_i\in[0,1]$ with $\sum_{i\le N}a_i\le q$ for which $\sum_i a_iv_i = \frac{s}{2}$. Set $z_i^\circ=\sigma_i(1-2a_i)$. Then, it is easy to verify that $Gz^\circ=s-2\sum_i a_i v_i=0$ and $p_i(z^\circ;\sigma)=a_i$, so
\[
\Psi_{G,\sigma}(z^\circ)\le \sum_{i\le N}(a_i+a_i^2/2)\le \frac32\sum_i a_i\le \frac32 q,
\]
using $a_i\le 1$. Namely, the optimal cost of the quadratic program~\eqref{eq:objective} is at most $\frac32q$. Since $z\in\argmin_{z'\in[-1,1]^N}\Psi_{G,\sigma}(z')$, it follows that $\Psi_{G,\sigma}(z)\le \frac32q$ from which we have $\|Gz\|_2^2\le 3Dq$ and $\sum_i p_i(z;\sigma)\le \frac32 q$ follows directly.
\end{proof}
Importantly, the quadratic program~\eqref{eq:objective} trades residual size against flip cost; the minimizer itself need not have a zero residual. The zero-residual point was used solely as a certificate to upper bound the optimal cost of the quadratic program.
\paragraph{Controlling rounding error}
Let $\sigma'=\Round(z;U)$ be the next signing and define
\begin{equation}\label{eq:rounding-error}
 e_i=\sigma'_i-z_i,\qquad X:=\sum_i e_i g_i,
 \qquad G\sigma'=Gz+X.
\end{equation}
In what follows, we condition on $G$ as well as the complete past - including $\sigma$ and $z$ and proceed to control the rounding error. Observe from the rounding procedure (Algorithm~\ref{alg:recenter}, Step 5) that
$\mathbb{P}[\sigma_i'\ne \sigma_i|G,\sigma,z] = \frac{1-\sigma_i z_i}{2}$ and $\mathbb{E}[\sigma_i'|G,\sigma,z] = z_i$. Conditional on $(G,\sigma,z)$, the variables $e_i$ are independent and centered. Moreover, for $p_i=(1-\sigma_i z_i)/2$,
\begin{equation}\label{eq:rounding-moments}
 \nu_i=\E[e_i^2\mid G,\sigma,z]=1-z_i^2=4p_i(1-p_i)\le4p_i,
 \qquad \E[e_i^4\mid G,\sigma,z]\le4\nu_i\le16p_i,
\end{equation}
The fourth-moment inequality follows from $|e_i|\le2$. We also have
$0\le\nu_i=1-z_i^2\le1$.

\begin{lemma}[Adaptive rounding error and number of changes]\label{lem:rounding}
Fix an aspect ratio $\alpha\ge2$ and a finite constant $B_\alpha$, and suppose $N/D\to \alpha$. For all sufficiently large $D$, on the event $\mathcal E_{\rm norm}$ of Lemma~\ref{lem:covariance}, the following holds uniformly over all $\sigma\in\{\pm 1\}^N$ and every $z\in[-1,1]^N$ satisfying
\[
 \sum_i p_i\le\frac32q\le N,\qquad 0\le q\le B_\alpha D.
\]
Conditional on $G,\sigma,z$, 
\begin{align}
 \Pp\!\left(\norm X_2>C\sqrt{Dq}+\frac{CD}{t} \ \middle|\ G,\sigma,z
\right)&\le\frac{C_\alpha}{D},\label{eq:rounding-norm}\\
 \Pp\!\left(\dist(\sigma',\sigma)>3q+\frac D{t^2} \ \middle|\ G,\sigma,z
\right)
 &\le\exp\left[-c\left(q+\frac D{t^2}\right)\right].\label{eq:rounding-flips}
\end{align}
Here, $C,c>0$ are universal, while $C_\alpha$ may depend on the (fixed) constants $\alpha$ and $B_\alpha$.
\end{lemma}
\begin{proof}
Throughout the rounding analysis, we condition on $G$ as well as everything generated before this round, i.e., $\sigma$ and $z$. Let $\Pp_0,\E_0$ and $\Var_0$ respectively denote the probability, expectation and variance over fresh rounding, conditional on $G,\sigma,z$. 

Let $C_z:=\E_0[XX^T]$. Conditional independence and centering give
\begin{equation}\label{eq:C-z}
 C_z=\sum_i\nu_i g_ig_i^T
     =G\operatorname{diag}(\nu_i)G^\top\preceq GG^\top,
\end{equation}
since $0\le\nu_i\le1$. On $\mathcal E_{\rm norm}$,
\begin{equation}\label{eq:trace-Cz}
 \E_0[\|X\|_2^2]=\Tr(C_z)
 \le4\sum_i p_i\|g_i\|_2^2
 \le6C_0Dq\le CDq,
\end{equation}
where $C$ is universal, and
\begin{equation}\label{eq:Czop}
 \|C_z\|_{\rm op}\le\|GG^\top\|_{\rm op}
 =\|G\|_{\rm op}^2\le C_\alpha D.
\end{equation}
These inequalities hold deterministically for every $z\in[-1,1]^N$
on the same event, so they remain valid for an adaptively chosen $z$.
To establish~\eqref{eq:rounding-norm}, we now turn to controlling $\Var_0(\|X\|_2^2)$. Observe that for
\[
\|X\|_2^2 = \sum_i e_i^2 \|g_i\|_2^2 + \sum_{i\ne j}e_ie_j c_{ij},\qquad c_{ij}:=\ip{g_i}{g_j}.
\]
Along with $\E_0[\|X\|_2^2] =\sum_i \nu_i \|g_i\|_2^2$ per~\eqref{eq:C-z}, we thus obtain
\[
\|X\|_2^2 - \E_0[\|X\|_2^2] = \underbrace{\sum_i (e_i^2-\nu_i)\|g_i\|_2^2}_{:=\Sigma_1} + \underbrace{\sum_{i\ne j}e_ie_jc_{ij}}_{:=\Sigma_2}.
\]
Consequently, $\Var_0(\|X\|_2^2) = \E_0[(\Sigma_1+\Sigma_2)^2] = \E_0[\Sigma_1^2]+2\E_0[\Sigma_1\Sigma_2]+\E_0[\Sigma_2^2]$. For $\E_0[\Sigma_1^2]$, note that conditional on $(G,\sigma,z)$, the random variables $(e_i^2-\nu_i)_{i=1}^N$ are independent and centered. As such, $\E_0[(e_i^2-\nu_i)(e_j^2-\nu_j)]=0$ for $i\ne j$, and thus
\begin{equation}\label{eq:A-square}
    \E_0[\Sigma_1^2]=\sum_i \E_0[(e_i^2-\nu_i)^2]\|g_i\|_2^4 = \sum_i \left(\E_0[e_i^4]-\nu_i^2\right)\|g_i\|_2^4.
\end{equation}
As for $\E_0[\Sigma_2^2]$, notice that if $(i,j)$ and $(i_0,j_0)$ are two pairs with $i\ne j$, $i_0\ne j_0$ and $|\{i,j\}\cap \{i_0,j_0\}|\le 1$, then we have $\E_0[e_ie_je_{i_0}e_{j_0}]=0$. Thus, 
\begin{equation}\label{eq:B-square}
\E_0[\Sigma_2^2] = 4\sum_{i<j}\E_0[e_i^2e_j^2]c_{ij}^2 = 4\sum_{i<j}\nu_i\nu_j\ip{g_i}{g_j}^2,
\end{equation}
using the conditional independence of $e_i^2$ and $e_j^2$ given $(G,\sigma,z)$. Lastly, $\E_0[\Sigma_1\Sigma_2]=0$ since $\E_0[(e_i^2-\nu_i)\|g_i\|_2^2 e_ke_\ell c_{k\ell}]=0$ (using the fact $\E_0[e_k]=\E_0[e_\ell]=0$ and conditional independence).

Now, note from~\eqref{eq:C-z} that
\begin{equation}\label{eq:trace-of-Cz}
    \Tr(C_z^2) = \sum_{i,j}\Tr(\nu_i\nu_j g_ig_i^T g_jg_j^T) = \sum_{i,j}\nu_i\nu_j\ip{g_i}{g_j}^2 = \sum_i \nu_i^2 \|g_i\|_2^4 + 2\sum_{i<j}\nu_i\nu_j\ip{g_i}{g_j}^2.
\end{equation}
The fact $\Var_0(\|X\|_2^2) = \E_0[\Sigma_1^2]+\E_0[\Sigma_2^2]$, equations~\eqref{eq:A-square} and~\eqref{eq:B-square}, along with~\eqref{eq:trace-of-Cz} imply
\begin{align}
    \Var_0(\|X\|_2^2) &=2\Tr(C_z^2)+\sum_i \left(\E_0[e_i^4]-3\nu_i^2\right)\|g_i\|_2^4.\label{eq:var-of-Xnorm}
\end{align}
Since $C_z$ is positive semidefinite,~\eqref{eq:trace-Cz} and~\eqref{eq:Czop} imply
\[
 \Tr(C_z^2)\le\|C_z\|_{\rm op}\Tr(C_z)\le C_\alpha D^2q.
\]
The remaining fourth-moment contribution is bounded by
\[
 \sum_i\E_0[e_i^4]\|g_i\|_2^4
 \le16(C_0D)^2\sum_i p_i
 \le24C_0^2D^2q,
\]
using~\eqref{eq:rounding-moments}. Discarding the nonpositive term
$-3\sum_i\nu_i^2\|g_i\|_2^4$ in~\eqref{eq:var-of-Xnorm} therefore gives
\begin{equation}\label{eq:rounding-variance}
 \Var_0(\|X\|_2^2)\le C_\alpha D^2q
 \le C_{\alpha,B_\alpha}D^3.
\end{equation}
Combining this with~\eqref{eq:trace-Cz}, Chebyshev's inequality yields
\[
 \Pp_0\!\left[\|X\|_2^2>CDq+D^2/t^2\right]
 \le\frac{C_{\alpha,B_\alpha}t^4}{D}.
\]
For fixed $\alpha$, $t=\frac12\sqrt{\log(N/D)}$ stays bounded.
Absorbing $t^4$ and the dependence on $B_\alpha$ into the failure
constant $C_\alpha$, then taking square roots, proves~\eqref{eq:rounding-norm}.
The constant multiplying $\sqrt{Dq}+D/t$ is universal.

As for~\eqref{eq:rounding-flips}, note that $B_i:=\ind\{\sigma_i'\ne \sigma_i\}$ are  conditionally independent Bernoulli random variables with mean $p_i$ and $\sum_i \E_0[B_i]:=v\le 3q/2$. Thus, by Bernstein's inequality~\cite[Theorem~2.8.4]{vershynin2018}:
\[
\Pp_0\left[d_H(\sigma',\sigma)>3q+\frac{D}{t^2}\right] \le  \Pp_0\left[\sum_i B_i >v+a\right]\le \exp\left(-\frac{a^2}{2(v+a/3)}\right)
\]
for $a=3q/2+D/t^2$
using $v\le 3q/2$. Routine algebra then yields~\eqref{eq:rounding-flips}.
\end{proof}
\paragraph{One-step contraction and the global induction}
\begin{proposition}[One recenter-and-round step]\label{prop:one-step}
Suppose that the cap event (Lemma~\ref{lem:caps}) and the norm event $\mathcal E_{\rm norm}$ (Lemma~\ref{lem:covariance}) hold, and the current signing $\sigma$ satisfies
\begin{equation}\label{eq:step-assumptions}
 \dist(\sigma,\xi)\le bN,\qquad
 \left\lceil\frac R{2t}\right\rceil\le\frac{3pN}{8}.
\end{equation}
Then, conditionally on the past,
\begin{align}
 \norm{G\sigma'}_2&\le C\sqrt{\frac{DR}{t}}+\frac{CD}{t},\label{eq:one-step-radius}\\
 \dist(\sigma',\sigma)&\le\frac{3R}{2t}+\frac D{t^2}.\label{eq:one-step-hamming}
\end{align}
with probability at least $1-C_\alpha/D-e^{-c_\alpha D}$, provided $R\le 2B_\alpha Dt$.
\end{proposition}
\begin{proof}
Note that since $d_H(\sigma,\xi)\le bN$ with $b=p/8$, by combining Lemmas~\ref{lem:caps} and~\ref{lem:robust-caps} we obtain for $v_i:=\sigma_i g_i$ that for every $u\in\mathbb{S}^{D-1}$, there exists at least $3pN/8$ indices $i\in[N]$ for which $\ip{u}{\sigma_ig_i}\ge t$. Setting $q=\frac{R}{2t}$, we have $\lceil q\rceil \le \frac{3p}{8}N$ and $q\le B_\alpha D = O(D)$. Applying Lemma~\ref{lem:comparison}, we find that $z:=\argmin_{z'\in[-1,1]^N}\Psi_{G,\sigma}(z')$ satisfies $\|Gz\|_2 \le \sqrt{3Dq}$, $\sum_i p_i\le \frac{3q}{2}$. Now, $G\sigma' = Gz+X$ yields that
\[
\|G\sigma'\|_2 \le \|Gz\|_2 + \|X\|_2 \le \sqrt{3Dq} + C\sqrt{Dq}+\frac{CD}{t}\le \bar{C}\sqrt{Dq} + \frac{\bar{C}D}{t}
\]
for some $\bar{C}>C>0$, using Lemma~\ref{lem:rounding}. Plugging $q=\frac{R}{2t}$ and adjusting the constants yield~\eqref{eq:one-step-radius}. Likewise,~\eqref{eq:one-step-hamming} is a direct consequence of~\eqref{eq:rounding-flips}. By a union bound over the two failures in Lemma~\ref{lem:rounding}, the conditional probability of these conclusions is at least $1-C_\alpha/D-e^{-c_\alpha D}$.
\end{proof}

\begin{lemma}[Scalar contraction]\label{lem:scalar}
For every universal $C>0$, there exists a universal $A\in[1,\infty)$ satisfying the following. If $(y_k)_{k=0}^\infty$ is a non-negative sequence satisfying
\begin{equation}\label{eq:scalar}
 y_{k+1}\le C(\sqrt{y_k}+1),
\end{equation}
then $y_k\le A$ for all $k\ge \lceil \log_2(\max\{y_0,1\})\rceil$. Further, $y_k\le \max\{y_0,A\}$ for all $k$.
\end{lemma}
\begin{proof}
Suppose $A=\max\{1,4C,16C^2\}$. Then, if $y_k\ge A\ge 16C^2$, we have $y_{k+1}\le C(\sqrt{y_k}+1) =Cy_k/\sqrt{y_k}+C\le y_k/4 + C$. Using $A\ge 4C$, we also have $C\le y_k/4$, yielding that if $y_k\ge A$, then $y_{k+1}\le y_k/2$. 

Moreover, if $y_k\le A$, then $y_{k+1}\le C(\sqrt{A}+1)\le \frac{A}{2}<A$, so $y_{k+1}\le A$. Consequently, induction gives $y_k\le \max\{A,2^{-k}y_0\}\le \max\{A,y_0\}$ for every $k\ge0$. Lastly, for $k\ge \lceil \log_2 (y_0\vee 1)\rceil$, then we have $2^{-k}y_0\le A$, so $y_k\le A$, establishing Lemma~\ref{lem:scalar}.
\end{proof}
Equipped with these results, we are ready to complete the proof of~\Cref{prop:Euclidean}.
\begin{proof}[Proof of Proposition~\ref{prop:Euclidean}]
Let $\mathcal{E}_{\rm geo}:=\mathcal{E}_{\rm cap}\cap\mathcal E_{\rm norm}$ be the intersection of the events in Lemmas~\ref{lem:caps} and~\ref{lem:covariance}. Conditional on the initial signing $\xi$, we have $G\xi\sim\cN(0,NI_D)$, so with high probability,
\begin{equation}\label{eq:initial-radius}
 R_0:=\|G\xi\|_2\le2\sqrt{ND}=2D\sqrt{\alpha_d}.
\end{equation}
Consequently, the event
\[
 \mathcal{E}:=\mathcal{E}_{\rm geo}\cap\{R_0\le2D\sqrt{\alpha_d}\}
\]
has probability $1-o_D(1)$. Write $R_k:=\|G\sigma^{(k)}\|_2$, $H_k:=\dist(\sigma^{(k)},\xi)$, and $K:=K_d=\lceil\log_2\alpha_d\rceil+1$. Since $\alpha_d\to\alpha$ is fixed, there is a fixed upper bound $K_*$ on $K_d$ for all sufficiently large $D$.

We first choose deterministic bounds that will hold throughout the iteration. Set
\[
 \delta_d:=\frac{9}{2t\sqrt{\alpha_d}}+\frac{1}{\alpha_d t^2},
\]
and let $C$ be the universal constant in Proposition~\ref{prop:one-step}. By increasing the universal threshold $\alpha_0$ in the theorem, we may assume that, for all sufficiently large $D$,
\begin{align*}
 C\left(\frac{\sqrt3\,\alpha_d^{1/4}}{\sqrt t}+\frac1t\right)
 &\le3\sqrt{\alpha_d},\\
 \frac{3}{2pt\sqrt{\alpha_d}}&\le\frac3{16},
 \qquad K\delta_d\le\frac p8.
\end{align*}
Indeed, the first two inequalities follow from $t=\frac12\sqrt{\log\alpha_d}$ and $pt\ge c\alpha_d^{-1/8}$, by~\eqref{eq:mills}. For the third, viewed as a function of the aspect ratio,
\[
 \frac{K\delta_d}{p}
 =O\!\left(\frac{\log\alpha_d}{\alpha_d^{3/8}}
       +\frac{\sqrt{\log\alpha_d}}{\alpha_d^{7/8}}\right),
\]
which tends to zero as the aspect ratio grows. We take $\alpha_0$ strictly larger than a threshold for these inequalities, so that they also hold when $\alpha_d$ approaches its limit from below. Since $p$ converges to a positive constant at every fixed $\alpha$, we also have $1/(pN)\le3/16$ for all sufficiently large $D$. Thus, whenever $R\le3D\sqrt{\alpha_d}$,
\[
 \frac{\lceil R/(2t)\rceil}{pN}
 \le\frac{R}{2tpN}+\frac1{pN}
 \le\frac{3}{2pt\sqrt{\alpha_d}}+\frac1{pN}
 \le\frac38.
\]
Finally, fix a finite constant $B_\alpha$ such that $3\sqrt{\alpha_d}/(2t)\le B_\alpha$ for all sufficiently large $D$.

Conditional on $G$ and $\xi$ satisfying $\mathcal E$, we now induct simultaneously on the radius and the Hamming distance. Stop at the first round in which either bound in Proposition~\ref{prop:one-step} fails. Until such a failure, we claim that for every $0\le k\le K$,
\begin{equation}\label{eq:induction-envelope}
 \begin{split}
 R_k&\le3D\sqrt{\alpha_d},\qquad
 H_k\le kN\delta_d\le bN,\\
 \left\lceil\frac{R_k}{2t}\right\rceil&\le\frac{3pN}{8},
 \qquad b=\frac p8.
 \end{split}
\end{equation}
At $k=0$, the radius bound follows from~\eqref{eq:initial-radius}, $H_0=0$, and the ceiling bound follows from the preceding deterministic estimate. Suppose~\eqref{eq:induction-envelope} holds at some $k<K$. These bounds, together with $R_k/(2t)\le B_\alpha D$, verify all hypotheses of Proposition~\ref{prop:one-step}. On its good rounding event,
\[
 R_{k+1}
 \le C\sqrt{\frac{DR_k}{t}}+\frac{CD}{t}
 \le CD\left(\frac{\sqrt3\,\alpha_d^{1/4}}{\sqrt t}+\frac1t\right)
 \le3D\sqrt{\alpha_d}.
\]
Also, writing $\Delta_k:=\dist(\sigma^{(k+1)},\sigma^{(k)})$, we have
\[
 \frac{\Delta_k}{N}
 \le\frac{3R_k}{2tN}+\frac{D}{Nt^2}
 \le\frac{9}{2t\sqrt{\alpha_d}}+\frac1{\alpha_d t^2}
 =\delta_d.
\]
Hence $H_{k+1}\le H_k+\Delta_k\le(k+1)N\delta_d\le bN$. The ceiling bound at time $k+1$ follows from the new radius bound. This proves all parts of~\eqref{eq:induction-envelope} simultaneously.

At any round before the first failure, the conditional failure probability is at most $C_\alpha/D+e^{-c_\alpha D}$. Summing over the possible first failure times therefore gives
\[
 \Pp(\text{a failure occurs})
 \le\Pp(\mathcal E^c)
   +K_*\left(\frac{C_\alpha}{D}+e^{-c_\alpha D}\right)
 =o_D(1).
\]
In particular, with high probability the radius recurrence holds at every round.

On this event, set $y_k:=tR_k/D$. Then
\[
 y_{k+1}\le C(\sqrt{y_k}+1),\qquad
 y_0\le2t\sqrt{\alpha_d}<\alpha_d,
\]
where the last inequality holds after a further universal increase of $\alpha_0$. Lemma~\ref{lem:scalar} and the choice of $K_d$ imply that $y_{K_d}\le A$ for a universal constant $A$. Since $\alpha_d\to\alpha$, it follows that, with high probability,
\[
 \|G\widehat\sigma\|_2=R_{K_d}
 \le\frac{AD}{t}
 =\frac{2AD}{\sqrt{\log\alpha_d}}
 \le\frac{A'D}{\sqrt{\log\alpha}},
\]
for a universal constant $A'>0$. This proves~\eqref{eq:euclidean-upper}.
\end{proof}
\subsection{Proof of Theorem~\ref{thm:stable-upper}}
We now leverage Proposition~\ref{prop:Euclidean} to establish Theorem~\ref{thm:stable-upper}. The following lemma is instrumental to extend from Euclidean radius to a spectral norm.
\begin{lemma}[Output residual has uniform direction]\label{lem:uniform-direction}
Let $\widehat{\sigma} = \A(G,\omega)$ be the output of Algorithm~\ref{alg:recenter}. Set $X=G\widehat{\sigma}$. Then, $X$ has a rotationally invariant distribution in $\R^D$ and that $\Pp[X=0]=0$. Thus, $X=RU$ where $R=\|X\|_2$, $U$ is uniform on $\mathbb{S}^{D-1}$, and $R$ and $U$ are independent.
\end{lemma}
\begin{proof}
Note that conditional on $G$ and $\xi$, $\Psi_{G,\sigma}(z)$ has Hessian $\frac1D G^TG + I_N/4$, thus it is strongly convex, implying that the problem $\min_{z\in[-1,1]^N}\Psi_{G,\sigma}(z)$ has a unique minimizer. Now, for any $Q\in O(D):=\{Q\in\R^{D\times D}:Q^TQ=QQ^T =I_D\}$ and $z\in\R^N$, we have $\|QGz\|_2^2 = \|Gz\|_2^2$, implying that $z^*(QG,\sigma):=\argmin_{z\in[-1,1]^N}\Psi_{QG,\sigma}(z) = \argmin_{z\in[-1,1]^N}\Psi_{G,\sigma}(z):=z^*(G,\sigma)$. As such, with the same randomness $\xi$, running $\A$ on $QG$ and $G$ yields the same trajectory, thus $\A(QG,\omega)=\A(G,\omega)$ for every $Q\in O(D)$.

Now, define the map $X:=X(G,\omega) = G\A(G,\omega)$. Then,
\begin{equation}\label{eq:X_QG}
    X(QG,\omega) = (QG)\A(QG,\omega) = QG\A(G,\omega) = QX(G,\omega).
\end{equation}
At the same time, $QG\distr G$ for any $Q\in O(D)$ since $G$ has iid $\cN(0,I_D)$ columns. Due to independence of $G$ and $\omega$, it follows that $(QG,\omega)\distr (G,\omega)$, as such $X(QG,\omega)\distr X(G,\omega)$. Since $X(QG,\omega) = QX(G,\omega)$ as shown in~\eqref{eq:X_QG}, we obtain that $X(G,\omega) \distr QX(G,\omega)$ and since $Q\in O(D)$ is arbitrary, it follows that $X$ is rotationally invariant.

As for $\Pp[X=0]=0$, note that for any fixed $\sigma\in\{-1,1\}^N$, $\Pp[G\sigma = 0] =0$, thus a union bound implies $\Pp[X=0] \le \Pp[\exists \sigma\in\{-1,1\}^N : G\sigma=0]=0$. The conclusion now follows from Lemma~\ref{lemma:uniform-direction}.
\end{proof}
Having established~\Cref{lem:uniform-direction}, we are ready to complete the proof of~\Cref{thm:stable-upper}.
\begin{proof}[Proof of Theorem~\ref{thm:stable-upper}]
Let $h\sim \cN(0,I_D)$ and $W=\sqrt{\frac2d}\mat(h)$. Then, $W\sim \GOE(d)$ and $U\distr h/\|h\|_2$. Classical GOE edge theorem gives $\norm W_{\op}\convp2$ \cite{anderson2010introduction}. Moreover, $\|W\|_F^2/d = 2\|h\|_2^2/d^2 \convp 1$ since $\|h\|_2^2/D\convp 1$ (by the WLLN) and $2D/d^2\to 1$. Lastly, since $\mat(U) \distr \frac{W}{\|W\|_F}$, it follows that
\[
\sqrt{d}\cdot \|\mat(U)\|_{\rm op}\distr \frac{\sqrt d\,\norm W_{\op}}{\norm W_{\mathrm F}}\convp2.
\]
We are now ready to complete the proof. Note that $X=G\widehat{\sigma}=RU$ and $\frac{1}{\sqrt{N}}\sum_{i=1}^N \widehat{\sigma}_i A_i = \sqrt{\frac2{dN}}\mat(X) = \sqrt{\frac2{dN}}R\mat(U)$. As such,
\[
\frac{1}{\sqrt{N}}\left\|\sum_{i=1}^N \widehat{\sigma}_i A_i\right\|_{\rm op} = \left(\frac{R}{D}\right)\left(\sqrt{d}\,\bigl\|\mat(U)\bigr\|_{\rm op}\right)\left(\frac{D\sqrt{2}}{d\sqrt{N}}\right).
\]
Note that the second factor converges to 2 in probability, while the third factor converges (deterministically) to $\frac{1}{\sqrt{2\tau}}$. As for the first term,~\Cref{prop:Euclidean} yields that whp as $d\to\infty$, 
\[
\frac{R}{D}\le \frac{C}{\sqrt{\log (2\tau)}},
\]
which for large fixed $\tau$ is at most $C'/\sqrt{\log(\tau+2)}$ by adjusting the universal constant, establishing~\eqref{eq:spectral-upper}.

Lastly, let $L=\log(1/\kappa)$. If $\tau\ge C_+/(\kappa^2L)$, then for all sufficiently small $\kappa$ one has $\log(2+\tau)\ge L$. Hence the upper bound in \eqref{eq:spectral-upper} is at most $(C/\sqrt{C_+})\kappa$; choosing $C_+$ sufficiently large makes this at most $\kappa$.
\end{proof}

\subsection{Noise stability of recenter-and-round}\label{sec:stability}
The recentering analysis used Gaussian geometry, but the continuity of the quadratic program is deterministic. The only discontinuity is the final coordinatewise rounding, which has a canonical coupling.

\begin{lemma}[Sensitivity of the box-constrained minimizer]\label{lem:sensitivity}
For two inputs $G,\widetilde G\in\R^{D\times N}$ and two current signings $\sigma,\widetilde\sigma\in\{-1,1\}^N$, let $z=z(G,\sigma)\in[-1,1]^N$ and $\widetilde z=z(\widetilde G,\widetilde\sigma)\in[-1,1]^N$. Then
\begin{align}
 \norm{z-\widetilde z}_2
 &\le3\norm{\sigma-\widetilde\sigma}_2
 +\frac4D\norm{(G^\top G-\widetilde G^\top\widetilde G)\widetilde z}_2\label{eq:sensitivity1}\\
 &\le3\norm{\sigma-\widetilde\sigma}_2
 +\frac4D(\norm G_{\op}+\norm{\widetilde G}_{\op})
    \norm{G-\widetilde G}_{\op}\sqrt N.\label{eq:sensitivity2}
\end{align}
\end{lemma}
\begin{proof}
Note that the gradient of $\Psi_{G,\sigma}$ obeys $\nabla \Psi_{G,\sigma}(z) = H_G z - \frac34\sigma$, where $H_G=G^\top G/D+I_N/4$. From the optimality of $z$, it follows that $\ip{\nabla \Psi_{G,\sigma}(z)}{w-z}\ge 0$ for every feasible $w\in[-1,1]^N$. As such, 
\[
 \ip{H_Gz-3\sigma/4}{\widetilde z-z}\ge0,\qquad
 \ip{H_{\widetilde G}\widetilde z-3\widetilde\sigma/4}{z-\widetilde z}\ge0.
\]
Setting $\Delta:=z-\widetilde z$,  adding the displays above and applying triangle inequality yields
\[
 \ip\Delta{H_G\Delta}
 \le\frac1D\left|\ip{(G^\top G-\widetilde G^\top\widetilde G)\widetilde z}{\Delta}\right|
 +\frac34\left|\ip{\sigma-\widetilde\sigma}{\Delta}\right|.
\]
Since $H_G\succeq I_N/4$, it follows that $\ip{\Delta}{H_G\Delta}\ge \frac14\|\Delta\|_2^2$. The claim is now immediate if $\Delta=0$, otherwise using Cauchy-Schwarz we obtain~\eqref{eq:sensitivity1}. Finally,
\begin{align*}
    \norm{G^TG - \widetilde{G}^T\widetilde{G}}_{\rm op} &= \norm{G^T(G-\widetilde{G}) + (G-\widetilde{G})^T\widetilde{G}}_{\rm op} \\
    &\le(\norm G_{\op}+\|\widetilde G\|_{\op})\norm{G-\widetilde G}_{\op},
\end{align*}
and the box constraint implies $\norm{\widetilde z}_2\le\sqrt N$. This proves \eqref{eq:sensitivity2}.
\end{proof}

\begin{lemma}[Common-uniform rounding coupling]\label{lem:coupling}
For $z,\widetilde z\in[-1,1]^N$, use the same independent uniforms to define $\sigma'=\Round(z;U)$ and $\widetilde\sigma'=\Round(\widetilde z;U)$. Conditional on these fractional points, the mismatch indicators are independent and
\begin{equation}\label{eq:coupling-mean}
 \E\dist(\sigma',\widetilde\sigma')=\frac12\norm{z-\widetilde z}_1
 \le\frac{\sqrt N}{2}\norm{z-\widetilde z}_2.
\end{equation}
For every $\gamma>0$,
\begin{equation}\label{eq:coupling-tail}
 \Pp\!\left(\frac1N\dist(\sigma',\widetilde\sigma')>
 \frac{\norm{z-\widetilde z}_2}{2\sqrt N}+\gamma\ \middle|\ z,\widetilde z\right)
 \le e^{-2\gamma^2N}.
\end{equation}
\end{lemma}
\begin{proof}
Coordinate $i$ disagrees exactly when $U_i$ falls between the two thresholds $(1+z_i)/2$ and $(1+\widetilde z_i)/2$; this interval has length $|z_i-\widetilde z_i|/2$. Thus, $\Pp\bigl[\sigma_i'\ne \widetilde{\sigma}_i'\mid z,\widetilde{z}\bigr] = \frac12|z_i-\widetilde{z}_i|$. As such, $\E[d_H(\sigma',\widetilde{\sigma}')\mid z,\widetilde{z}] =\frac12\|z-\widetilde{z}\|_1\le \frac{\sqrt{N}}{2}\|z-\widetilde{z}\|_2$ by Cauchy-Schwarz, proving~\eqref{eq:coupling-mean}.

As for the~\eqref{eq:coupling-tail}, we have by the Hoeffding's inequality that for any $\gamma>0$,
\[
\Pp\Bigl[d_H(\sigma',\widetilde{\sigma}')-\E[d_H(\sigma',\widetilde{\sigma}')\mid z,\widetilde{z}] >N\gamma \big\lvert z,\widetilde{z}\Bigr]\le e^{-2\gamma^2 N}.
\]
From here,~\eqref{eq:coupling-tail} follows from the bound~\eqref{eq:coupling-mean}.
\end{proof}

\begin{proof}[Proof of Theorem~\ref{thm:algorithm-stability}]
Fix the limiting ratio $\alpha$ and $\eps>0$. Let $K_*$ be a fixed upper bound on $K_d$ for all sufficiently large dimensions. Couple the two runs with the same internal seed—i.e., same initial signing $\xi$ and uniform variables at every round.  Define
\[
x_k=\frac1N\dist(\sigma^{(k)},\widetilde\sigma^{(k)}),\qquad x_0=0.
\]
Recalling the estimate $\Pp[\|A\|_{\op}>\sqrt{D}+\sqrt{N}+u]\le \exp(-u^2/2)$ valid for all $u>0$ and all $A\in\R^{D\times N}$ with iid $\cN(0,1)$ entries (\cite{vershynin-nonasymp}), we have by taking $u=\sqrt{D}$ that
\begin{equation}\label{eq:noise-op-event}
 \norm G_{\op}+\norm{G^\rho}_{\op}\le C_\alpha\sqrt D,\qquad
 \norm{G-G^\rho}_{\op}\le C_\alpha\sqrt{1-\rho}\sqrt D,
\end{equation}
where we appealed to the fact $G-G^\rho$ has independent $\cN(0,2(1-\rho))$ entries and $N/D\to \alpha$, where $\alpha$ is fixed in $D$.
On \eqref{eq:noise-op-event}, Lemma~\ref{lem:sensitivity} and
$\norm{\sigma^{(k)}-\widetilde\sigma^{(k)}}_2=2\sqrt{Nx_k}$ give
\[
 \frac{\norm{z^{(k)}-\widetilde z^{(k)}}_2}{\sqrt N}
 \le6\sqrt{x_k}+C_\alpha\sqrt{1-\rho}.
\]
Conditional on the two inputs and the common past, the next uniforms are fresh. Applying Lemma~\ref{lem:coupling} at each round, we obtain that except with probability at most $K_*e^{-2\gamma^2N}$,
\begin{equation}\label{eq:stability-recursion}
 x_{k+1}\le3\sqrt{x_k}+a,\qquad a=C_\alpha\sqrt{1-\rho}+\gamma.
\end{equation}
Define $F_0(a)=0$ and $F_{j+1}(a)=3\sqrt{F_j(a)}+a$. It is easy to see that $x_k\le F_k(a)$ for all $k$ and that for each $j\le K_*$, $F_j(\cdot)$ is continuous at zero with $F_j(0)=0$. Choose $a_0>0$ so that $\max_{j\le K_*}F_j(a_0)<\eps$, then choose $0<\gamma<a_0/2$ and $\rho_0<1$ with $C_\alpha\sqrt{1-\rho_0}<a_0/2$. From this, for every $\rho\ge \rho_0$,~\eqref{eq:stability-recursion} gives $x_{K_d}\le F_K(a)\le F_K(a_0)<\eps$ outside an event whose probability tends to zero uniformly in $\rho$, establishing Theorem~\ref{thm:algorithm-stability}.
\end{proof}

\section{\newedit{The Online Algorithm: Performance Guarantee}}\label{sec:online-proofs}
In this section, we prove Theorem~\ref{thm:greedy-performance}. We begin by introducing useful notation.
\subsection{Performance of the Frobenius-greedy algorithm}
\paragraph{Preliminaries} Recall that $D = \frac{d(d+1)}{2}$ and as $N/d^2\to \tau$, we have $N/D\to 2\tau$. Moreover, recall the operator $\svec:\Sym_d\to \R^D$ from~\eqref{eq:svec}: $ \svec(M) = \bigl((M_{ii})_{\newedit{i}=1}^d,(\sqrt{2}M_{ij})_{1\le i<j\le d}\bigr)\in\R^D$.
Observe that $\svec$ is an isometry for the Frobenius norm, i.e., $\|\svec(M)\|_2 = \|M\|_F$ and $\ip{M}{N}_{\rm F} = {\rm trace}(M^TN) = \ip{\svec(M)}{\svec(N)}$, and for $A\sim \GOE(d)$, we have $g:=\sqrt{d/2}\svec(A)\sim \cN(0,I_D)$. 

Next, define
\begin{equation}\label{eq:g_i-x_i}
    g_i=\sqrt{\frac{d}{2}}\svec(A_i),\quad x_i = \sqrt{\frac{d}{2}}\svec(S_i),\quad\text{and}\quad R_i = \|x_i\|_2.
\end{equation}
\newedit{Throughout this proof, take ${\rm sign}(0)=-1$, so that the update agrees with the specified tie rule. Extend the iid input sequence beyond $N$ when considering a fixed time interval $[0,T]$; this does not change the first $N$ iterates.}
The iterates of Algorithm~\ref{alg:Frobenius-greedy} \newedit{imply} $\sigma_i = -{\rm sign}\ip{x_{i-1}}{g_i}$ for $i\ge 1$. Thus, we have the rule
\begin{equation}\label{eq:recurse}
    x_i = x_{i-1}-{\rm sign}\ip{x_{i-1}}{g_i}g_i,\quad \newedit{i\ge 1},\,x_0:=0.
\end{equation}
\paragraph{Recursion, martingale and drift decomposition} Using~\eqref{eq:recurse}, we obtain $\|x_i\|_2^2 = \|x_{i-1}\|_2^2 + \|g_i\|_2^2 - 2|\ip{x_{i-1}}{g_i}|$. With this, we arrive at
\begin{equation}\label{eq:Z_i}
    R_i^2 = R_{i-1}^2 - 2R_{i-1}|Z_i| + \|g_i\|_2^2,\quad\text{where}\quad \newedit{Z_i = \langle u_{i-1},g_i\rangle,\quad u_{i-1}=x_{i-1}/R_{i-1}\text{ when }R_{i-1}>0.}
\end{equation}
\newedit{When $R_{i-1}=0$, take $u_{i-1}$ to be a fixed deterministic unit vector. The radial identity still holds because its final cross term vanishes. Moreover, $R_{i-1}>0$ almost surely for every $i\ge2$: a zero residual would require one of the finitely many signed Gaussian sums to vanish, each a probability-zero event. Conditional} on the $\sigma$-field $\mathcal{F}_{i-1}:=\sigma(g_1,\dots,g_{i-1})$, $Z_i\sim \cN(0,1)$. Starting with $Y_0:=0$, define
\begin{equation}\label{eq:Y-k,eta-k,b}
    Y_k = \frac{R_k^2}{D^2},\quad \eta_k = \bigl(\|g_k\|_2^2 - D\bigr) -2R_{k-1}\bigl(|Z_k|-\sqrt{2/\pi}\bigr).
\end{equation}
A useful observation is that $\mathbb{E}[\eta_k \mid \mathcal{F}_{k-1}] = 0$. With this, we have
\[
Y_k-Y_{k-1} = \frac{1}{D}\bigl(1-\newedit{\sqrt{8/\pi}}\sqrt{Y_{k-1}}\bigr) + \frac{\eta_k}{D^2}. 
\]
Using $Y_0=0$, a telescoping argument then yields
\begin{equation}\label{eq:drift-plus-martingale}
    Y_m = \underbrace{\frac1D \sum_{k=1}^m b(Y_{k-1})}_{\text{(drift term)}} + \underbrace{\frac{1}{D^2}\sum_{k=1}^m \eta_k}_{:=M_m\text{ (martingale term)}}
\end{equation}
where $b:[0,\infty)\to \R$ is defined by $b(x) = 1-\sqrt{8/\pi}\sqrt{x}$. Our arguments will reveal that the drift term converges to the solution of a certain differential equation—dubbed \emph{fluid limit}—while the martingale term vanishes as $D\to\infty$.
\begin{lemma}[Martingale Convergence]\label{lemma:martingale}
    Let $M_m = \frac{1}{D^2}\sum_{k=1}^m \eta_k$. Then, $(M_m)_{m\ge 1}$ is a martingale sequence adapted to the filtration $\mathcal{F}_m:=\sigma(g_1,\dots,g_m)$ and for every fixed $T<\infty$, as $D\to\infty$,
    \[
    \mathbb{E}\left[\max_{m\le TD}|M_m|^2\right]\to 0.
    \]
\end{lemma}
\begin{proof}[Proof of Lemma~\ref{lemma:martingale}]
    Clearly $M_m = \frac{\eta_m}{D^2}+M_{m-1}$, so $\mathbb{E}[M_m\mid \mathcal{F}_{m-1}] = M_{m-1} + \frac{1}{D^2}\mathbb{E}[\eta_m \mid \mathcal{F}_{m-1}] = M_{m-1}$ using tower property and the fact $M_{m-1}$ is $\mathcal{F}_{m-1}$-measurable. We now argue square integrability. Note that~\eqref{eq:Z_i} implies that $R_k^2 - R_{k-1}^2 =-2R_{k-1}|Z_k|+\|g_k\|_2^2$, thus $R_k^2-R_{k-1}^2\le \|g_k\|^2$. By a telescoping argument, we obtain
    \begin{equation}\label{eq:R_k-2nd_mom}
            \mathbb{E}[R_k^2] \le \sum_{j\le k}\mathbb{E}\|g_j\|_2^2 \le kD,
    \end{equation}
    using the fact $\|g_j\|^2$ is a $\chi^2$ random variable with parameter $D$. Turning to $\eta_k$ and using the inequality $(a+b)^2\le 2(a^2+b^2)$ with $a=\|g_k\|_2^2 - D$ and $b = 2R_{k-1}(|Z_k|-\sqrt{2/\pi})$ we have 
    \begin{align}
        \mathbb{E}\eta_k^2&\le 2\mathbb{E}\bigl[\bigl(\|g_k\|_2^2-D\bigr)^2\bigr] + 8\mathbb{E}\bigl[R_{k-1}^2 (|Z_k|-\sqrt{2/\pi})^2\bigr] \nonumber\\
        &=2{\rm Var}(\|g_k\|_2^2) + \newedit{8}\mathbb{E}\bigl[R_{k-1}^2\mathbb{E}\bigl[(|Z_k|-\sqrt{2/\pi})^2\mid \mathcal{F}_{k-1}\bigr]\bigr] \label{eq:R-is-measurable} \\
        &=4D + 8\left(1-\frac2\pi\right)\mathbb{E}[R_{k-1}^2]\label{eq:use-variances} \\
        &\le D\left[4+8(1-\frac2\pi)(k-1)\right]\label{eq:use-R_k-bound}
    \end{align}
    where~\eqref{eq:R-is-measurable} uses the tower property along with the fact $R_{k-1}$ is $\mathcal{F}_{k-1}$-measurable,~\eqref{eq:use-variances} uses the fact variance of $\|g_k\|_2^2$ is $2D$ and conditional on $\mathcal{F}_{k-1}$, 
    \[
    \mathbb{E}\bigl[\bigl(|Z_k|-\sqrt{2/\pi}\bigr)^2\mid \mathcal{F}_{k-1}\bigr] = {\rm Var}(|\mathcal{N}(0,1)|) = 1-\frac2\pi,
    \]
    and~\eqref{eq:use-R_k-bound} uses~\eqref{eq:R_k-2nd_mom}. In particular, all `increments' are square-integrable. Let now $n_D:=\lfloor TD\rfloor$. \newedit{Observe that for $k<\ell$, $\mathbb E[\eta_k\eta_\ell]=0$.} With this, we have
 \begin{align}
     \mathbb{E}|M_{n_D}|^2 &=\frac{1}{D^4}\sum_{k=1}^{n_D}\mathbb{E}[\eta_k^2] + \frac{2}{D^4}\sum_{1\le k<\ell\le n_D}\mathbb{E}[\eta_k\eta_\ell] \nonumber\\
     &\le \frac{1}{D^3}\sum_{k=1}^{n_D}\left(4+8\left(1-\frac2\pi\right)(k-1)\right) \label{eq:use-orthogonality} \\
     &\le \frac{4T}{D^2}+\frac{4(1-2/\pi)T^2}{D}\label{eq:n_D-sum}
 \end{align}
 where~\eqref{eq:use-orthogonality} uses the orthogonality of the increments—i.e., $\mathbb{E}[\eta_k\eta_\ell] = \mathbb{E}[\mathbb{E}[\eta_k\eta_\ell \mid \mathcal{F}_{\ell-1}]] = \mathbb{E}[\eta_k\mathbb{E}[\eta_\ell\mid \mathcal{F}_{\ell-1}]] = 0$ for every $k<\ell$ using the fact $\eta_k$ is measurable with respect to $\mathcal{F}_t$ for every $t\ge k$—and the bound~\eqref{eq:use-R_k-bound}; while~\eqref{eq:n_D-sum} follows from the facts $n_D =\lfloor TD\rfloor \le TD$ and $\sum_{k=1}^{n_D} (k-1) =\frac{n_D(n_D-1)}{2} \le \frac{T^2D^2}{2}$. Applying now Doob's ($L^p$-)maximal inequality
 \[
 \mathbb{E}\left[\max_{0\le m\le n}|M_m|^p\right]\le \left(\frac{p}{p-1}\right)^p \mathbb{E}|M_n|^p
 \]
 valid for all martingale sequences $(M_m)_{m\ge 1}$, $p>1$ and $\mathbb{E}|M_n|^p<\infty$, we deduce 
 \[
 \mathbb{E}\left[\max_{m\le n_D}|M_m|^2\right] \le 4\mathbb{E}|M_{n_D}|^2 = O(\frac1D)\to 0,
 \]
 as $D\to\infty$, for all fixed $T<\infty$. This completes the proof of Lemma~\ref{lemma:martingale}.
\end{proof}
\paragraph{Fluid limit} For the drift term, we will establish a convergence result to the solution of an ODE. Set $\bY_D(t):=Y_{\lfloor Dt\rfloor}$ for $t\ge 0$. In particular, $\bY_D(t) = Y_k$ iff $\frac{k}{D}\le t<\frac{k+1}{D}$ for $k\ge 0$. The key observation is that
\[
\int_{\frac{k-1}{D}}^{\frac{k}{D}} b\bigl(\bY_D(s)\bigr)\;ds = \frac1D b(Y_{k-1}) \Rightarrow \frac1D \sum_{k=1}^m b(Y_{k-1}) = \int_0^{\frac{m}{D}} b\bigl(\bY_D(s)\bigr)\;ds,
\]
where we remind that $b(y) = 1-\sqrt{8/\pi}\sqrt{y}$. Combining this with~\eqref{eq:drift-plus-martingale}, we obtain that for $m=\lfloor Dt \rfloor$ and $0\le t\le T$:
\begin{equation}\label{eq:Y_m-ode}
     \bY_D(t)= Y_m = \int_0^{\frac{m}{D}}b\bigl(\bY_D(s)\bigr)\;ds + M_m = \int_0^t b\bigl(\bY_D(s)\bigr)\;ds + \underbrace{M_m - \int_{\lfloor Dt\rfloor/D}^t b\bigl(\bY_D(s)\bigr)\;ds}_{:=E_D(t)}.
\end{equation}
We next control the error term $E_D(t)$.
\begin{lemma}\label{lemma:error-bound}
    As $D\to\infty$, $\|E_D\|_T:=\sup_{0\le t\le T}|E_D(t)| \xrightarrow{\Pp} 0$.
\end{lemma}
\begin{proof}[Proof of Lemma~\ref{lemma:error-bound}]
    Lemma~\ref{lemma:martingale} already implies that since \newedit{$\max_{m\le TD}|M_m|\xrightarrow{L^2}0$}, we have by Markov's inequality that $\max_{m\le TD}|M_m|\xrightarrow{\Pp}0 $. As for the integral part, note the bound $|b(y)| = |1-\sqrt{8/\pi}\sqrt{y}|\le 1+\sqrt{8/\pi}\sqrt{y}$ and the fact \newedit{$t-\frac1D<\lfloor Dt\rfloor/D\le t$}. We then have
    \begin{equation}\label{eq:max-Y_m}
            \left|\int_{\newedit{\lfloor Dt\rfloor/D}}^t b\bigl(\bY_D(s)\bigr)\;ds\right|\le \frac1D\left(1+\sqrt{8/\pi}\sqrt{\sup_{s\le T}\bY_D(s)}\right)\le \frac1D\left(1+\sqrt{8/\pi}\sqrt{\max_{m\le TD}Y_m}\right).
    \end{equation}
    Observe that $Y_m = \frac{R_m^2}{D^2}$ and $R_m^2\le \sum_{j\le n_D}\|g_j\|_2^2$ \newedit{for all $m\le n_D:=\lfloor TD\rfloor$} (which follows by the bound $R_j^2-R_{j-1}^2\le \|g_j\|_2^2$). Thus, 
    \[
    \max_{m\le TD}Y_m \le \frac{1}{D^2}\sum_{j\le n_D}\|g_j\|_2^2.
    \]
    Using Chebyshev's inequality, we have
    \[
    \mathbb{P}\left[\left|\frac{1}{D^2}\sum_{j\le n_D}\|g_j\|_2^2 -\frac{\lfloor TD\rfloor}{D}\right|>\epsilon\right]\le \frac{n_D}{D^4\epsilon^2}{\rm Var}(\|g_j\|_2^2) = O\left(\frac{1}{D^2}\right).
    \]
    Since $\lfloor TD\rfloor /D \to T$ as $D\to\infty$, this yields that for any $M>T$, we have
    \[
    \mathbb{P}\left[\max_{m\le TD}Y_m<M\right] = 1-o_D(1).
    \]
    Thus,~\eqref{eq:max-Y_m} implies that 
    \[
   \mathbb{P}\left[\sup_{t\in[0,T]}\left|\int_{\newedit{\lfloor Dt\rfloor/D}}^t b\bigl(\bY_D(s)\bigr)\;ds\right|<\frac1D\left(1+\sqrt{8/\pi}\sqrt{M}\right)\right]= 1-o_D(1),
    \]
    so as $D\to\infty$,
    \[
    \sup_{t\in[0,T]}\left|\int_{\newedit{\lfloor Dt\rfloor/D}}^t b\bigl(\bY_D(s)\bigr)\;ds\right|\xrightarrow{\Pp} 0.
    \]
    Lemma~\ref{lemma:error-bound} then follows by the triangle inequality (on $M_m$ and $\sup_{t\in[0,T]}|\int_{\newedit{\lfloor Dt\rfloor/D}}^t b\bigl(\bY_D(s)\bigr)\;ds|$).
\end{proof}
Consequently,
\[
\bY_D(t) = \int_0^t b\bigl(\bY_D(s)\bigr)\;ds + E_D(t),\quad\text{where}\quad \sup_{0\le t\le T}|E_D(t)|\xrightarrow{\Pp} 0.
\]
Now, let $y(t)$ be the solution to the integral equation
\begin{equation}\label{eq:y-def}
    y(t) = \int_0^t b(y(s))\;ds,
\end{equation}
corresponding to the ODE $y'(t) = b(y(t))$ with $y(0)=0$. 
\begin{lemma}[Tracking lemma]\label{lemma:tracking-lemma}
Let $T>0$, and let $b:[0,\infty)\to\mathbb{R}$ be nonincreasing. Suppose that $y,Y:[0,T]\to[0,\infty)$ and
$e:[0,T]\to\mathbb R$ satisfy $e(0)=0$,
\[
b\circ y,\ b\circ Y\in L^1([0,T]),
\qquad
\|e\|_{\infty}:=\sup_{t\in[0,T]}|e(t)|<\infty,
\]
and, for every $t\in[0,T]$,
\[
y(t)=\int_0^t b(y(s))\,ds
\quad\text{and}\quad
Y(t)=\int_0^t b(Y(s))\,ds+e(t).
\]
Then
\[
\sup_{t\in[0,T]}|Y(t)-y(t)|
\le
2\sup_{t\in[0,T]}|e(t)|.
\]
\end{lemma}
Note that $Y(0)=y(0)=0$ in the setup above, and that the continuity of neither $Y$ nor $e$ is assumed.
\begin{proof}[Proof of Lemma~\ref{lemma:tracking-lemma}]
Let $\delta=\|e\|_\infty<\infty$. Define $q(t):=Y(t)-e(t)-y(t)$. In particular, $q(0)=0$, 
    \[
    q(t) = \int_0^t b(Y(s))\;ds - \int_0^t b(y(s))\;ds
    \]
    and thus $q$ is continuous. We prove that $|q(t)|\le \delta$. To this end, suppose $q(t)>\delta$ for some $t$. Let $s=\max\{u\in[0,t]:q(u)=\delta\}$. Since $q$ is continuous with $q(0)=0$, the intermediate value theorem asserts that the set $\{u\in[0,t]:q(u)=\delta\}$ is non-empty; it is also bounded and thus compact. Hence, the maximum exists and $q(u)>\delta$ for $u\in(s,t]$. In this interval, we have that
    \[
    Y(u) - y(u) =q(u)+e(u)>\delta - \|e\|_\infty=0.
    \]
    Since $b$ is non-increasing, it follows that $b(Y(u))\le b(y(u))$ for $u\in (s,t]$. But then,
    \[
    q(t) - q(s) = \int_s^t \bigl[b(Y(u)) - b(y(u))\bigr]\;du\le 0,
    \]
    implying that $\delta<q(t)\le q(s)=\delta$, a contradiction. As for the lower bound, assume $q(t)<-\delta$ for some $t>0$ and let $s<t$ be the last time for which $q(s)=-\delta$. An analogous argument yields that $b(Y(u))\ge b(y(u))$ for $u\in(s,t]$, which then implies $q(t)-q(s) \ge 0$, a contradiction. Thus, $|q(t)|\le \delta=\|e\|_\infty$. Finally, since $Y(t) - y(t) = q(t)+e(t)$, Lemma~\ref{lemma:tracking-lemma} follows by the triangle inequality.
\end{proof}
Combining Lemmas~\ref{lemma:error-bound} and~\ref{lemma:tracking-lemma}, we immediately \newedit{obtain that} for any fixed $T<\infty$,
\begin{equation}\label{eq:fluid-limit}
\sup_{0\le t\le T}\bigl|\bY_D(t) - y(t)\bigr| \xrightarrow{\Pp} 0.
\end{equation}
With this, we now establish:
\begin{proposition}\label{prop:norm-convergence}
    As $D\to\infty$, \[
    \frac{R_N}{D}\xrightarrow{\Pp} \sqrt{y(2\tau)}:=r(2\tau).
    \]
\end{proposition}
\begin{proof}[Proof of Proposition~\ref{prop:norm-convergence}]
    Note that we established that for any fixed $T<\infty$, $\sup_{0\le t\le T}|\bY_D(t) - y(t)|\xrightarrow{\Pp}0$. Let $t_D = \frac{N}{D}\to 2\tau$. Fix a time $T\in(2\tau,\infty)$ and note that $\bY_D(t_D) = Y_{\lfloor Dt_D\rfloor} = Y_N$. Then, 
    \begin{align*}
        \bigl|\bY_D(t_D)-y(2\tau)\bigr|&\le \bigl|\bY_D(\frac{N}{D})-y(\frac{N}{D})\bigr| + \bigl|y(\frac{N}{D}) - y(2\tau)\bigr| \\
        &\le \sup_{s\in [0,T]}\bigl|\bY_D(s)-y(s)\bigr| + \bigl|y(\frac{N}{D}) - y(2\tau)\bigr|.
        \end{align*}
        The first term above converges to zero in probability due to~\eqref{eq:fluid-limit} while the second term converges to $0$ as $D\to\infty$ since $N/D\to 2\tau$ and $y$ (defined in~\eqref{eq:y-def}) is continuous. So, $R_N/D = \sqrt{Y_N}\to \sqrt{y(2\tau)}$ by the continuous mapping theorem.
\end{proof}
Proposition~\ref{prop:norm-convergence} thus gives a control over the norm $R_N$ of $x_N$—equivalently, the Frobenius norm $\|S_N\|_F$. In what follows, we convert this result into a guarantee for $\|S_N\|_{\rm op}$ using rotational symmetry.
\paragraph{Rotational invariance} Recall now the update rule $x_i = x_{i-1}-{\rm sign}\ip{x_{i-1}}{g_i}g_i$ from~\eqref{eq:recurse} and define the map $T(x,g):=x-{\rm sign}\ip{x}{g}g$. We next prove that the distribution of $x_k$ is rotationally invariant for every $k\ge 1$. 

\begin{lemma}\label{lemma:rotational-invariance}
For any orthogonal matrix $Q\in O(D) = \{Q\in\R^{D\times D}:Q^TQ = QQ^T = I_D\}$ and $k\ge 1$, $Qx_k\distr x_k$. 
\end{lemma}
\begin{proof}[Proof of Lemma~\ref{lemma:rotational-invariance}]
    We prove the statement by inducting on $k$, where the base case $k=1$ is clear since $x_1=g_1$ (recall $\sigma_1=1$). Suppose $Qx_{k-1}\distr x_{k-1}$. Note that for any $Q\in O(D)$, we have $T(Qx,Qg) = Qx-{\rm sign}\ip{Qx}{Qg}Qg = Q(x-{\rm sign}\ip{x}{g}g) = QT(x,g)$. Thus, 
    \[
    Qx_k = QT(x_{k-1},g_k) = T(Qx_{k-1},Qg_k).
    \]
Since $Qx_{k-1}\distr x_{k-1}$ by the inductive hypothesis and $Qg_k \distr g_k$ (as $g_k\sim \cN(0,I_D)$), \newedit{and the two vectors remain independent,} it follows immediately that $T(Qx_{k-1},Qg_k)\distr T(x_{k-1},g_k) =x_k$. Hence, $Qx_k \distr x_k$.
\end{proof}
The following lemma is a standard measure-theoretic fact; we include a proof for completeness.
\begin{lemma}\label{lemma:uniform-direction}
    Let $X\in \R^D$ be a rotationally invariant random vector and $X\ne 0$ almost surely. Then, $U=\frac{X}{\|X\|}$ is uniform on $\mathbb{S}^{D-1}:=\{v\in\R^D:\|v\|_2=1\}$ and independent of $R=\|X\|$.
\end{lemma}
\begin{proof}[Proof of Lemma~\ref{lemma:uniform-direction}]
Let $R=\|X\|$ and $U=\frac{X}{\|X\|_2}$. Consider $\phi:X\mapsto (\|X\|,\frac{X}{\|X\|})$. Then, 
 \[
 \phi(QX) = (\|QX\|_2 ,\frac{QX}{\|QX\|_2}) = (\|X\|_2,Q\frac{X}{\|X\|_2})= (R,QU).
 \]
 Since $QX\distr X$, we have $\phi(QX)=(R,QU)\distr \phi(X) = (R,U)$. With this, take any bounded measurable functions $f:\R^D\to\R$ and $h:[0,\infty)\to \R$ and note that $\mathbb{E}[f(QU)h(R)] = \mathbb{E}[f(U)h(R)]$ for every fixed $Q\in O(D)$. Let now $\mu_{\rm Haar}$ be the Haar measure on $O(D)$. Averaging over $Q\sim\mu_{\rm Haar}$,
 \[
 \mathbb{E}[f(U)h(R)] = \int_{O(D)}\mathbb{E}[f(QU)h(R)]\;\mu_{\rm Haar}(dQ) = \mathbb{E}\left[h(R)\int_{O(D)}f(QU)\;\mu_{\rm Haar}(dQ)\right]
 \]
 by Fubini's theorem. Fix any $u\in\mathbb{S}^{D-1}$. Then for $Q\sim \mu_{\rm Haar}$, it holds that $Qu$ is uniform on $\mathbb{S}^{D-1}$, denoted by $\nu_D$. With this, 
 \[
 \int_{O(D)}f(Qu)\;\mu_{\rm Haar}(dQ) = \int_{\mathbb{S}^{D-1}}f(v)\nu_D(dv),
 \]
 where the latter integral is independent of $u$. Thus,
 \[
 \mathbb{E}[f(U)h(R)] = \left[\int_{\mathbb{S}^{D-1}}f(v)\;\nu_D(dv)\right]\mathbb{E}[h(R)],
 \]
 where the first integral is nothing but $\mathbb{E}[f(U)]$, proving $U\perp R$ as $f,h$ are arbitrary. The fact $U\sim {\rm Unif}(\mathbb{S}^{D-1})$ follows by taking $f=\ind_A$ for any Borel set $A$, yielding $\mathbb{P}[U\in A] =\nu_D(A)$.
\end{proof}
\begin{lemma}\label{lemma:B_N-GOE}
    Let $x_N=R_N U_N$ where $R_N = \|x_N\|_2$ and define $B_N = \svec^{-1}(U_N)\in\R^{d\times d}$—i.e., $U_N = {\rm svec}(B_N)$. Then, 
    \[
    B_N \distr \frac{W}{\|W\|_F},\quad\text{where}\quad W\sim \GOE(d).
    \]
\end{lemma}
\begin{proof}[Proof of Lemma~\ref{lemma:B_N-GOE}]
    Note that $x_N$ is rotationally symmetric per Lemma~\ref{lemma:rotational-invariance}. Thus, by Lemma~\ref{lemma:uniform-direction}, we have that $U_N\sim {\rm Unif}(\mathbb{S}^{D-1})$. Now let $G\sim \cN(0,I_D)$. It is a classical fact that $\frac{G}{\|G\|}$ is uniformly distributed on $\mathbb{S}^{D-1}$, so $B_N\distr \svec^{-1}(G/\|G\|_2)$. Observe that $H:=\svec^{-1}(G)$ satisfies $\|H\|_F = \|G\|_2$, $H_{aa}\sim \cN(0,1)$ and $H_{ab} = H_{ba}\newedit{\sim}\cN(0,\frac12)$ for $a<b$, where the collection $((H_{aa})_{a=1}^d,(H_{ab})_{1\le a<b\le d})$ is independent. Hence, $W = \sqrt{\frac2d}H\sim \GOE(d)$, implying that
    \[
    B_N \distr \frac{H}{\|H\|_F} = \frac{W}{\|W\|_F}.
    \]
\end{proof}
We are now ready to complete the proof of Theorem~\ref{thm:greedy-performance}.
\begin{proof}[Proof of Theorem~\ref{thm:greedy-performance}]
    For $\Phi(\tau):=\sqrt{\frac{2}{\tau}}r(2\tau)$, we now prove that as $D\to\infty$,
    \[
    \frac{1}{\sqrt{N}}\|S_N\|_{\rm op}\xrightarrow{\Pp} \Phi(\tau).
    \]
    To this end, note that
    \[
    S_N = \sqrt{\frac2d}R_NB_N, \quad\text{where}\quad R_N = \|x_N\|_2,\,B_N = \svec^{-1}\left(\frac{x_N}{R_N}\right).
    \]
Using this fact, we have
\begin{equation}\label{eq:S_N-op-norm}
    \|S_N\|_{\rm op} = \sqrt{\frac2d}R_N \|B_N\|_{\rm op}\distr \sqrt{\frac2d}R_N \frac{\|W\|_{\rm op}}{\|W\|_F},
\end{equation}
where we invoked Lemma~\ref{lemma:B_N-GOE} as well as the independence of $R_N$ and $B_N$ established in Lemma~\ref{lemma:uniform-direction}. We now study~\eqref{eq:S_N-op-norm}. By the Proposition~\ref{prop:norm-convergence}, we have $\frac{R_N}{D}\xrightarrow{\Pp} r(2\tau)$. Further, it is a well-known result (see, e.g.~\cite{anderson2010introduction}) that for $W\sim \GOE(d)$, $\|W\|_{\rm op}\xrightarrow{\Pp} 2$. Lastly, by the (weak) law of large numbers, $\frac{1}{\sqrt{d}}\|W\|_F \xrightarrow{\Pp} 1$. Combining these facts, 
\[
\frac{1}{\sqrt{N}}\|S_N\|_{\rm op}\distr \sqrt{\frac{2}{Nd}}R_N \frac{\|W\|_{\rm op}}{\|W\|_F} = \frac{D\sqrt{2}}{d\sqrt{N}}\cdot \frac{R_N}{D}\cdot \frac{\|W\|_{\rm op}}{\frac{1}{\sqrt{d}}\|W\|_F}\xrightarrow{\text{d}} \sqrt{\frac2\tau}r(2\tau)=\Phi(\tau)
\]
using Slutsky's lemma as well as the fact that $\frac{D\sqrt{2}}{d\sqrt{N}}\to \frac{1}{\sqrt{2\tau}}$ \newedit{deterministically}. Since $\Phi(\tau)$ is constant for fixed $\tau$, we obtain $\frac{1}{\sqrt{N}}\|S_N\|_{\rm op}\xrightarrow{\Pp} \Phi(\tau)$, completing the proof of Theorem~\ref{thm:greedy-performance}.
\end{proof}
\subsection{Density required for success}
\paragraph{Differential equation and uniqueness} The fact that $r(t)$ in~\eqref{eq:r-implicit} solves the differential equation $\frac{d}{dt}r(t)^2 = 1-\sqrt{8/\pi} r(t)$ can be seen either by direct differentiation or as follows. We have that $2r(t)r'(t) = 1-\sqrt{8/\pi} r(t)$, which yields that
\[
\frac{dt}{dr} = \frac{2r}{1-\sqrt{8/\pi} r}\Rightarrow t = \int_0^r \frac{2u}{1-\sqrt{8/\pi} u}\;du.
\]
Note that $\frac{2u}{1-\sqrt{8/\pi} u} = -\sqrt{\pi/2} + \newedit{\frac{\sqrt{\pi/2}}{1-\sqrt{8/\pi}\,u}}$, implying that
\[
t = \int_0^r \frac{2u}{1-\sqrt{8/\pi} u}\;du = \newedit{-\sqrt{\pi/2}\,r} - \frac{\pi}{4}\log\bigl(1-\sqrt{8/\pi} r\bigr), r<\sqrt{\pi/8}.
\]
\newedit{The right-hand side of~\eqref{eq:r-implicit} is continuous and strictly increasing in $r\in[0,\sqrt{\pi/8})$, from zero to infinity, so it defines $r(t)$ uniquely. Its expansion $t=r^2+O(r^3)$ at zero proves that $(r^2)\prime(0)=1$; the differential equation therefore also holds at the initial point.}
We now argue uniqueness. \newedit{Write $r_1=r$.} Let $r_2(t)\ge 0$ be any other solution to the differential equation~\eqref{eq:diff-eq} with $r_2(0)=0$. Consider $g(t):=(r_1(t)^2-r_2(t)^2)^2$. We have:
\begin{align*}
    g'(t) &=\newedit{2(r_1(t)^2-r_2(t)^2)\bigl((r_1^2)'(t)-(r_2^2)'(t)\bigr)}\\
    &=4\sqrt{2/\pi}(r_1(t)^2-r_2(t)^2)(r_2(t)-r_1(t))\\
    &=-4\newedit{\sqrt{2/\pi}}(r_1(t)-r_2(t))^2(r_1(t)+r_2(t))\le 0,
\end{align*}
where we used, \newedit{for $t>0$,} the fact that $2r_1(t)r_1'(t) = 1-\sqrt{8/\pi} r_1(t)$, so that $r_1(t)r_1'(t) - r_2(t)r_2'(t) = \sqrt{2/\pi}(r_2(t)-r_1(t))$. Since $g(0)=0$ and $g(t)\ge 0$ for all $t\ge 0$, it must in fact hold that $g$ is identically zero, proving uniqueness.
\paragraph{Curve inversion} We now use the condition $\kappa>\Phi(\tau)$ to obtain a  threshold on $\tau$. To that end, notice that the implicit equation~\eqref{eq:r-implicit} yields that 
\[
2\tau  = -\sqrt{\pi/2}\,r(2\tau) -\frac{\pi}{4}\log\bigl(1-\sqrt{8/\pi} r(2\tau)\bigr) = -\frac{\pi a}{4}-\frac{\pi}{4}\log(1-a),
\]
where $a=\sqrt{8/\pi} r(2\tau)$. This implies, $\frac{8\tau}{\pi} = -a-\log(1-a)$. At the same time,
\[
\Phi(\tau)^2 = \frac{2}{\tau}r(2\tau)^2 = \frac{\pi a^2}{4\tau} = \frac{2a^2}{-a-\log(1-a)}=:h(a).
\]
Likewise, let $L(a):=-a-\log(1-a)$. By direct calculation, $L'(a) = \frac{a}{1-a}$ and $L''(a) = \frac{1}{(1-a)^2}$. Moreover,
\[
h'(a) = \frac{4a L(a) - 2L'(a)a^2}{L(a)^2} = \frac{2a(2L(a) - aL'(a))}{L(a)^2}.
\]
Set $q(a): =aL'(a) - 2L(a)$. Then, $q(0)=0$ and $q'(a) = aL''(a) -L'(a) = \frac{a}{(1-a)^2}- \frac{a}{1-a} = \frac{a^2}{(1-a)^2}>0$. Thus, $q(a)>0$ for all $a>0$, which in turn implies that 
\[
h'(a) = \newedit{-}\frac{2aq(a)}{L(a)^2}<0,\quad \text{for}\quad a>0.
\]
Thus, as $a$ increases from zero to one, $h(a)$ continuously decreases from $4$ to $0$. So, for any $\kappa\in(0,2)$, there is indeed a unique $a_\kappa$ for which
\begin{equation}\label{eq:h-a-k}
    \kappa^2 = \frac{2a_\kappa^2}{-a_\kappa - \log(1-a_\kappa)} = h(a_\kappa).
\end{equation}
Define now 
\[
\tau_{\rm FG}(\kappa):=\frac{\pi}{8}\bigl[-a_\kappa - \log(1-a_\kappa)\bigr]= \frac{\pi}{4\kappa^2}a_\kappa^2.
\]
Observe that $\kappa>\Phi(\tau)$ iff $h(a_\kappa) = \kappa^2>\Phi(\tau)^2 = h(a(\tau))$ (recall that $h$ is decreasing), which holds iff $a(\tau)>a_\kappa$, where $a(\tau) = \sqrt{8/\pi} r(2\tau)$. 

As for the last expression, note that $L(a(\tau)) = \frac{8\tau}{\pi}$ and that since $L'(a)>0$, $L$ is strictly increasing, so $a(\tau)>a_\kappa$ iff $L(a(\tau))>L(a_\kappa)$. This means,
\[
\frac{8\tau}{\pi} = L(a(\tau)) >L(a_\kappa) = \frac{8}{\pi}\tau_{\rm FG}(\kappa),
\]
where we used the fact $\tau_{\rm FG}(\kappa) = \frac{\pi}{8}L(a_\kappa)$. \newedit{This implies that the algorithm succeeds whp when $\tau>\tau_{\rm FG}(\kappa)$, and fails whp when $\tau<\tau_{\rm FG}(\kappa)$. The equality case remains open in this analysis.}
\paragraph{$\kappa\to 0$ Asymptotics for $\tau_{\rm FG}(\kappa)$} We lastly prove that $a_\kappa = 1+o_\kappa(1)$ as $\kappa\to 0$. It is clear that $a_\kappa<1$. Next, note the Taylor expansion for $L(a)=-a-\log(1-a)$, yielding
\[
L(a) = \sum_{m\ge 2}\frac{a^m}{m}\le \frac{a^2}{2(1-a)}.
\]
Using~\eqref{eq:h-a-k}, we have
\[
\kappa^2 =\frac{2a_\kappa^2}{L(a_\kappa)}\ge 4(1-a_\kappa).
\]
Thus,
\[
1>a_\kappa\ge 1-\frac{\kappa^2}{4},
\]
yielding the desired asymptotic.
\begin{remark}
    In fact, a much faster convergence guarantee for $a_\kappa$ holds. Let $\delta_\kappa:=1-a_\kappa$. Then,~\eqref{eq:h-a-k} implies that $-1+\delta_\kappa-\log\delta_\kappa  = \frac{2(1-\delta_\kappa)^2}{\kappa^2}$. For $\kappa$ sufficiently small, we have $a_\kappa\ge \frac12$, consequently, $-\log\delta_\kappa\ge \frac{1}{2\kappa^2}$, \newedit{and hence $\delta_\kappa/\kappa^2\to0$. Substituting back gives $-\log\delta_\kappa=2/\kappa^2+1+o_\kappa(1)$, so that}
    \[
    a_\kappa = 1-\exp\left(-\frac{2}{\kappa^2} - 1 +o_\kappa(1)\right),
    \]
    thus the convergence is in fact exponentially fast.
\end{remark}

{\color{addedblue}
\section{Hardness for Stable Algorithms}\label{sec:stable-hardness}
In this section, we establish our main hardness result for stable algorithms (\Cref{thm:stable-lower}). We begin by passing to a more convenient representation.

To this end, given iid inputs $A_1,\dots,A_N\sim \GOE(d)$, let $g_i = \sqrt{d/2}\svec(A_i)\sim \cN(0,I_D)$ and $G=[g_1 \cdots g_N]\in\R^{D\times N}$ with iid $\cN(0,1)$ entries, as in~\eqref{eq:vectorization}. Observe that for any $\sigma\in\{-1,1\}^N$, the operator norm of $\sum_{i=1}^N \sigma_i A_i$ is at most $\kappa\sqrt{N}$, if and only if:
\begin{equation}\label{eq:cal-S-set}
    \sigma \in \mathcal{S}_{d,N}(G;\kappa):=\left\{\sigma\in\{-1,1\}^N : \sqrt{2/d}\,\bigl\|\mat(G\sigma)\bigr\|_{\op}\le \kappa\sqrt{N}\right\},
\end{equation}
where $\mat:\R^D\to\Sym_d$ is the inverse of the operator $\svec$. In what follows, we will work with $G\in\R^{D\times N}$ obtained from the inputs $A_1,\dots,A_N$.
\subsection{The overlap-gap property}
\paragraph{Interpolation}
For $\sigma,\sigma'\in\{-1,+1\}^N$, write
\begin{equation}\label{eq:overlap-Q}
    Q(\sigma,\sigma')=\frac1N\ip\sigma{\sigma'}
 =1-\frac2N\dist(\sigma,\sigma').
\end{equation}
Let $G_0,G_1,\ldots,G_m\in\R^{D\times N}$ be independent standard Gaussian matrices and define
\begin{equation}\label{eq:interpolation}
 G_i(\theta)=\cos\theta\,G_0+\sin\theta\,G_i,
 \qquad 0\le\theta\le\pi/2.
\end{equation}
Each interpolated matrix is marginally standard Gaussian.
\paragraph{Parameters and the geometric statement}
Put $\ellzero=\log(4/e)>0$ and $L_\kappa=\log(1/\kappa)$. Assume
\begin{equation}\label{eq:ogp-density}
 0<\tau\le\frac{\ellzero}{72\kappa^2L_\kappa}.
\end{equation}
For sufficiently small $\kappa$, choose a fixed integer
\begin{equation}\label{eq:ogp-parameters-m}
     m\ge\max\left\{50,\left\lceil\frac{32\tau\log2}{\ellzero}\right\rceil\right\}
\end{equation}
and define parameters $0<\eta<\beta<1$ and $\gamma_k>0$ by
\begin{equation}\label{eq:ogp-parameters-beta-eta-gamma}
\beta=1-8\kappa^2,\quad\eta=\frac{1-\beta}{2m}=\frac{4\kappa^2}{m},
 \quad\gamma_\kappa=\frac12\kappa^2L_\kappa.
\end{equation}
Notice that $m$ is allowed to grow as $\kappa\downarrow 0$, while it remains fixed as $d\to\infty$.
\begin{theorem}[Ensemble overlap-gap property]\label{thm:ogp}
There is a universal $\kappa_0>0$ such that the following holds. Fix any $0<\kappa<\kappa_0$ and $\tau$ satisfying~\eqref{eq:ogp-density}, an $m\in\mathbb{N}$ as in~\eqref{eq:ogp-parameters-m}, $\beta,\eta$ and $\gamma_k$ as in~\eqref{eq:ogp-parameters-beta-eta-gamma}, and an arbitrary finite set $\mathcal I\subset[0,\pi/2]$ satisfying $|\mathcal I|\le\exp(\gamma_\kappa N)$. For iid matrices $G_0,\dots,G_m\in\R^{D\times N}$ with iid standard normal entries and $N/d^2\to \tau$, with probability at least $1-\exp(-cN)$ there are no signings $\sigma^1,\ldots,\sigma^m\in\{-1,1\}^N$ and values $\theta_1,\ldots,\theta_m\in\mathcal I$ satisfying the following simultaneously:
\begin{itemize}
    \item For every $i=1,\dots,m$, $\sigma^i \in \cS_{d,N}(G_i(\theta_i),\kappa)$ for the sets in~\eqref{eq:cal-S-set}.
    \item The overlaps~\eqref{eq:overlap-Q} satisfy $Q(\sigma^i,\sigma^j)\in[\beta-\eta,\beta]$ for every $1\le i<j\le m$.
\end{itemize}
The constant $c>0$ in the exponent is allowed to depend on fixed parameters.
\end{theorem}
\paragraph{Counting overlap-constrained $m$-tuples} We begin by bounding the number of $m$-tuples $\sigma^1,\dots,\sigma^m$ satisfying the overlap constraint $Q(\sigma^i,\sigma^j) \in[\beta-\eta,\beta]$ for every $1\le i<j\le m$. Let
\begin{equation}\label{eq:pov}
 p_{\mathrm{ov}}=\frac{1-\beta+\eta}{2}
 =4\kappa^2\left(1+\frac1{2m}\right)\le4.04\kappa^2,
\end{equation}
where the last inequality holds for all sufficiently small $\kappa$.  Note that if $Q(\sigma^i,\sigma^j)\ge \beta-\eta$ for all $i<j$,  then every $\sigma^i$ is in the Hamming ball of radius $p_{\mathrm{ov}}N$ around $\sigma^1$. With $H(p):=-p\log p-(1-p)\log(1-p)$ denoting the binary entropy, the number of such tuples is at most
\begin{equation}\label{eq:tuple-count}
 2^N\left(\sum_{j=0}^{\lfloor p_{\mathrm{ov}}N\rfloor}\binom Nj\right)^{m-1}
 \le\exp\{N\log2+(m-1)NH(p_{\mathrm{ov}})\}.
\end{equation}
Here, the upper bound follows by standard estimates and holds for $0\le p_{\rm ov}\le 1/2$, which indeed holds for all small enough $\kappa$.
 Discarding pairwise restrictions, this bound is crude; yet it suffices for our argument.
\paragraph{Covariance lower bound}
Towards controlling the probability of the event that $\sigma^i\in \cS_{d,N}(G_i(\theta_i);\kappa)$ for all $i$, we study the corresponding covariance structure.
\begin{lemma}\label{lem:ogp-covariance}
Fix deterministic angles $\theta_1,\dots,\theta_m\in\mathcal{I}$ and deterministic signings $\sigma^1,\dots,\sigma^m\in\{-1,1\}^N$ satisfying $Q(\sigma^i,\sigma^j)\in[\beta-\eta,\beta]$ for all $i<j$. Define
\begin{equation}\label{eq:V_I}
     V_i:=\frac{1}{\sqrt{N}}G_i(\theta_i)\sigma^i\in\R^D.
\end{equation}
Let $Q=(Q_{ij})_{i,j=1}^m$ be the Gram matrix of overlaps, where $Q_{ij}:=Q(\sigma^i,\sigma^j)$, and set
\[
\Gamma:=\mathrm{diag}(\cos\theta_1,\dots,\cos\theta_m).
\]
Then, the concatenated vector $(V_1,\ldots,V_m)\in\R^{mD}$ is centered normal with covariance $B\otimes I_D$, where $B:=\Gamma Q\Gamma + (I_m - \Gamma^2)$. Moreover,
\begin{equation}\label{eq:B-lower}
 B\succeq\frac{1-\beta}{2}I_m,
 \qquad \det(B\otimes I_D)\ge\left(\frac{1-\beta}{2}\right)^{mD}.
\end{equation}
\end{lemma}
\begin{proof}
Note that  $Q$ has diagonal entries 1 and off-diagonal entries in $[\beta-\eta,\beta]$. Thus
\[
 Q=(1-\beta)I_m+\beta\boldsymbol1\boldsymbol1^\top-E,
\]
where $E$ is symmetric, has zero diagonal, and $0\le E_{ij}\le\eta$ for every $i\ne j$. Hence,
\[
\|E\|_{\rm op}\le \max_i \sum_j |E_{ij}|\le m\eta = \frac{1-\beta}{2},
\]
and
\begin{equation}\label{eq:lambda-min-Q}
    \lambda_{\min}(Q)\ge \frac{1-\beta}{2}:=\lambda.
\end{equation}
Writing $\Gamma=\operatorname{diag}(\cos\theta_1,\ldots,\cos\theta_m)$, a calculation involving independent Gaussians gives
\[
 B=\Gamma Q\Gamma+I_m-\Gamma^2.
\]
Lastly, since $\Tr(Q)=m$, we have $\lambda\le \lambda_{\min}(Q)\le 1$. Therefore, for every $x\in\R^m$, 
\begin{align*}
    x^T Bx &= (\Gamma x)^T Q(\Gamma x) +\sum_i (1-\cos^2\theta_i)x_i^2 \\
    &\ge \sum_i \bigl[\lambda \cos^2\theta_i + 1-\cos^2 \theta_i]x_i^2\ge \lambda \|x\|_2^2
\end{align*}
using~\eqref{eq:lambda-min-Q} and the fact $\lambda:=(1-\beta)/2\le 1$. This proves $B\succeq \lambda I_m = \frac{1-\beta}{2}I_m$. Lastly, since $\det(B\otimes I_D) = (\det B)^D$~\eqref{eq:B-lower} follows.
\end{proof}
The lower bound is uniform over all angles $\theta_1,\dots,\theta_m$, including the value $\theta_i=0$.
\paragraph{Small-ball probability and first moment argument}

If $\sigma^i$ is feasible for $G_i(\theta_i)$, we have by~\eqref{eq:cal-S-set} that
\[
\sqrt{\frac2d}\norm{\mat(G_i(\theta_i)\sigma^i)}_{\op}
  \le\kappa\sqrt N.
\]
Using $\norm M_{\F}\le\sqrt d\norm M_{\op}$, the fact $\bigl\|G_i(\theta_i)\sigma_i\bigr\|_2 = \bigl\|\mat(G_i(\theta_i)\sigma_i)\bigr\|_F$ yields that for $V_i$ defined in~\eqref{eq:V_I}
\begin{equation}\label{eq:necessary-euclidean}
  \norm{V_i}_2\le \frac{\kappa d}{\sqrt2}=:\mathsf R_d
\end{equation}
For any $s>0$, the volume formula and the elementary lower bounds on the gamma function give
\begin{equation}\label{eq:ball-volume}
 \operatorname{Vol}(B_2^D(s))
 =\frac{\pi^{D/2}s^D}{\Gamma(D/2+1)}
 \le\left(\frac{2\pi e s^2}{D}\right)^{D/2}.
\end{equation}
Recall also that the maximum density of a centered Gaussian with covariance $B\otimes I_D$ is $(2\pi)^{-mD/2}(\det B)^{-D/2}$. Consequently, \Cref{lem:ogp-covariance} gives
\begin{align}
 \Pp\bigl[\norm{V_i}_2\le\mathsf R_d\text{ for all }i\bigr]
 &\le\left(\frac{2e\mathsf R_d^2}{D(1-\beta)}\right)^{mD/2}\notag\\
 &\le\left(\frac{2e\kappa^2}{1-\beta}\right)^{mD/2}
 =\exp\left(-\frac{mD\ellzero}{2}\right).\label{eq:joint-small-ball}
\end{align}
Here $\mathsf R_d^2/D=\kappa^2d/(d+1)\le\kappa^2$ and $1-\beta=8\kappa^2$.

\begin{proof}[Proof of \Cref{thm:ogp}]
Let $N_{\mathcal I}$ count the configurations $(\sigma^1,\dots,\sigma^m)$ satisfying the conditions of~\Cref{thm:ogp}. Combining the tuple count~\eqref{eq:tuple-count} with the small ball probability~\eqref{eq:joint-small-ball}, a union bound over $|\mathcal{I}|^m$ choices for $(\theta_1,\dots,\theta_m)$ yields that
\begin{equation}\label{eq:ogp-expectation}
 \log\E N_{\mathcal I}
 \le N\log2+m\underbrace{\left[N\big(H(p_{\mathrm{ov}})+\gamma_\kappa\big)-\frac{D\ellzero}{2}\right]}_{:=\Upsilon}.
\end{equation}
For $0<p\le1/2$, $H(p)\le p\log(e/p)$. Thus,~\eqref{eq:pov} implies, uniformly for $m\ge 50$ and small $\kappa$,
\begin{equation}\label{eq:entropy-estimate}
H(p_{\mathrm{ov}})+\gamma_\kappa\le9\kappa^2L_\kappa.
\end{equation}
Here, we used the fact in the regime $\kappa\downarrow 0$, the leading coefficient of $H(4.04\kappa^2)$ relative to $\kappa^2L_\kappa$ is $8.08$, while $\gamma_\kappa$ contributes $0.5$, thus their sum at most 9.

Since $N/d^2\to\tau$ and $D/d^2\to1/2$, we can further control $\Upsilon$ in~\eqref{eq:ogp-expectation} by 
\[
\frac{1}{d^2}\Upsilon\le  9\tau\kappa^2L_\kappa-\frac{\ellzero}{4}+o(1)
 \le\frac{\ellzero}{8}-\frac{\ellzero}{4}+o(1)
 \le-\frac{\ellzero}{9}
\]
for sufficiently large $d$. For $m\ge32\tau\log2/\ellzero$, the term $m\Upsilon$ therefore dominates $N\log 2$, thus $\E N_{\mathcal I}\le e^{-cN}$. Applying Markov's inequality proves the theorem.
\end{proof}
We now convert this geometric property to an algorithmic hardness result.
\paragraph{Independent-input endpoint chaos}
\begin{lemma}[Endpoint chaos property]\label{lem:chaos}
Under the assumptions and parameter choices of~\Cref{thm:ogp}, let $H_1,\ldots,H_m$ be $D\times N$independent standard Gaussian  matrices. With probability at least $1-e^{-c''N}$, there are no feasible signings $\sigma^i\in\cS_{d,N}(H_i;\kappa)$ with
\begin{equation}\label{eq:chaos}
 Q(\sigma^i,\sigma^j)>\beta-\eta,\quad \text{for all}\quad i<j.
\end{equation}
Here, $c''>0$ depends only on the fixed parameters.
\end{lemma}
\begin{proof}
Note that the number of $m$-tuples satisfying~\eqref{eq:chaos} is still bounded by~\eqref{eq:tuple-count}. For each fixed tuple $(\bs^1,\dots,\bs^m)$,  $H_i\sigma^i/\sqrt N$ are independent standard Gaussian vectors in $\R^D$. By \eqref{eq:necessary-euclidean}--\eqref{eq:ball-volume}, we have
\[
\mathbb{P}\bigl[\sigma^i \in \cS_{d,N}(H_i;\kappa),\forall i\bigr]\le (e\kappa^2)^{mD/2}.
\]
Hence the logarithm of the expected number of tuples $(\sigma^1,\dots,\sigma^m)$ is at most
\[
 N\log2+m\left[NH(p_{\mathrm{ov}})+\frac D2\log(e\kappa^2)\right].
\]
For sufficiently small $\kappa$, $\log(e\kappa^2)=1-2L_\kappa\le-L_\kappa$. Using \eqref{eq:entropy-estimate} and the Gaussian density bound, we have
\[
\frac{NH(p_{\mathrm{ov}})+\frac D2\log(e\kappa^2)}{d^2}\le 
 \frac{\ellzero}{8}-\frac{L_\kappa}{4}+o(1)
 \le-\frac{L_\kappa}{8}
\]
for sufficiently small $\kappa$ and large $d$. The choice of $m$ again makes the negative term dominate $N\log2$, after which Markov's inequality establishes~\Cref{lem:chaos}.
\end{proof}
\subsection{Proof of hardness for stable algorithms}
We next provide the details of the interpolation/Ramsey argument, specialized to the present model.
\begin{lemma}[Hamming control of overlap]\label{lem:overlap-continuity}
For any $\sigma_1,\sigma_2,\sigma_1',\sigma_2'\in\{-1,1\}^N$,
\begin{equation}\label{eq:overlap-continuity}
 |Q(\sigma_1,\sigma_2)-Q(\sigma'_1,\sigma'_2)|
 \le\frac2N\big(\dist(\sigma_1,\sigma'_1)+\dist(\sigma_2,\sigma'_2)\big).
\end{equation}
\end{lemma}
\begin{proof}
Changing one coordinate in either argument changes the unnormalized inner product by at most $2$, from which the conclusion follows via triangle inequality.
\end{proof}
\paragraph{The finite interpolation and Ramsey construction}
\begin{proof}[Proof of~\Cref{thm:stable-lower}]
Assume a stable algorithm (in the sense of~\Cref{def:strong}) succeeds with probability $1-o_d(1)$ as $d\to\infty$. Fix $m,\beta,\eta$ as in \eqref{eq:ogp-parameters-m} and~\eqref{eq:ogp-parameters-beta-eta-gamma}. Apply stability with $\delta=\eta/16$, and let $\rho_0<1$ be the resulting threshold. Choose a fixed integer $Q\ge1$ such that
\[
 \rho_Q=\cos\left(\frac{\pi}{2Q}\right)\ge\rho_0,
 \qquad\theta_k=\frac{k\pi}{2Q},\quad0\le k\le Q.
\]
Let $T=R_{Q+1}(m)$ be a multicolor Ramsey number: every coloring of the edges of the complete graph on $T$ vertices with $Q+1$ colors contains a monochromatic $K_m$. It is known that $T$ is finite, and in fact $T\le (Q+1)^{(Q+1)m}$, see~\cite{conlon2015recent} for a discussion (and~\cite{erdos1935combinatorial} for the so-called neighborhood
chasing argument used for proving $T\le (Q+1)^{(Q+1)m}$). Note that all quantities selected thus far are fixed before taking $d\to\infty$.

Generate independent $G_0,G_1,\ldots,G_T$ and one independent common seed $\omega$ and construct the interpolation paths $G_i(\theta_k)$~\eqref{eq:interpolation} for $i\in[T]:=\{1,\dots,T\}$ For each path $i\in[T]$ and grid point $0\le k\le Q$ define
\[
\sigma_i(k):=\A\bigl(G_i(\theta_k),\omega\bigr)\in\{-1,1\}^N
\]
At $k=0$, all inputs and all seeds are identical, so the outputs are identical, $\sigma_i(0)=\sigma_j(0)$ for $1\le i<j\le T$. Next, for a fixed $i$, it holds that
\[
\E\bigl[G_i(\theta_k)_{ab}G_i(\theta_{k+1})_{ab}\bigr] = \cos(\theta_{k+1}-\theta_k)=\rho_Q
\]
for every $1\le a\le D$ and $1\le b\le N$. Since $T$ and $Q$ are fixed independently of $d$, the success and stability assumptions and a union
bound show that, with probability tending to one,
\begin{align}
 \sigma_i(k)&\in\cS_{d,N}(G_i(\theta_k);\kappa)&&\text{for all }i,k,\label{eq:all-feasible}\\
 \dist(\sigma_i(k+1),\sigma_i(k))&\le\delta N&&\text{for all }i,\ 0\le k<Q.\label{eq:all-stable}
\end{align}
Call the event in~\eqref{eq:all-feasible} as $\mathcal{E}_{\rm sat}$ and the event in~\eqref{eq:all-stable} as $\mathcal{E}_{\rm st}$. 
Applying~\Cref{thm:ogp} to each $S\subset [T]$ with $|S|=m$, we obtain that there is no ensemble-OGP configuration on the grid $I_Q = \{\theta_0,\dots,\theta_Q\}$. with high probability, by a union bound over all such subsets and using the fact $\binom{T}{m}$ is constant in $d$. Call this event $\mathcal{E}_{\rm ogp}$. Lastly,  applying~\Cref{lem:chaos} to each such subset at the independent endpoints $k=Q$, with probability tending to one no $m$ endpoint outputs that are all feasible have all
pairwise overlaps greater than $\beta-\eta$; we denote this event by $\mathcal{E}_{\rm chaos}$. 

By a union bound, the event $\mathcal{E}_{\rm sat}\cap \mathcal{E}_{\rm st}\cap \mathcal{E}_{\rm ogp}\cap \mathcal{E}_{\rm chaos}$ holds with probability $1-o_d(1)$. In the remainder of the proof, we work on this high probability event.

\paragraph{Ramsey step} Color each edge $(i,j)$ of the complete graph on $[T]$ as follows. If
\[
Q(\sigma_i(Q),\sigma_j(Q))>\beta-\eta,
\]
assign color $0$. Otherwise, we have $Q(\sigma_i(Q),\sigma_j(Q))\le\beta-\eta$ while $Q(\sigma_i(0),\sigma_j(0))=1$. Using stability, there exists a first $k\in\{1,\dots,Q\}$ for which $Q(\sigma_i(k),\sigma_j(k))\le\beta$; we color $(i,j)$ with $k$. By minimality of $k$, \Cref{lem:overlap-continuity}, and \eqref{eq:all-stable},
\[
 \beta\ge Q(\sigma_i(k),\sigma_j(k))>\beta-4\delta\ge\beta-\eta.
\]
Thus an edge of color $k\ge 1$ has overlap in $[\beta-\eta,\beta]$ at time step $\theta_k$.

Ramsey's theorem yields the existence of a monochromatic clique of size $m$. Note that if the clique is colored $0$, it contradicts chaos property (\Cref{lem:chaos}). A clique of color $k\ge 1$ supplies an $m$-tuple $\sigma^1(\theta_k),\dots,\sigma^m(\theta_k)$, all feasible with respect to a correlated instance, and $Q(\sigma^i(\theta_k),\sigma^j(\theta_k))\in[\beta-\eta,\beta]$ for every $i<j$, contradicting ensemble multi OGP (\cref{thm:ogp}).  This contradiction proves the theorem.
\end{proof}
\subsection{An alternative notion of stability}\label{sec:quantitative-stability}
For comparison with the literature, we record a finite-parameter version of stability. 
\begin{definition}\label{def:alt-stability}
An algorithm is $(\kappa,\rho,p_f,p_{\mathrm{st}},f,L)$-stable if it succeeds with probability at least $1-p_f$ and, for entrywise $\rho$-correlated Gaussian pair $(G,\widetilde{G})$ coupled with a common seed,
\begin{equation}\label{eq:parameterized}
 \Pp\left[\dist(\A(G),\A(\widetilde G))
 \le f+L\norm{G-\widetilde G}_{\mathrm F}\right]\ge 1-p_{\mathrm{st}}.
\end{equation}
\end{definition}
\Cref{def:alt-stability} is the matrix discrepancy analogue
of the finite-parameter stability notions considered
in~\cite{gamarnikKizildagPerkinsXu2022,
gamarnikKizildagPerkinsXu2023,liSchrammZhou2024}.  Let $\A_{D,N}$ denote~\Cref{alg:recenter}, and fix
$\kappa>0$ and an aspect ratio $\alpha=2\tau$ in the
success regime of~\Cref{thm:stable-upper}. By~\Cref{thm:algorithm-stability}, for every $\epsilon>0$
there exists $\rho_0=\rho_0(\epsilon,\alpha)<1$ such that
\[
 p_{\rm st}(d)
 := \sup_{\rho\in[\rho_0,1)} \Pp\left[
 \dist\bigl(\A_{D,N}(G,\omega),
            \A_{D,N}(G^\rho,\omega)\bigr)
 >\epsilon N
 \right]
 =o_d(1).
\]
Since $L\|G-G^\rho\|_F \ge 0$, this also bounds the failure
probability in~\eqref{eq:parameterized} with $f=\epsilon N$,
for every $\rho\in[\rho_0,1)$ and every $L\ge0$. Moreover,~\Cref{thm:stable-upper} yields that the algorithm is successful with probability $1-p_f(d)$, for some $p_f(d)=o_d(1)$. Consequently, $\A_{D,N}$ is
$(\kappa,\rho,p_f(d),p_{\rm st}(d),\epsilon N,L)$-stable
for these parameters.

Our hardness proof also applies under the weaker
finite-parameter assumptions of~\Cref{def:alt-stability},
as recorded in the following corollary.
\begin{corollary}\label{cor:quantitative}
Fix the parameters of~\Cref{thm:ogp} and a constant $\overline{L}>0$. Choose
\[
 Q\ge\max\{1,64\pi\overline L/\eta\},\qquad
 \rho=\cos\left(\frac\pi{2Q}\right),\qquad
 f=\frac{\eta N}{32},\qquad L=\overline L\sqrt{\frac ND}.
\]
There are positive constants $p_f,p_{\mathrm{st}}$, depending only on the fixed parameters, such that for all
sufficiently large $d$ no algorithm is stable with the parameter tuple $(\kappa,\rho,p_f,p_{\mathrm{st}},f,L)$.
\end{corollary}
\begin{proof}[Proof of~\Cref{cor:quantitative}]
Use the same fixed Ramsey number $T=R_{Q+1}(m)$ and the interpolation grid. For consecutive inputs on one interpolation path, $G_i(\theta_{k+1})-G_i(\theta_k)$ has iid centered Gaussian entries with variance $2(1-\rho)$. Since $(1-\rho)\le \pi^2/(8Q^2)$, chi-square concentration shows that 
\[
 \norm{G_i(\theta_{k+1})-G_i(\theta_k)}_{\mathrm F}
 \le\frac\pi Q\sqrt{DN}
\]
with probability tending to one  for every $i\in[T]$ and $0\le k\le Q-1$. Whenever \eqref{eq:parameterized} holds,  the Hamming distance of adjacent outputs is therefore at most
\[
 \frac{\eta N}{32}+\frac{\pi\overline L}{Q}N
 \le\frac{3\eta N}{64}<\frac{\eta N}{16}.
\]
Choose $p_f$ and $p_{\rm st}$ smaller than $[10T(Q+1)]^{-1}$. Taking a union bound for all success and stability events yields a probability bounded away from zero, while the geometric event and chi-square event above has failure probabilities tending to zero. The result follows by applying the same deterministic Ramsey construction on the resulting nonempty event.
\end{proof}
\par}

\section{\newedit{Hardness for Online Algorithms}}\label{sec:online-hardness}
We begin with an auxiliary probability estimate. Let $\mathbb{B}(\boldsymbol{a},R):=\{\boldsymbol{v}:\|\boldsymbol{v}-\boldsymbol{a}\|_2\le R\}$ denote the Euclidean ball of radius $R$, centered at $\boldsymbol{a}\in\R^D$,  where $D:=d(d+1)/2$. 
\begin{lemma}\label{lemma:AuxilProb}
Let $\boldsymbol{Z}\sim \cN(0,I_D)$. For any $\boldsymbol{a}\in\R^{D}$, $R>0$ and $d$ sufficiently large,
   \[
\mathbb{P}\Bigl[\|\boldsymbol{Z} - \boldsymbol{a}\|_2\le R\Bigr] \le \newedit{\left(\frac{eR^2}{D}\right)^{D/2}}.
   \]
\end{lemma}
\begin{proof}[Proof of Lemma~\ref{lemma:AuxilProb}]
    Note that the density of $\boldsymbol{Z}$ is given by
    \[
    f_{\boldsymbol{Z}}(\boldsymbol{z}) = (2\pi)^{-\frac{D}{2}} \exp\left(-\newedit{\frac12}\|\boldsymbol{z}\|_2^2\right) \le (2\pi)^{-\frac{D}{2}}.
    \]
    Thus, 
\begin{equation}\label{eq:Estimate1}
            \mathbb{P}\Bigl[\|\boldsymbol{Z} - \boldsymbol{a}\|_2\le R\Bigr] = \mathbb{P}\bigl[\boldsymbol{Z}\in \mathbb{B}(\boldsymbol{a},R)\bigr]\le (2\pi)^{-\frac{D}{2}}{\rm Vol}(\mathbb{B}(\boldsymbol{a},R)).
        \end{equation}
Since the volume of an Euclidean ball is invariant under translation, we have from standard results (see, e.g.,~\cite{lasserre2001quick}) that
\begin{equation}\label{eq:Estimate2}
    {\rm Vol}(\mathbb{B}(\boldsymbol{a},R)) =R^{D} {\rm Vol}(\mathbb{B}(0,1)) = \frac{(\pi R^2)^{D/2}}{\Gamma\left(\frac{D}{2}+1\right)},
\end{equation}
where $\Gamma(\cdot)$ is the Gamma function. Stirling's formula yields
\begin{equation}\label{eq:Estimate3}
    \Gamma\left(\frac{D}{2}+1\right)  = \sqrt{\pi D}\left(\frac{D}{2e}\right)^{\frac{D}{2}}\Bigl(1+o_d(1)
    \Bigr).
\end{equation}
Combining~\eqref{eq:Estimate1}-\eqref{eq:Estimate3} with the trivial bound $\sqrt{\frac{1}{\pi D}}(1+o_d(1))\le 1$ for $d$ large establishes Lemma~\ref{lemma:AuxilProb}.
\end{proof}
{\color{addedblue}
\begin{proof}[Proof of Theorem~\ref{thm:Online-LowerBd}]
Set $C=1/(10e)$ and $\kappa^*=(8e)^{-1/2}$. Fix $0<\kappa\le\kappa^*$ and $0<\tau\le C/\kappa^2$, and let $N/d^2\to\tau$. Choose
\[
 \Delta=4e\kappa^2,\qquad q=\lfloor\Delta N\rfloor,
 \qquad \Delta_N=q/N.
\]
Then $0<\Delta\le1/2$, $1\le q<N$ for all sufficiently large $d$, and $\Delta_N\to\Delta$.

Fix an arbitrary online algorithm $\A$ with independent random seed $\omega$, and write $\sigma=\A(A_1,\dots,A_N;\omega)$. Condition on
\[
 \cF=\sigma(\omega,A_1,\dots,A_{N-q}).
\]
By online measurability, the signed prefix
\[
 S=\sum_{j=1}^{N-q}\sigma_j A_j
\]
is $\cF$-measurable. The suffix matrices remain independent $\GOE(d)$ matrices, independent of $\cF$. For each fixed $u\in\{-1,1\}^q$, let
\[
 M_u=\sum_{j=1}^q u_j A_{N-q+j}.
\]
Then $M_u\distr\sqrt q\,W$ for $W\sim\GOE(d)$, and $M_u$ is independent of $\cF$. Using the Frobenius isometry~\eqref{eq:svec}, we have
\[
 Z_u:=\sqrt{\frac{d}{2q}}\,\svec(M_u)\sim\cN(0,I_D),
 \qquad D=\frac{d(d+1)}2.
\]
Since $\|B\|_F\le\sqrt d\,\|B\|_{\rm op}$ for every symmetric matrix $B$, Lemma~\ref{lemma:AuxilProb} gives, almost surely,
\begin{align*}
 &\Pp\bigl(\|S+M_u\|_{\rm op}\le\kappa\sqrt N\mid\cF\bigr)\\
 &\quad\le\Pp\left(\left\|Z_u+\sqrt{\frac{d}{2q}}\,\svec(S)\right\|_2
       \le\frac{\kappa d}{\sqrt{2\Delta_N}}\,\middle|\,\cF\right)
 \le\left(\frac{e\kappa^2d}{\Delta_N(d+1)}\right)^{D/2}.
\end{align*}
If the algorithm succeeds, its suffix signs equal one of the $2^q$ vectors $u\in\{-1,1\}^q$. A conditional union bound, followed by expectation, therefore yields
\begin{equation}\label{eq:online-prefix-bound}
 p_d:=\Pp(\A\text{ succeeds})
 \le 2^q\left(\frac{e\kappa^2d}{\Delta_N(d+1)}\right)^{D/2}.
\end{equation}
This bound holds for every online algorithm, regardless of how it chooses the suffix signs.

Taking logarithms in~\eqref{eq:online-prefix-bound} and using $N/d^2\to\tau$ gives
\begin{align*}
 \limsup_{d\to\infty}\frac{\log p_d}{d^2}
 &\le \tau\Delta\log2+\frac14\log\frac{e\kappa^2}{\Delta}\\
 &=\left(4e\tau\kappa^2-\frac12\right)\log2
 \le-\frac{\log2}{10}.
\end{align*}
In particular, for all sufficiently large $d$,
\[
 p_d\le\exp\left(-\frac{\log2}{20}d^2\right)
 \le\exp\left(-\frac{\log2}{40\tau}N\right),
\]
where the last inequality uses $N\le2\tau d^2$. Choosing
$c=(\log2)/(80\tau)>0$ gives $p_d<e^{-cN}$, proving the theorem with the stated strict conclusion.
\end{proof}
}

\section*{Use of AI Tools}
The authors used AI models, in particular GPT-5.6 and GPT-6, with Pro and Astra modes, during the preparation of this manuscript. The algorithmic lower bounds were established by the authors (later simplified by the AI), suggesting the location of the algorithmic phase transition points up to constants. Guided by these thresholds, the authors then engaged in extensive conversations with the AI model from which the matching algorithms, particularly the recenter-and-round algorithm (\Cref{alg:recenter}), and the key algorithmic idea of using rotational symmetry to convert small Euclidean residual to operator norm bounds emerged. The authors take full responsibility for the mathematical content and the correctness of the final text.

\bibliographystyle{amsalpha}
\makeatletter
\let\originalbibitem\bibitem
\def\bibitem{\@ifnextchar[\markedbibitem\markedbibitemplain}
\def\setbibcolor#1{\def\currentbibkey{#1}\def\oldkeya{unchangedentry}\def\oldkeyb{lasserre2001quick}%
 \ifx\currentbibkey\oldkeya\color{black}\else\ifx\currentbibkey\oldkeyb\color{black}\else\color{addedblue}\fi\fi}
\def\markedbibitem[#1]#2{\setbibcolor{#2}\originalbibitem[#1]{#2}}
\def\markedbibitemplain#1{\setbibcolor{#1}\originalbibitem{#1}}
\makeatother
\bibliography{bibliography}
\end{document}